\documentclass[a4paper, 11pt]{article}
\usepackage[utf8]{inputenc}
\usepackage[ngermanb, english]{babel}
\usepackage{amsmath}
\usepackage{amssymb}
\usepackage{amsthm}
\usepackage{tikz-cd}
\usepackage{bm}
\usepackage{commath}
\usepackage{mathrsfs}
\usepackage{mathtools}
\usepackage{slashed}
\usepackage{tensor}
\usepackage{parskip}
\usepackage{epstopdf}
\usepackage{graphicx}
\usepackage[nottoc]{tocbibind}
\usepackage{babelbib}
\usepackage{hyperref}
\usepackage[top=2.5cm, left=2.5cm, right=2.5cm, bottom=2.5cm]{geometry}

\theoremstyle{plain}
\newtheorem{thm}{Theorem}[section]
\newtheorem{lem}[thm]{Lemma}
\newtheorem{prop}[thm]{Proposition}
\newtheorem{col}[thm]{Corollary}

\theoremstyle{definition}
\newtheorem{defn}[thm]{Definition}
\newtheorem{exmp}[thm]{Example}

\theoremstyle{remark}
\newtheorem{rem}[thm]{Remark}

\providecommand{\sectionref}[1]{Section~\ref{#1}}
\providecommand{\sectionsaref}[2]{Sections~\ref{#1}~and~\ref{#2}}

\providecommand{\eqnref}[1]{Equation~\eqref{#1}}
\providecommand{\eqnsref}[1]{Equations~\eqref{#1}}
\providecommand{\eqnsaref}[2]{Equations~\eqref{#1}~and~\eqref{#2}}
\providecommand{\eqnssaref}[3]{Equations~\eqref{#1},~\eqref{#2}~and~\eqref{#3}}
\providecommand{\propsaref}[2]{Propositions~\ref{#1}~and~\ref{#2}}

\newlength{\graphlength}
\newlength{\cgreenlength}
\newcommand{\deDonder}{{d \negmedspace D \mspace{-2mu}}}

\newcommand{\one}{\mathbb{I}}
\newcommand{\coone}{\hat{\mathbb{I}}}
\newcommand{\Q}{{\boldsymbol{Q}}}
\newcommand{\HQ}{\mathcal{H}_{\Q}}
\newcommand{\EQ}{{\mathcal{E}_{\Q}}}
\newcommand{\RQ}{\mathcal{R}_{\Q}}
\newcommand{\RQO}{{\mathcal{R}_{\Q}^{[0]}}}
\newcommand{\RQI}{{\mathcal{R}_{\Q}^{[1]}}}
\newcommand{\AQ}{\mathcal{A}_{\Q}}
\newcommand{\SQ}[1]{\mathcal{S} \left ( #1 \right )}
\newcommand{\CQ}[1]{\mathcal{C} \left ( #1 \right )}

\newcommand{\DQ}[1]{\mathcal{D} \left ( #1 \right )}
\newcommand{\DQprime}[1]{\mathcal{D}^\prime \left ( #1 \right )}
\newcommand{\IQ}[1]{\mathcal{I} \left ( #1 \right )}
\newcommand{\IQnolim}[1]{\mathcal{I} ( #1 )}
\newcommand{\IQrv}{\mathcal{I}^{r, \mathbf{v}}}
\newcommand{\GQ}{\mathcal{G}_{\Q}}

\newcommand{\QQ}{\mathbf{Q}_{\Q}}
\newcommand{\qQ}{\mathbf{q}_{\Q}}
\newcommand{\Det}[1]{\operatorname{Det} \left ( #1 \right )}
\newcommand{\val}[1]{\operatorname{Val} \left ( #1 \right )}

\newcommand{\sym}[1]{\operatorname{Sym} \left ( #1 \right )}
\newcommand{\res}[1]{\operatorname{Res} \left ( #1 \right )}
\newcommand{\resnolim}[1]{\operatorname{Res} ( #1 )}
\newcommand{\cpl}[1]{\operatorname{Cpl} \left ( #1 \right )}

\newcommand{\cplgrd}[1]{\operatorname{CplGrd} \left ( #1 \right )}
\newcommand{\vtxgrd}[1]{\operatorname{VtxGrd} \left ( #1 \right )}
\newcommand{\vtxgrdnolim}[1]{\operatorname{VtxGrd} ( #1 )}
\newcommand{\extcpl}[1]{\operatorname{ExtCpl} \left ( #1 \right )}
\newcommand{\intcpl}[1]{\operatorname{IntCpl} \left ( #1 \right )}
\newcommand{\extvtx}[1]{\operatorname{ExtVtx} \left ( #1 \right )}
\newcommand{\intvtx}[1]{\operatorname{IntVtx} \left ( #1 \right )}

\newcommand{\ins}[2]{\operatorname{Ins} \left ( #1 \rhd #2 \right )}
\newcommand{\insaut}[3]{\operatorname{Ins}_\textup{Aut} \left ( #1 \rhd #2; #3 \right )}
\newcommand{\insrr}[1]{\operatorname{Ins}^{r, \mathbf{v}} \left ( #1 \right )}
\newcommand{\isoemb}[2]{\operatorname{Iso}_\textup{Emb} \left ( #1 \hookrightarrow #2 \right )}

\newcommand{\bettio}[1]{\kappa \left ( #1 \right )}
\newcommand{\bettii}[1]{\lambda \left ( #1 \right )}
\newcommand{\degp}[1]{\operatorname{Deg}_p \left ( #1 \right )}
\newcommand{\sdd}[1]{\omega \left ( #1 \right )}
\newcommand{\csdd}[1]{\varpi \left ( #1 \right )}
\newcommand{\rsdd}[1]{\rho \left ( #1 \right )}
\newcommand{\ssdd}[1]{\sigma \left ( #1 \right )}
\newcommand{\ssddn}[1]{\sigma_\text{n} \left ( #1 \right )}
\newcommand{\ssddr}[1]{\sigma_\text{r} \left ( #1 \right )}
\newcommand{\ssdds}[1]{\sigma_\text{s} \left ( #1 \right )}
\newcommand{\D}[1]{\Delta \left ( #1 \right )}
\newcommand{\antipode}[1]{S \left ( #1 \right )}
\newcommand{\mult}{m}
\newcommand{\ring}{\Bbbk}
\newcommand{\field}{\mathbb{K}}
\newcommand{\FR}[1]{\Phi \left ( #1 \right )}
\newcommand{\FRP}[1]{\left ( \Phi \circ \mathscr{A} \right ) \left ( #1 \right )}

\newcommand{\RFR}[1]{\Phi_{\mathscr{R}} \left ( #1 \right )}
\newcommand{\regFR}{{\Phi_\mathscr{E}^\varepsilon}}
\newcommand{\renFR}{{\Phi_\mathscr{R}}}
\newcommand{\countertermsymbol}{S_\mathscr{R}^\regFR}
\newcommand{\counterterm}[1]{\countertermsymbol \left ( #1 \right )}
\newcommand{\renscheme}[1]{\mathscr{R} \left ( #1 \right )}
\newcommand{\textfrac}[2]{#1 / #2}

\newcommand{\id}{\operatorname{Id}}
\newcommand{\precombgreen}{\mathfrak{x}}
\newcommand{\combgreen}{\mathfrak{X}}
\newcommand{\rescombgreen}{\mathfrak{X}}
\newcommand{\combcharge}{\mathfrak{Q}}
\newcommand{\iQ}{\mathfrak{i}_\Q}
\newcommand{\ZvQ}{\mathbb{Z}^{\mathfrak{v}_\Q}}
\newcommand{\ZqQ}{\mathbb{Z}^{\mathfrak{q}_\Q}}
\newcommand{\surject}{\to \!\!\!\!\! \to}
\newcommand{\cgreen}[1]{\vcenter{\hbox{\includegraphics[width=\cgreenlength]{#1}}}}
\newcommand{\tcgreen}[1]{\vcenter{\hbox{\includegraphics[width=0.75\cgreenlength]{#1}}}}
\newcommand{\enter}{\vspace{\baselineskip}}
\newcommand{\mathbbit}[1]{{\mspace{-1mu} \italicbox{$\mathbb{#1}$} \mspace{2mu}}}
\newcommand{\bbL}{\mathbbit{L}}
\newcommand{\bbI}{\mathbbit{I} \mspace{1mu}}
\newcommand{\bbT}{\mathbbit{T} \mspace{1mu}}
\newcommand{\bbsL}{\mathbbit{\scriptstyle{L}}}

\newcommand{\bbsT}{\mathbbit{\scriptstyle{T}}}
\newcommand{\GCD}[2]{\operatorname{GCD} \left ( #1, #2 \right )}
\newcommand{\legnumber}[1]{\mspace{-12mu} \raisebox{0.2ex}{$\scriptscriptstyle{#1}$} \mspace{9mu}}
\newcommand{\legnumberlongitudinal}[1]{\mspace{-23mu} \raisebox{0.2ex}{$\scriptscriptstyle{#1}$} \mspace{15mu}}
\newcommand{\legnumberghost}[1]{\mspace{-13mu} \raisebox{0.2ex}{$\scriptscriptstyle{#1}$} \mspace{9mu}}
\newcommand{\legnumberexponent}[1]{\mspace{-16mu} \raisebox{0.05ex}{$\scriptscriptstyle{#1}$} \mspace{9mu}}
\newcommand{\legnumberexponentlong}[1]{\mspace{-17mu} \raisebox{0.05ex}{$\scriptscriptstyle{#1}$} \mspace{-3mu}}
\newcommand{\legnumberexponentlongitudinal}[1]{\mspace{-29mu} \raisebox{0.05ex}{$\scriptscriptstyle{#1}$} \mspace{23mu}}
\newcommand{\legnumberexponentghost}[1]{\mspace{-17mu} \raisebox{0.05ex}{$\scriptscriptstyle{#1}$} \mspace{9mu}}

\newsavebox{\foobox}
\newcommand{\italicbox}[2][.25]
{%
	\mbox
	{%
		\sbox{\foobox}{#2}%
		\hskip\wd\foobox
		\pdfsave
		\pdfsetmatrix{1 0 #1 1}%
		\llap{\usebox{\foobox}}%
		\pdfrestore
	}%
}

\makeatletter
\newcommand{\subalign}[1]{%
  \vcenter{%
    \Let@ \restore@math@cr \default@tag
    \baselineskip\fontdimen10 \scriptfont\tw@
    \advance\baselineskip\fontdimen12 \scriptfont\tw@
    \lineskip\thr@@\fontdimen8 \scriptfont\thr@@
    \lineskiplimit\lineskip
    \ialign{\hfil$\m@th\scriptstyle##$&$\m@th\scriptstyle{}##$\crcr
      #1\crcr
    }%
  }
}
\makeatother

\title{\textsc{Gauge Symmetries and Renormalization}}
\author{David Prinz\footnote{Department of Mathematics and Department of Physics at Humboldt University of Berlin and Max Planck Institute for Gravitational Physics (Albert Einstein Institute) in Potsdam-Golm; prinz@\{math.hu-berlin.de, physik.hu-berlin.de, aei.mpg.de\}}}
\date{July 31, 2022}

\begin{document}

\maketitle

\begin{abstract}
	We study the perturbative renormalization of quantum gauge theories in the Hopf algebra setup of Connes and Kreimer. It was shown by van Suijlekom (2007) that the quantum counterparts of gauge symmetries --- the so-called Ward--Takahashi and Slavnov--Taylor identities --- correspond to Hopf ideals in the respective renormalization Hopf algebra. We generalize this correspondence to super- and non-renormalizable Quantum Field Theories, extend it to theories with multiple coupling constants and add a discussion on transversality. In particular, this allows us to apply our results to (effective) Quantum General Relativity, possibly coupled to matter from the Standard Model, as was suggested by Kreimer (2008). To this end, we introduce different gradings on the renormalization Hopf algebra and study combinatorial properties of the superficial degree of divergence. Then we generalize known coproduct and antipode identities to the super- and non-renormalizable cases and to theories with multiple vertex residues. Building upon our main result, we provide criteria for the compatibility of these Hopf ideals with the corresponding renormalized Feynman rules. A direct consequence of our findings is the well-definedness of the Corolla polynomial for Quantum Yang--Mills theory without reference to a particular renormalization scheme.
\end{abstract}

\section{Introduction} \label{sec:introduction}

In classical physics, Noether's Theorem relates symmetries to conserved quantities \cite{Noether}. In the context of classical gauge theories it states that gauge symmetries correspond to charge conservation. Thus, gauge symmetries are a fundamental ingredient in physical theories, such as the Standard Model and General Relativity. When considering their quantizations, identities related to their gauge invariance ensure that the gauge fields are indeed transversal: More precisely, the Ward--Takahashi and Slavnov--Taylor identities ensure that photons and gluons do only possess their two experimentally verified transversal degrees of freedom \cite{Ward,Takahashi,Taylor,Slavnov}.\footnote{Actually, Slavnov--Taylor identities were first discovered diagrammatically by Gerard 't Hooft in \cite{tHooft}.} As we will see in the upcoming examples, these identities relate the prefactors of different monomials in the Lagrange density that are linked via gauge transformations. In the corresponding Quantum Field Theories, however, each of these monomials might obtain a different energy dependence through its \(Z\)-factor, which could spoil this symmetry. Additionally, in order to calculate the propagator of the gauge field, a gauge fixing needs to be chosen. This gauge fixing then needs to be accompanied by its corresponding ghost and antighost fields, which have Grassmannian parity, i.e.\ obey Fermi--Dirac statistics. More precisely, the ghost fields are designed to satisfy the residual gauge transformations as equations of motion, whereas the antighost fields are designed to be constant and thus act as Lagrange multipliers.\footnote{We remark that this is the case for the Faddeev--Popov ghost construction. It is possible to construct a more general setup, which mixes or even reverses their roles \cite{Baulieu_Thierry-Mieg}.} This then ensures transversality of the gauge bosons, if the \(Z\)-factors fulfill certain identities. These identities, in turn, depend on the Feynman rules and the chosen renormalization scheme.

As a first example, consider Quantum Yang--Mills theory with a Lorenz gauge fixing, given via the Lagrange density
\begin{equation} \label{eqn:qym_lagrange_density_introduction}
\begin{split}
	\mathcal{L}_\text{QYM} & := \mathcal{L}_\text{YM} + \mathcal{L}_\text{GF} + \mathcal{L}_\text{Ghost} \\ & \phantom{:} = \eta^{\mu \nu} \eta^{\rho \sigma} \delta_{a b} \left ( - \frac{1}{4 \mathrm{g}^2} F^a_{\mu \rho} F^b_{\nu \sigma} - \frac{1}{2 \xi} \big ( \partial_\mu A^a_\nu \big ) \big ( \partial_\rho A^b_\sigma \big ) \right ) \dif V_\eta \\
	& \phantom{:=} + \eta^{\mu \nu} \left ( \frac{1}{\xi} \overline{c}_a \left ( \partial_\mu \partial_\nu c^a \right ) + \mathrm{g} \tensor{f}{^a _b _c} \overline{c}_a \left ( \partial_\mu \big ( c^b A^c_\nu \big ) \right ) \right ) \dif V_\eta \, ,
\end{split}
\end{equation}
where \(F^a_{\mu \nu} := \mathrm{g} \big ( \partial_\mu A^a_\nu - \partial_\nu A^a_\mu \big ) - \mathrm{g}^2 \tensor{f}{^a _b _c} A^b_\mu A^c_\nu\) is the local curvature form of the gauge boson \(A^a_\mu\). Furthermore, \(\dif V_\eta := \dif t \wedge \dif x \wedge \dif y \wedge \dif z\) denotes the Minkowskian volume form. Additionally, \(\eta^{\mu \nu} \partial_\mu A^a_\nu \equiv 0\) is the Lorenz gauge fixing functional and \(\xi\) denotes the gauge fixing parameter. Finally, \(c^a\) and \(\overline{c}_a\) are the gauge ghost and gauge antighost, respectively. Then, the decomposition into its monomials is given via 
\begin{equation}
\begin{split}
	\mathcal{L}_\text{QYM} & \equiv - \frac{1}{2} \eta^{\mu \nu} \eta^{\rho \sigma} \delta_{ab} \left ( \big ( \partial_\mu A^a_\rho \big ) \big ( \partial_\nu A^b_\sigma \big ) - \left ( 1 - \frac{1}{\xi} \right ) \big ( \partial_\mu A^a_\nu \big ) \big ( \partial_\rho A^b_\sigma \big ) \right ) \dif V_\eta \\
	& \phantom{\equiv} + \frac{1}{2} \mathrm{g} \eta^{\mu \nu} \eta^{\rho \sigma} f_{abc} \left ( \big ( \partial_\mu A^a_\rho \big ) \big ( A^b_\nu A^c_\sigma \big ) \right ) \dif V_\eta \\
	& \phantom{\equiv} - \frac{1}{4} \mathrm{g}^2 \eta^{\mu \nu} \eta^{\rho \sigma} \tensor{f}{^a _b _c} f_{ade} \big ( A^b_\mu A^c_\rho A^d_\nu A^e_\sigma \big ) \dif V_\eta \\
	& \phantom{\equiv} + \frac{1}{\xi} \eta^{\mu \nu} \overline{c}_a \left ( \partial_\mu \partial_\nu c^a \right ) \dif V_\eta + \mathrm{g} \eta^{\mu \nu} \tensor{f}{^a _b _c} \overline{c}_a \left ( \partial_\mu \big ( c^b A^c_\nu \big ) \right ) \dif V_\eta \, ,
\end{split}
\end{equation}
each of which contributing to a different Feynman rule. To absorb the upcoming divergences in a multiplicative manner, we multiply each monomial with an individual function \(Z^r \! \left ( \varepsilon \right )\) in the regulator \(\varepsilon \in \mathbb{R}\):
\begin{equation}
\begin{split}
	\mathcal{L}_\text{QYM}^\text{R} \left ( \varepsilon \right ) & := - \frac{1}{2} \eta^{\mu \nu} \eta^{\rho \sigma} \delta_{ab} \, Z^{A^2} \! \! \left ( \varepsilon \right ) \! \left ( \! \big ( \partial_\mu A^a_\rho \big ) \big ( \partial_\nu A^b_\sigma \big ) - \left ( 1 - \frac{Z^{1 / \xi} \! \left ( \varepsilon \right )}{\xi} \right ) \! \big ( \partial_\mu A^a_\nu \big ) \big ( \partial_\rho A^b_\sigma \big ) \! \right ) \dif V_\eta \\
	& \phantom{:=} + \frac{1}{2} \mathrm{g} \eta^{\mu \nu} \eta^{\rho \sigma} f_{abc} \, Z^{A^3} \! \! \left ( \varepsilon \right ) \left ( \big ( \partial_\mu A^a_\rho \big ) \big ( A^b_\nu A^c_\sigma \big ) \right ) \dif V_\eta \\
	& \phantom{:=} - \frac{1}{4} \mathrm{g}^2 \eta^{\mu \nu} \eta^{\rho \sigma} \tensor{f}{^a _b _c} f_{ade} \, Z^{A^4} \! \! \left ( \varepsilon \right ) \big ( A^b_\mu A^c_\rho A^d_\nu A^e_\sigma \big ) \dif V_\eta \\
	& \phantom{:=} + \frac{1}{\xi} \eta^{\mu \nu} Z^{\overline{c} c / \xi} \! \left ( \varepsilon \right ) \left ( \overline{c}_a \left ( \partial_\mu \partial_\nu c^a \right ) \right ) \dif V_\eta + \mathrm{g} \eta^{\mu \nu} \tensor{f}{^a _b _c} \,  Z^{\overline{c} c A} \! \left ( \varepsilon \right ) \left ( \overline{c}_a \left ( \partial_\mu \big ( c^b A^c_\nu \big ) \right ) \! \right ) \dif V_\eta \, ,
\end{split}
\end{equation}
where the regulator \(\varepsilon\) is related to the energy scale through the choice of a regularization scheme. Then, the invariance of \(\mathcal{L}_\text{QYM}^\text{R} \left ( \varepsilon \right )\) under (residual) gauge transformations away from the reference point (where all \(Z\)-factors fulfill \(Z^r \! \left ( \varepsilon_0 \right ) = 1\)) depends on the following identities:
\begin{subequations} \label{eqns:z-factor_identities_qym}
\begin{align}
	\frac{\big ( Z^{A^3} \! \! \left ( \varepsilon \right ) \! \big )^2}{Z^{A^2} \! \! \left ( \varepsilon \right )} & \equiv Z^{A^4} \! \! \left ( \varepsilon \right )
	\intertext{and}
	\frac{Z^{A^3} \! \! \left ( \varepsilon \right )}{Z^{A^2 / \xi} \! \left ( \varepsilon \right )} & \equiv \frac{Z^{\overline{c} c A} \! \left ( \varepsilon \right )}{Z^{\overline{c} c / \xi} \! \left ( \varepsilon \right )}
\end{align}
\end{subequations}
for all \(\varepsilon\) in the domain of the regularization scheme and with \(Z^{A^2 / \xi} \! \left ( \varepsilon \right ) := Z^{A^2} \! \! \left ( \varepsilon \right ) Z^{1 / \xi} \! \left ( \varepsilon \right )\). This example is continued, using the terminology introduced throughout the article, in \exref{exmp:qym}, where we will reencounter these identities in a different language. We remark that these identities are essential for the Faddeev--Popov ghost construction to work, such that physical gluons are indeed transversal.

As a second example, consider (effective) Quantum General Relativity with the metric decomposition \(g_{\mu \nu} \equiv \eta_{\mu \nu} + \varkappa h_{\mu \nu}\), where \(h_{\mu \nu}\) is the graviton field and \(\varkappa := \sqrt{\kappa}\) the graviton coupling constant (with \(\kappa := 8 \pi G\) the Einstein gravitational constant), and a linearized de Donder gauge fixing, given via the Lagrange density
\begin{equation} \label{eqn:qgr_lagrange_density_introduction}
\begin{split}
	\mathcal{L}_\text{QGR} & := \mathcal{L}_\text{GR} + \mathcal{L}_\text{GF} + \mathcal{L}_\text{Ghost} \\ & \phantom{:} = - \frac{1}{2 \varkappa^2} \left ( \sqrt{- \Det{g}} R + \frac{1}{2 \zeta} \eta^{\mu \nu} \deDonder^{(1)}_\mu \deDonder^{(1)}_\nu \right ) \dif V_\eta \\ & \phantom{:=} - \frac{1}{2} \eta^{\rho \sigma} \left ( \frac{1}{\zeta} \overline{C}^\mu \left ( \partial_\rho \partial_\sigma C_\mu \right ) + \overline{C}^\mu \left ( \partial_\mu \big ( \tensor{\Gamma}{^\nu _\rho _\sigma} C_\nu \big ) - 2 \partial_\rho \big ( \tensor{\Gamma}{^\nu _\mu _\sigma} C_\nu \big ) \right ) \right ) \dif V_\eta \, ,
\end{split}
\end{equation}
where \(R := g^{\nu \sigma} \tensor{R}{^\mu _\nu _\mu _\sigma}\) is the Ricci scalar (with \(\tensor{R}{^\rho _\sigma _\mu _\nu} := \partial_\mu \tensor{\Gamma}{^\rho _\nu _\sigma} - \partial_\nu \tensor{\Gamma}{^\rho _\mu _\sigma} + \tensor{\Gamma}{^\rho _\mu _\lambda} \tensor{\Gamma}{^\lambda _\nu _\sigma} - \tensor{\Gamma}{^\rho _\nu _\lambda} \tensor{\Gamma}{^\lambda _\mu _\sigma}\) the Riemann tensor). Again, \(\dif V_\eta := \dif t \wedge \dif x \wedge \dif y \wedge \dif z\) denotes the Minkowskian volume form, which is related to the Riemannian volume form \(\dif V_g\) via \(\dif V_g \equiv \sqrt{- \Det{g}} \dif V_\eta\). Additionally, \(\deDonder^{(1)}_\mu := \eta^{\rho \sigma} \Gamma_{\mu \rho \sigma} \equiv 0\) is the linearized de Donder gauge fixing functional and \(\zeta\) denotes the gauge fixing parameter. Finally, \(C_\mu\) and \(\overline{C}^\mu\) are the graviton-ghost and graviton-antighost, respectively. We refer to \cite{Prinz_2,Prinz_4} for more detailed introductions and further comments on the chosen conventions. Then, we decompose \(\mathcal{L}_\text{QGR}\) with respect to its powers in the gravitational coupling constant \(\varkappa\), the gauge fixing parameter \(\zeta\) and the ghost field \(C\) as follows\footnote{We omit the term \(\mathcal{L}_\text{QGR}^{-1,0,0}\) as it is given by a total derivative.}
\begin{equation}
	\mathcal{L}_\text{QGR} \equiv \sum_{i = 0}^\infty \sum_{j = -1}^0 \sum_{k = 0}^1 \mathcal{L}_\text{QGR}^{i,j,k} \, ,
\end{equation}
where we have set \(\mathcal{L}_\text{QGR}^{i,j,k} := \eval[1]{\left ( \mathcal{L}_\text{QGR} \right )}_{O (\varkappa^i \zeta^j C^k)}\), cf.\ \cite[Section 3]{Prinz_4}. Again, to absorb the upcoming divergences in a multiplicative manner, we multiply each monomial with an individual function \(Z^r \! \left ( \varepsilon \right )\) in the regulator \(\varepsilon \in \mathbb{R}\):
\begin{equation}
	\mathcal{L}_\text{QGR}^\text{R} \left ( \varepsilon \right ) := \sum_{i = 0}^\infty \sum_{j = -1}^0 \sum_{k = 0}^1 Z^{i,j,k} \! \left ( \varepsilon \right ) \mathcal{L}_\text{QGR}^{i,j,k} \, ,
\end{equation}
where the regulator \(\varepsilon\) is related to the energy scale through the choice of a regularization scheme. Then, the invariance of \(\mathcal{L}_\text{QGR}^\text{R} \left ( \varepsilon \right )\) under (residual) diffeomorphisms away from the reference point (where all \(Z\)-factors fulfill \(Z^r \! \left ( \varepsilon_0 \right ) = 1\)) depends on the following identities:
\begin{subequations} \label{eqns:z-factor_identities_qgr}
	\begin{align}
	\frac{Z^{i,0,0} \! \left ( \varepsilon \right ) Z^{1,0,0} \! \left ( \varepsilon \right )}{Z^{0,0,0} \! \left ( \varepsilon \right )} & \equiv Z^{(i+1),0,0} \! \left ( \varepsilon \right )
	\intertext{and}
	\frac{Z^{i,0,0} \! \left ( \varepsilon \right )}{Z^{0,-1,0} \! \left ( \varepsilon \right )} & \equiv \frac{Z^{i,0,1} \! \left ( \varepsilon \right )}{Z^{0,-1,1} \! \left ( \varepsilon \right )}
	\end{align}
\end{subequations}
for all \(i \in \mathbb{N}_+\) and \(\varepsilon\) in the domain of the regularization scheme. This example is continued, using the terminology introduced throughout the article, in \exref{exmp:qgr}, where we will reencounter these identities in a different language. Again, we remark that these identities are essential for the Faddeev--Popov ghost construction to work, such that physical gravitons are indeed transversal.

Identities for \(Z\)-factors, such as \eqnsaref{eqns:z-factor_identities_qym}{eqns:z-factor_identities_qgr}, are known in the literature as Ward--Takahashi identities in the realm of Quantum Electrodynamics and Slavnov--Taylor identities in the realm of Quantum Chromodynamics \cite{Ward,Takahashi,Taylor,Slavnov,tHooft}. We will study these identities on a general level and thus call them `quantum gauge symmetries (QGS)'. In particular, our results are directly applicable to (effective) Quantum General Relativity in the sense of \eqnref{eqns:z-factor_identities_qgr}, as was first suggested in \cite{Kreimer_QG1} and then studied for a scalar toy model in \cite{Kreimer_vSuijlekom}. We also refer to \cite{Prinz_2} for a more detailed introduction to (effective) Quantum General Relativity coupled to Quantum Electrodynamics and to \cite{Prinz_4} for the Feynman rules of (effective) Quantum General Relativity coupled to the Standard Model. In the present article we study the renormalization-related properties of quantum gauge symmetries, such as \eqnsaref{eqns:z-factor_identities_qym}{eqns:z-factor_identities_qgr}. This is done in the framework of the Connes--Kreimer renormalization Hopf algebra: In this setup, the organization of subdivergences of Feynman graphs is encoded into the coproduct of a Hopf algebra \cite{Kreimer_Hopf_Algebra} and the renormalization of Feynman rules is described via an algebraic Birkhoff decomposition \cite{Connes_Kreimer_0}. Then, the aforementioned identities induce symmetries inside the renormalization Hopf algebra, which was first studied via Hochschild cohomology \cite{Kreimer_Anatomy} and then shown to be Hopf ideals \cite{vSuijlekom_QED,vSuijlekom_QCD,vSuijlekom_BV,Kreimer_vSuijlekom}. The aim of this article is to generalize these results in several directions: We first introduce an additional coupling-grading in \defnref{defn:connectedness_gradings_renormalization_hopf_algebra}, which allows us to study theories with multiple coupling constants, like Quantum General Relativity coupled to the Standard Model. Furthermore, and more substantially, this allows us to discuss the transversality of such (generalized) quantum gauge theories, cf.\ \defnref{defn:transversal_structure} and \defnref{defn:quantum_gauge_symmetries}. Moreover, we generalize known coproduct and antipode identities to super- and non-renormalizable Quantum Field Theories (QFTs) in \propref{prop:proj_div_graphs_coprod} and \sectionref{sec:coproduct_and_antipode_identities}. This requires a detailed study of combinatorial properties of the superficial degree of divergence, as is presented in \sectionref{sec:a_superficial_argument}. The analysis then culminates in \thmref{thm:quantum_gauge_symmetries_induce_hopf_ideals}, stating that quantum gauge symmetries correspond to Hopf ideals also in this generalized context. Finally, we provide criteria for the validity of quantum gauge symmetries in terms of Feynman rules and renormalization schemes: Technically speaking, this corresponds to the situation that the aforementioned Hopf ideals are in the kernel of the counterterm map or even the renormalized Feynman rules. The result is then presented in \thmref{thm:criterion_ren-hopf-mod}. A consequence thereof is that the Corolla polynomial for Quantum Yang--Mills theory is well-defined without reference to a particular renormalization scheme if the renormalization scheme is proper, cf.\ \defnref{defn:proper_renormalization_scheme} and \remref{rem:corolla_polynomial}. The Corolla polynomial is a graph polynomial in half-edges that relates amplitudes in Quantum Yang--Mills theory to amplitudes in \(\phi^3_4\)-theory \cite{Kreimer_Yeats,Kreimer_Sars_vSuijlekom,Kreimer_Corolla}. More precisely, this graph polynomial is used for the construction of a so-called Corolla differential that acts on the parametric representation of Feynman integrals \cite{Kreimer_Sars_vSuijlekom,Sars_PhD,Golz_PhD}. Thereby, the corresponding cancellation-identities \cite{tHooft_Veltman,Citanovic,Sars_PhD,Kissler_Kreimer,Gracey_Kissler_Kreimer,Kissler} are encoded into a double complex of Feynman graphs, leading to Feynman graph cohomology \cite{Kreimer_Sars_vSuijlekom,Berghoff_Knispel}. We remark that this construction has been successfully generalized to Quantum Yang--Mills theories with spontaneous symmetry breaking \cite{Prinz_1} and Quantum Electrodynamics with spinors \cite{Golz_1,Golz_2,Golz_3}. The application of this formulation to (effective) Quantum General Relativity is a topic of ongoing research.

This article is organized as follows: We start in \sectionref{sec:preliminaries_of_hopf_algebraic_renormalization} with a brief introduction to Hopf algebraic renormalization, stating the necessary definitions and conventions. Then, in \sectionref{sec:a_superficial_argument}, we study combinatorial properties of the superficial degree of divergence, allowing to state our results not only for the class of renormalizable Quantum Field Theories, but also the more involved classes of super- and non-renormalizable Quantum Field Theories. Next, in \sectionref{sec:coproduct_and_antipode_identities}, we reprove and generalize known coproduct and antipode identities. Subsequently, in \sectionref{sec:quantum_gauge_symmetries_and_subdivergences}, we show that quantum gauge symmetries induce Hopf ideals inside the renormalization Hopf algebra, even in our general context. Thereafter, in \sectionref{sec:quantum_gauge_symmetries_and_renormalized_feynman_rules}, we provide criteria for their validity on the level of renormalized Feynman rules, that is, criteria for the unrenormalized Feynman rules and the renormalization scheme. Finally, in \sectionref{sec:conclusion}, we conclude our investigations and provide an outlook into further projects.

\section{Preliminaries of Hopf algebraic renormalization} \label{sec:preliminaries_of_hopf_algebraic_renormalization}

We start this article by briefly recalling the relevant definitions and notations from Hopf algebraic renormalization: We consider \(\Q\) to be a local Quantum Field Theory (QFT), i.e.\ a QFT given by a Lagrange functional. Then, in a nutshell, the renormalization Hopf algebra\footnote{We use the symbol \(\HQ\) by abuse of notation simultaneously for the vector space \(\HQ\) as well as for the complete renormalization Hopf algebra \((\HQ, m, \one, \Delta, \coone, S)\).} \(\HQ\) of a QFT \(\Q\) consists of a vector space \(\HQ\) with algebra structure \((\HQ, m, \one)\), coalgebra structure \((\HQ, \Delta, \coone)\) and antipode \(S \colon \HQ \to \HQ\). More precisely, given the set \(\GQ\) of 1PI Feynman graphs of \(\Q\), the vector space \(\HQ\) is defined as the vector space over \(\mathbb{Q}\) generated by the elements of the set \(\GQ\) and disjoint unions thereof. Then, the multiplication \(m\) is simply given via disjoint union, with the empty graph as unit. The interesting structures are the coproduct \(\Delta\) and the antipode \(S\): It was realized by Kreimer that the organization of subdivergences of Feynman graphs can be encoded into a coalgebra structure on \(\HQ\) \cite{Kreimer_Hopf_Algebra}. Then, building upon this, Connes and Kreimer formulated the renormalized Feynman rules \(\Phi_\mathscr{R}\) as an algebraic Birkhoff decomposition with respect to the renormalization scheme \(\mathscr{R}\) \cite{Connes_Kreimer_0}. See \defnref{defn:renormalization_hopf_algebra} and \defnref{defn:fr_reg_ren_counterterm} for the formal definitions and \cite[Section 3]{Prinz_2} for a more detailed introduction using the same notations and conventions.\footnote{We remark that some of the introductory material in this section is borrowed from \cite[Section 3]{Prinz_2}.} This mathematical formulation of the renormalization operation allows for a precise analysis of symmetries via Hopf ideals. We want to deepen this viewpoint in the context of quantum gauge theories by generalizing results from \cite{Kreimer_Anatomy,vSuijlekom_QED,vSuijlekom_QCD,vSuijlekom_BV,Kreimer_vSuijlekom}. Finally, we also mention some detailed introductory texts \cite{Sweedler,Manchon,Figueroa_Gracia-Bondia,Guo,Panzer,Yeats}.

\enter

\begin{defn}[Multiset over a set] \label{defn:multiset_over_a_set}
	Let \(M\) and \(S\) be sets. The set \(M\) is called a multiset over \(S\), if \(M\) contains elements of \(S\) in arbitrary multiplicity. Then, the multiset
	\begin{equation} \label{eqn:multiset_over_a_set}
		\pi \, : \quad M \to S \, , \quad m_s \mapsto s \, ,
	\end{equation}
	where \(\pi\) projects the elements of \(M\) to \(S\), can be canonically identified with the set \((s, n_s) \in \widetilde{M}^S \subset S \times \mathbb{N}_0\), where the natural number \(n_s\) indicates the multiplicity of each element \(s\) in \(M\) (which can possibly be zero). We call \(\widetilde{M}^S\) the multiset representation of \(M\) over \(S\). Given the multiset representation \(\widetilde{M}^S\) of \(M\) over \(S\), we define the two projections
	\begin{subequations}
	\begin{align}
		\varsigma \, & : \quad \widetilde{M}^S \to S \, , \quad (s, n_s) \mapsto s
		\intertext{and}
		\varrho \, & : \quad \widetilde{M}^S \to \mathbb{N}_0 \, , \quad (s, n_s) \mapsto n_s \, .
	\end{align}
	\end{subequations}
	Additionally, if the elements in the set \(S\) are ordered, we define the corresponding multiset-vector as the vector \(\mathbf{n} := (n_1, \dots, n_\mathfrak{s})^\intercal \in \mathbb{N}_0^\mathfrak{s}\), where \(\mathfrak{s}\) is the cardinality of the set \(S\) and \(n_i\) denotes the multiplicity of the element \(s_i\) in \(M\). Furthermore, two multisets \(M_1\) and \(M_2\) over the same set \(S\) are called isomorphic, if each element \(s \in S\) has the same multiplicity \(n_s\) in either \(M_1\) and \(M_2\). In the following, we will always assume that the underlying set \(S\) is ordered and thus use the equivalence between multisets and their multiset-vectors.
\end{defn}

\enter

\begin{rem}
	Given the situation of \defnref{defn:multiset_over_a_set}, a multiset \(M\) over \(S\) and its multiset representation \(\widetilde{M}^S\) are in general different sets, as they might have different cardinalities. As an extreme example, every set can be seen as a multiset over the singleton. Therefore, its multiset representation consist only of the element \((*, n)\), where \(n\) is the cardinality of the set. On the other hand, a set viewed as a multiset over itself has the same cardinality as its underlying set, but its elements are distinctly marked. Accordingly, its multiset representation consist of elements \((s, 1)\), for each element \(s\) in the underlying set. In this spirit, a multiset \(M\) over a set \(S\) can be seen as a \(S\)-colored set, by means of the map \(\pi\) from \eqnref{eqn:multiset_over_a_set}.
\end{rem}

\enter

\begin{defn}[Residue, amplitude and coupling constant set] \label{defn:residue_amplitude_and_coupling_constant_set}
	Let \(\Q\) be a QFT given via the Lagrange density \(\mathcal{L}_\Q\). Then each monomial in \(\mathcal{L}_\Q\) describes either a fundamental interaction or a propagation of the involved particles. We collect this information in two sets, called vertex residue set \(\RQO\) and edge residue set \(\RQI\), as follows: The first set consists of all fundamental interactions and the second set consists of all propagators, or, equivalently, particle types of \(\Q\). Finally, the residue set is then defined as the disjoint union
	\begin{equation}
		\RQ := \RQO \sqcup \RQI \, .
	\end{equation}
	We denote the cardinality of the vertex set via \(\mathfrak{v}_\Q := \# \RQO\). Furthermore, we define the set of amplitudes \(\AQ\) as the set containing all possible external leg structures of 1PI Feynman graphs. In particular, it is given as the disjoint union
	\begin{equation}
		\AQ := \RQ \sqcup \mathcal{Q}_\Q \, ,
	\end{equation}
	where \(\mathcal{Q}_\Q\) denotes the set of pure quantum corrections, that is, interactions which are only possible via trees or Feynman graphs, but not directly via residues in the set \(\RQ\). If \(\Q\) is a quantum gauge theory, we add additional labels to the edge-types: One for the physical degrees of freedom and at least one for the unphysical degrees of freedom, cf.\ \remref{rem:longitudinal_and_transversal_gauge_fields} and \cite{Prinz_8}. Moreover, we denote by \(\qQ\) the set of physical coupling constants and, if present, gauge fixing parameters appearing in the Lagrange density \(\mathcal{L}_\Q\). Finally, we define the function
	\begin{equation} \label{eqn:coupling-coloring_function}
		\theta \, : \quad \AQ \to \qQ \, , \quad r \mapsto \begin{cases} q_r := q_v \left ( \prod_{e \in E \left ( v \right )} \sqrt{q_e} \right ) & \text{if \(r = v \in \RQO\)} \\ 1 & \text{else, i.e.\ \(r \in \left ( \AQ \setminus \RQO \right )\)} \end{cases} \, ,
	\end{equation}
	where \(q_v\) denotes the coupling constant that scales the vertex-type \(v\), \(q_e\) denotes the gauge fixing parameter that is associated to the edge-type \(e\) if it is unphysical and finally \(E \left ( v \right )\) denotes the set of edges that are attached to the vertex \(v\). We denote the cardinality of the set of physical coupling constants via \(\mathfrak{q}_\Q := \# \qQ\).
\end{defn}

\enter

\begin{defn}[Transversal structure] \label{defn:transversal_structure}
	Let \(\Q\) be a quantum gauge theory. Then each independent gauge fixing term induces a longitudinal projection operator \(\boldsymbol{L}\) for the propagator of the corresponding gauge field. Together with the respective identity operator \(\boldsymbol{I}\) we define the associated transversal projection operator \(\boldsymbol{T}\) via
	\begin{equation}
		\boldsymbol{T} := \boldsymbol{I} - \boldsymbol{L} \, .
	\end{equation}
	We refer to the set \(\set{\boldsymbol{L}, \boldsymbol{I}, \boldsymbol{T}}\) as transversal structure. Additionally, let \(\mathfrak{f}_\Q\) denote the number of independent gauge fixing terms of \(\Q\).\footnote{This includes in particular the coupling of gravity to gauge theories, which requires independent gauge fixing terms for the diffeomorphism invariance and the gauge invariance, cf.\ e.g.\ \cite{Prinz_2,Prinz_4}. With that we also obtain two separate transversal structures: \(\set{L, I, T}\) for the Quantum Yang--Mills theory part and \(\set{\bbL, \bbI, \bbT}\) for the (effective) Quantum General Relativity part, cf.\ \eqnsref{eqn:projection_tensors_qym} and \eqnsref{eqn:projection_tensors_qgr}.} Then we consider the union
	\begin{equation}
		\mathcal{T}_\Q := \bigcup_{k = 1}^{\mathfrak{f}_\Q} \set{\boldsymbol{L}, \boldsymbol{I}, \boldsymbol{T}}_k
	\end{equation}
	and refer to it as the transversal structure of \(\Q\).
\end{defn}

\enter

\begin{rem} \label{rem:longitudinal_and_transversal_gauge_fields}
	The `physical' and `unphysical' labels together with the particle-type labels of \defnref{defn:residue_amplitude_and_coupling_constant_set} connect as follows to the physics of quantum gauge theories: Physical particle-types are transversal gauge field edges, canceled ghost field edges and matter field edges, respectively. Contrary, unphysical particle-types are longitudinal or canceled gauge field edges, ghost field edges and canceled matter field edges, respectively. Thus, our `physical' and `unphysical' labels are related to cancellation identities \cite{tHooft_Veltman,Citanovic,Sars_PhD,Kissler_Kreimer,Gracey_Kissler_Kreimer,Kissler} and the marking of edges in the construction of Feynman graph cohomology \cite{Kreimer_Sars_vSuijlekom,Berghoff_Knispel}. Additionally, if \(\Q\) has several longitudinal projection operators we need to keep track which longitudinal projection induced the cancellation of an edge. This is the reason why we extend our setup to allow for possibly several distinct `unphysical' labels, each of which is related to the corresponding gauge fixing parameter. This discussion will be studied in detail in \cite{Prinz_8}, cf.\ \cite{Prinz_5,Prinz_6}.
\end{rem}

\enter

\begin{defn}[(Feynman) graphs and Feynman graph set] \label{defn:feynman_graphs}
	A graph \(G := \left ( V, E, \beta \right )\) is given via a set of vertices \(V\), a set of edges \(E = E_0 \sqcup E_1\), where \(E_0\) is the subset of unoriented and \(E_1\) is the subset of oriented edges, and a morphism\footnote{We remark that the map \(\beta\) is necessary if graphs are allowed to have multi-edges or simultaneously oriented and unoriented edges, which is typically the case in physics.}
	\begin{equation}
		\beta \, : \quad E \to \left ( V \times V \times \mathbb{Z}_2 \right ) \, , \quad e \mapsto \begin{cases} \left ( v_1, v_2; 0 \right ) & \text{if \(e \in E_0\)} \\ \left ( v_i, v_t; 1 \right ) & \text{if \(e \in E_1\)} \end{cases}\, ,
	\end{equation}
	mapping edges to tuples of vertices together with their binary orientation information; if the edge is oriented, the order of the vertices is first initial then terminal. Given a graph \(G\), the corresponding sets are denoted via \(V \equiv V \left ( G \right ) \equiv G^{[0]}\) and \(E \equiv E \left ( G \right ) \equiv G^{[1]}\), where we omit the dependence on the graph \(G\) only if there is no ambiguity possible. Finally, given a QFT \(\Q\), we define a Feynman graph \(\Gamma := ( G, \{ *_p, *_f \}, E_\text{Ext}, \tau )\) as a graph \(G\) with the following extra structure: We add a set of external edges \(E_\text{Ext}\) and two external vertices \(\{ *_p, *_f \}\), where \(*_p\) is the endpoint for past external edges and \(*_f\) is the endpoint for future external edges. Then, we extend the map \(\beta\) to the set of external edges \(E_\text{Ext}\) via
	\begin{equation}
		\eval{\beta}_{E_\text{Ext}} \! \! \! \! \! \! \! : \quad E_\text{Ext} \to \left ( \big ( V \sqcup \set{*_p, *_f} \! \big ) \times \big ( V \sqcup \set{*_p, *_f} \! \big ) \times \mathbb{Z}_2 \right ) \, , \quad e \mapsto \begin{cases} \left ( v_1, v_2; 0 \right ) & \text{if \(e \in E_0\)} \\ \left ( v_i, v_t; 1 \right ) & \text{if \(e \in E_1\)} \end{cases}\, .
	\end{equation}
	Additionally, the vertex set \(V\) and the edge set \(E \sqcup E_\text{Ext}\) are considered as multisets over the vertex residue set \(\RQO\) and the edge residue set \(\RQI\), respectively:
	\begin{equation} \label{eqn:residue-coloring_function}
		\tau \, : \quad \left ( V \sqcup E \sqcup E_\text{Ext} \right ) \to \RQ \, , \quad r \mapsto \begin{cases} r_v \in \RQO & \text{if \(r \in V\)} \\ r_e \in \RQI & \text{if \(r \in E \sqcup E_\text{Ext}\)} \end{cases} \, ,
	\end{equation}
	where the map \(\tau\) corresponds to the map \(\pi\) from \eqnref{eqn:multiset_over_a_set}. Thus, using the coloring function \(\tau\), we view Feynman graphs as \(\RQ\)-colored graphs. Two Feynman graphs from the same QFT \(\Q\) are considered to be isomorphic, if they are isomorphic as \(\RQ\)-colored graphs and if they have the same external leg structure, cf.\ \defnref{defn:automorphisms_of_feynman_graphs}. Furthermore, using the composition \(\theta \circ \tau\) with the coupling constant function \(\theta\) from \eqnref{eqn:coupling-coloring_function}, we can also view the vertex set \(V\) as a multiset over the coupling constant set \(\qQ\), where then the composition \(\theta \circ \tau\) corresponds to the map \(\pi\) from \eqnref{eqn:multiset_over_a_set}. We remark, however, that edges are unlabeled in this coloring (labeled by \(1\)). Finally, a graph is called `one-particle irreducible (1PI)', if it is still connected after the removal of any of its internal edges.\footnote{In the mathematical literature these graphs are called bridge-free.} We denote the set of all 1PI Feynman graphs of \(\Q\) by \(\GQ\).
\end{defn}

\enter

\begin{defn}[Residue of a Feynman graph]
	Let \(\Q\) be a local QFT with residue set \(\RQ\) and 1PI Feynman graph set \(\GQ\). Then the external leg structure of a Feynman graph \(\Gamma \in \GQ\) is called its residue and denoted via \(\res{\Gamma} \in \AQ\). It is considered as the graph obtained from \(\Gamma\) by shrinking all its internal edges to a single vertex.
\end{defn}

\enter

\begin{defn}[Sets of half-edges, corollas and external vertex residues] \label{defn:half-edges_corollas_external-vertex-residue-set}
	Given a graph \(G\), we define the set of half-edges \(H \left ( G \right ) \equiv G^{[1/2]}\) via
	\begin{equation}
		H \left ( G \right ) := \set{h_v \cong (v,e) \left \vert \, v \in V , \, e \in E \text{ and } v \in \beta \left ( e \right ) \right .} \, ,
	\end{equation}
	where \(v \in \beta \left ( e \right )\) means that the vertex \(v\) is attached to the edge \(e\). The set of half-edges is then accompanied by the involution \(\iota\), which interchanges an internal half-edge with the internal half-edge that it forms the internal edge with and furthermore fixates external half-edges:
	\begin{equation}
		\iota \, : \quad H \left ( G \right ) \to H \left ( G \right ) \, , \quad (v,e) \mapsto \begin{cases} (w,e) & \text{if \(e \in E \left ( G \right )\) and \(v, w \in \beta \left ( e \right )\)} \\ (v,e) & \text{if \(e \in E_\text{Ext} \left ( G \right )\)} \end{cases}
	\end{equation}
	Thus, \(\iota\) can be used to reproduce the set \(E \left ( G \right ) \sqcup E_\text{Ext} \left ( G \right )\) from the set \(H \left ( G \right )\). Additionally, we define the set of corollas \(C \left ( G \right )\)
	\begin{equation}
		C \left ( G \right ) := \set{\left . c_v \cong \left ( v, \set{h_v \in H \left ( G, v \right )} \right ) \right \vert v \in V} \, ,
	\end{equation}
	where \(\{ h_v \in H \left ( G, v \right ) \}\) is the set of half-edges attached to the vertex \(v\). We also apply these constructions to Feynman graphs \(\Gamma\) by means of its underlying graph \(G\). Finally, we define the set of external vertex residues \(W \left ( \Gamma \right )\) of a (not necessary connected) Feynman graph \(\Gamma\) via
	\begin{equation}
		W \left ( \Gamma \right ) := \set{r_\gamma := \res{\gamma} \left \vert \, \text{\(\gamma\) connected component of \(\Gamma\) and } r_\gamma \in \RQO \right . } \, ,
	\end{equation}
	i.e.\ \(W \left ( \Gamma \right )\) is a multiset over \(\RQO\), by means of \(\tau\) from \eqnref{eqn:residue-coloring_function}, and furthermore a multiset over \(\qQ\), by means of \(\theta \circ \tau\) from \eqnsaref{eqn:coupling-coloring_function}{eqn:residue-coloring_function}.\footnote{We remark that if \(\Gamma\) is connected, then \(W \left ( \Gamma \right )\) contains at most one element.}
\end{defn}

\enter

\begin{defn}[Automorphisms of (Feynman) graphs] \label{defn:automorphisms_of_feynman_graphs}
	Let \(G\) be a graph. A map \(a \colon G \to G\), or, more precisely, the collection of its two underlying maps
	\begin{subequations}
	\begin{align}
		a_V \, & : \quad V \to V \, , \quad v_1 \mapsto v_2
		\intertext{and}
		a_E \, & : \quad E \to E \, , \quad e_1 \mapsto e_2
	\end{align}
	\end{subequations}
	is called an automorphism of \(G\), if \(a_V\) and \(a_E\) are bijections and additionally they are compatible with \(\beta\) in the sense that \(\beta \circ a_E \equiv ( a_V \times a_V \times \id_{\mathbb{Z}_2} ) \circ \beta\). Furthermore, given a Feynman graph \(\Gamma\) and a map \(\alpha \colon \Gamma \to \Gamma\). Then \(\alpha\) is called automorphism of \(\Gamma\), if \(\alpha\) is an automorphism of the underlying graph, additionally compatible with the coloring function \(\tau\), i.e.\ \(\tau \circ \alpha \equiv \tau\), and the identity on external edges. The group of automorphisms of a Feynman graph \(\Gamma\) will be denoted via \(\operatorname{Aut} \left ( \Gamma \right )\) and its rank via \(\operatorname{Sym} \left ( \Gamma \right )\), to which we refer as its symmetry factor.
\end{defn}

\enter

\begin{defn}[Feynman graph invariants] \label{defn:betti-numbers_and_multi-indices}
	Let \(\Q\) be a QFT, \(\GQ\) its 1PI Feynman graph set, \(\RQO\) its vertex residue set and \(\qQ\) its coupling constant set. We equip the elements in the sets \(\RQO\) and \(\qQ\) with an arbitrary ordering, in order to have well-defined multiset vectors in the sense of \defnref{defn:multiset_over_a_set}. Given a Feynman graph \(\Gamma \in \GQ\), we associate the following two numbers and four multi-indices to it:
	\begin{itemize}
		\item \(\bettio{\Gamma} \in \mathbb{N}_0\) denotes the number of its connected components
		\item \(\bettii{\Gamma} \in \mathbb{N}_0\) denotes the number of its independent loops,\footnote{In the mathematical literature this is usually called a cycle.} where we only consider loops by internal edges, i.e.\ remove the two external vertices \(\{ *_p, *_f \}\)
		\item \(\intvtx{\Gamma} \in \ZvQ\) denotes the multiset vector of \(V \left ( \Gamma \right )\) over \(\RQO\), with respect to \(\tau\) from \eqnref{eqn:residue-coloring_function}
		\item \(\extvtx{\Gamma} \in \ZvQ\) denotes the multiset vector of \(W \left ( \Gamma \right )\) over \(\RQO\), with respect to \(\tau\) from \eqnref{eqn:residue-coloring_function}
		\item \(\intcpl{\Gamma} \in \ZqQ\) denotes the multiset vector of \(V \left ( \Gamma \right )\) over \(\qQ\), with respect to \(\theta \circ \tau\) from \eqnsaref{eqn:coupling-coloring_function}{eqn:residue-coloring_function}
		\item \(\extcpl{\Gamma} \in \ZqQ\) denotes the multiset vector of \(W \left ( \Gamma \right )\) over \(\qQ\), with respect to \(\theta \circ \tau\) from \eqnsaref{eqn:coupling-coloring_function}{eqn:residue-coloring_function}
	\end{itemize}
	Then we extend these invariants to the unit \(\one \in \HQ\) by \(0 \in \mathbb{N}_0\), \(\mathbf{0} \in \ZvQ\) and \(\mathbf{0} \in \ZqQ\), respectively, and to disjoint unions of 1PI Feynman graphs via addition, i.e.
	\begin{equation}
		\operatorname{Inv} \left ( \Gamma_1 \sqcup \Gamma_2 \right ) := \operatorname{Inv} \left ( \Gamma_1 \right ) + \operatorname{Inv} \left ( \Gamma_2 \right ) \, ,
	\end{equation}
	where \(\operatorname{Inv} \left ( \Gamma \right )\) is any of the invariants above and \(\Gamma_1, \Gamma_2 \in \GQ\).
\end{defn}

\enter

\begin{defn}[Weight of residues]
	Let \(\mathcal Q\) be a QFT with residue set \(\RQ\). We introduce a weight function
	\begin{equation}
		\omega \, : \quad \RQ \to \mathbb{Z} \, , \, , \quad r \mapsto \degp{\FR{r}} \, ,
	\end{equation}
	which maps a residue \(r \in \RQ\) to the degree of its corresponding Feynman rule \(\FR{r}\), viewed as a polynomial in momenta (or, in position space, derivatives).
\end{defn}

\enter

\begin{defn}[Superficial degree of divergence] \label{defn:sdd}
	Let \(\mathcal Q\) be a QFT with weighted residue set \((\RQ, \omega)\) and Feynman graph set \(\GQ\). We turn \(\GQ\) into a weighted set as well by extending \(\omega\) to the function
	\begin{equation}
		\omega \, : \quad \GQ \to \mathbb{Z} \, , \quad \Gamma \mapsto d \lambda \left ( \Gamma \right ) + \sum_{v \in V \left ( \Gamma \right )} \omega \left ( v \right ) + \sum_{e \in E \left ( \Gamma \right )} \omega \left ( e \right ) \, , \label{eqn:sdd}
	\end{equation}
	where \(d\) is the dimension of spacetime. Then, the weight \(\omega \left ( \Gamma \right )\) of a Feynman graph \(\Gamma\) is called its `superficial degree of divergence (SDD)'. A Feynman graph \(\Gamma\) is called superficially divergent if \(\sdd{\Gamma} \geq 0\) and superficially convergent if \(\sdd{\Gamma} < 0\). Finally, we set \(\sdd{\one} := 0\) for convenience.
\end{defn}

\enter

\begin{defn}[Set of superficially divergent subgraphs of a Feynman graph] \label{defn:set_of_divergent_subgraphs}
	Let \(\mathcal Q\) be a QFT with weighted Feynman graph set \((\GQ, \omega)\) and  \(\Gamma \in \GQ\) a Feynman graph. Then we denote by \(\DQ{\Gamma}\) the set of superficially divergent subgraphs of \(\Gamma\), i.e.
	\begin{subequations}
	\begin{align}
		\DQ{\Gamma} & := \set{\one \subseteq \gamma \subseteq \Gamma \, \left \vert \; \gamma = \bigsqcup_i \gamma_i \text{ with } \gamma_i \in \GQ \text{ and } \sdd{\gamma_i} \geq 0 \right .} \, ,
		\intertext{and by \(\mathcal{D}^\prime \left ( \Gamma \right )\) the set of proper divergent subgraphs}
		\DQprime{\Gamma} & := \set{ \gamma \in \DQ{\Gamma} \, \left \vert \; \one \subsetneq \gamma \subsetneq \Gamma \right .} \, .
	\end{align}
	\end{subequations}
	We remark that the condition \(\res{\gamma_i} \in \RQ\) for all divergent \(\gamma_i\) ensures the well-definedness of the renormalization Hopf algebra, cf.\ \cite[Subsection 3.3]{Prinz_2}.
\end{defn}

\enter

\begin{defn}[The (associated) renormalization Hopf algebra] \label{defn:renormalization_hopf_algebra}
	Given a QFT \(\Q\) with weighted Feynman graph set \((\GQ, \omega)\). Then the renormalization Hopf algebra is modeled on the \(\mathbb{Q}\)-vector space generated by 1PI Feynman graphs from the set \(\GQ\) and disjoint unions thereof. More precisely, the multiplication \(m \colon \HQ \otimes_\mathbb{Q} \HQ \to \HQ\) is given via disjoint union and the coproduct \(\Delta \colon \HQ \to \HQ \otimes_\mathbb{Q} \HQ\) is given via the decomposition of (products of) 1PI Feynman graphs into the sum of all pairs of divergent subgraphs with the remaining cographs:
	\begin{equation}
		\Delta \, : \quad \HQ \to \HQ \otimes_\mathbb{Q} \HQ \, , \quad \Gamma \mapsto \sum_{\gamma \in \DQ{\Gamma}} \gamma \otimes_\mathbb{Q} \Gamma / \gamma \, ,
	\end{equation}
	where the cograph \(\Gamma / \gamma\) is defined by shrinking the internal edges of \(\gamma\) in \(\Gamma\) to a new vertex for each connected component of \(\gamma\). Furthermore, the unit \(\one \colon \mathbb{Q} \hookrightarrow \HQ\) is given via a multiple of the empty graph and the counit \(\coone \colon \HQ \surject \mathbb{Q}\) is given via the map sending all non-empty graphs to zero and the empty graph to its prefactor.\footnote{We remark that technically the unit and counit of the algebra and their respective functions are separate objects, which can be conveniently identified.} Moreover, the antipode is recursively defined as the negative of the convolution product with itself and the projector onto the augmentation ideal, cf.\ \defnref{defn:convolution_product} and \defnref{defn:augmentation_ideal}, i.e.\ via the normalization \(S \left ( \one \right ) := \one\) and else as follows:
	\begin{equation}
		S \, : \quad \HQ \to \HQ \, , \quad \Gamma \mapsto - \Gamma - \sum_{\DQprime{\Gamma}} S \left ( \gamma \right ) \Gamma / \gamma \, .
	\end{equation}
	In the following, we will omit the ground field from the tensor product, i.e.\ set \(\otimes := \otimes_\mathbb{Q}\). Finally, we remark that especially in the context of quantum gauge theories the above construction can be ill-defined, which requires the notion of an `associated renormalization Hopf algebra' \cite[Subsection 3.3]{Prinz_2}. For the realm of this article, however, it suffices to know that there exists a well-defined Hopf algebra structure as defined above, to which we refer to as `(associated) renormalization Hopf algebra', see also the comment at the end of \defnref{defn:set_of_divergent_subgraphs}.
\end{defn}

\enter

\begin{defn}[Convolution product] \label{defn:convolution_product}
	Let \(\ring\) be a ring, \(A\) a \(\ring\)-algebra and \(C\) a \(\ring\)-coalgebra. Then, using the multiplication \(\mult_A\) on \(A\) and the comultiplication \(\Delta_C\) on \(C\), we can turn the \(\ring\)-module \(\text{Hom}_{\ring-\mathsf{Mod}} \left ( C , A \right )\) of \(\ring\)-linear maps from \(C\) to \(A\) into a \(\ring\)-algebra as well, by defining the convolution product \(\star\) as follows: Given \(f, g \in \text{Hom}_{\ring-\mathsf{Mod}} \left ( C , A \right )\), then we set
	\begin{equation}
		f \star g := \mult_A \circ \left ( f \otimes g \right ) \circ \Delta_C \, .
	\end{equation}
	Obviously, this definition extends trivially if \(A\) or \(C\) possesses additionally a bi- or Hopf algebra structure. It is commutative, if \(C\) is cocommutative and \(A\) is commutative. Finally, given a \(\ring\)-Hopf algebra \(H\), we remark that the algebra of endomorphisms \((\text{Hom}_{\ring-\mathsf{Mod}} \left ( H , H \right ), \star)\) is a group, with the antipode \(S\) being the \(\star\)-inverse to the identity morphism \(\operatorname{Id}_H\).
\end{defn}

\enter

\begin{defn}[Augmentation ideal] \label{defn:augmentation_ideal}
	Given a bi- or a Hopf algebra \(B\), then the kernel of the coidentity
	\begin{equation}
		\operatorname{Aug} \left ( \HQ \right ) := \operatorname{Ker} \big ( \mspace{1mu} \coone \mspace{1mu} \big )
	\end{equation}
	is an ideal, called the augmentation ideal. Additionally, we denote the projection map to it via \(\mathscr{A}\), i.e.\
	\begin{equation}
		\mathscr{A} \, : \quad \HQ \surject \operatorname{Aug} \left ( \HQ \right ) \subset \HQ \, , \quad \mathfrak{G} \mapsto \sum_{\substack{\set{\alpha_\text{s}, \mathfrak{G}_\text{s}} \in \SQ{\mathfrak{G}}\\\coone \left ( \mathfrak{G}_\text{s} \right ) = 0}} \alpha_\text{s} \mathfrak{G}_\text{s} \, ,
	\end{equation}
	where \(\SQ{\mathfrak{G}}\) denotes the set of summands of \(\mathfrak{G} \in \HQ\), cf.\ \defnref{defn:sets_summands_connected_components}.
\end{defn}

\enter

\begin{defn}[Basis decomposition of Hopf algebra elements] \label{defn:sets_summands_connected_components}
	Let \(\Q\) be a QFT, \(\GQ\) the set of its 1PI Feynman graphs and \(\HQ\) its (associated) renormalization Hopf algebra. Given an element \(\mathfrak{G} \in \HQ\), we are interested in its decomposition with respect to elements in the set \(\GQ\). Therefore, we denote by \(\SQ{\mathfrak{G}}\) the set of its summands, grouped into tuples of prefactors \(\alpha_\text{s} \in \mathbb{Q}\) (where we exclude the trivial case \(\alpha_\text{s} = 0\) if \(\mathfrak{G}_\text{s} \in \operatorname{Aug} \left ( \HQ \right )\)) and graphs \(\mathfrak{G}_\text{s} \in \HQ\) that can be disjoint unions, i.e.\ \(\mathfrak{G}_\text{s} = \bigsqcup_i \Gamma_i\) for \(\Gamma_i \in \GQ\),\footnote{This is actually the decomposition from \eqnref{eqn:decomposition_connected_components}.} such that
	\begin{equation}
		\mathfrak{G} \equiv \sum_{\set{\alpha_\text{s}, \mathfrak{G}_\text{s}} \in \SQ{\mathfrak{G}}} \alpha_\text{s} \mathfrak{G}_\text{s} \, .
	\end{equation}
	Additionally, we also write \(\mathfrak{G}_\text{s} \in \SQ{\mathfrak{G}}\) instead of \(\mathfrak{G}_\text{s} \in \set{\alpha_\text{s}, \mathfrak{G}_\text{s}} \in \SQ{\mathfrak{G}}\), if we are only interested in properties of the graph \(\mathfrak{G}_\text{s}\). Furthermore, given such a \(\mathfrak{G}_\text{s} \in \SQ{\mathfrak{G}}\), we denote by \(\CQ{\mathfrak{G}_\text{s}}\) the set of its connected components (where we exclude the identity \(\one \in \HQ\) if \(\mathfrak{G}_\text{s} \in \operatorname{Aug} \left ( \HQ \right )\)), such that
	\begin{equation} \label{eqn:decomposition_connected_components}
		\mathfrak{G}_\text{s} \equiv \prod_{\mathfrak{G}_\text{c} \in \CQ{\mathfrak{G}_\text{s}}} \mathfrak{G}_\text{c} \, .
	\end{equation}
	In particular, we have
	\begin{equation}
		\mathfrak{G} \equiv \sum_{\set{\alpha_\text{s}, \mathfrak{G}_\text{s}} \in \SQ{\mathfrak{G}}} \alpha_\text{s} \left ( \prod_{\mathfrak{G}_\text{c} \in \CQ{\mathfrak{G}_\text{s}}} \mathfrak{G}_\text{c} \right ) \, .
	\end{equation}
\end{defn}

\enter

\begin{defn}[Connectedness and gradings of the renormalization Hopf algebra] \label{defn:connectedness_gradings_renormalization_hopf_algebra}
	Given the situation of \defnref{defn:betti-numbers_and_multi-indices}, we construct the following three gradings on the renormalization Hopf algebra \(\HQ\): Let \(\mathfrak{G} \in \HQ\) be an element with \(\mathfrak{G}_\text{s} \in \SQ{\mathfrak{G}}\), we associate the following number and two multi-indices to \(\mathfrak{G}_\text{s}\):
	\begin{itemize}
		\item Loop-grading, denoted via \(L, l \in \mathbb{N}_0\), and given by
		\begin{equation}
		\operatorname{LoopGrd} \left ( \mathfrak{G}_\text{s} \right ) := \sum_{\mathfrak{G}_\text{c} \in \CQ{\mathfrak{G}_\text{s}}} \lambda \left ( \mathfrak{G}_\text{c} \right )
		\end{equation}
		\item Vertex-grading, denoted via \(\mathbf{V}, \mathbf{v} \in \ZvQ\), and given by
		\begin{equation}
		\vtxgrd{\mathfrak{G}_\text{s}} := \intvtx{\mathfrak{G}_\text{s}} - \extvtx{\mathfrak{G}_\text{s}} \label{eqn:resgrd}
		\end{equation}
		\item Coupling-grading, denoted via \(\mathbf{C}, \mathbf{c} \in \ZqQ\), and given by
		\begin{equation}
		\cplgrd{\mathfrak{G}_\text{s}} := \intcpl{\mathfrak{G}_\text{s}} - \extcpl{\mathfrak{G}_\text{s}}
		\end{equation}
	\end{itemize}
	In statements that are valid in any of these three gradings, we denote the grading function by \(\operatorname{Grd}\) and the gradings via \(\mathbf{G}\) and \(\mathbf{g}\). Furthermore, we denote the unit multi-index with respect to a vertex residue \(v \in \RQO\) or a coupling constant \(q \in \qQ\) via \(\mathbf{e}_v\) and \(\mathbf{e}_q\), respectively. Moreover, we denote the restriction of an object or an element to any of these three gradings via
	\begin{equation}
		\left ( \HQ \right )_\mathbf{G} := \eval{\HQ}_\mathbf{G} \, , \label{eqn:notation_restriction_grading}
	\end{equation}
	and omit the brackets, if no lower index is present. Clearly,
	\begin{equation}
		\left ( \mathcal{H}_{\Q} \right )_{L = 0} \cong \left ( \mathcal{H}_{\Q} \right )_{\mathbf{V} = \mathbf{0}} \cong \left ( \mathcal{H}_{\Q} \right )_{\mathbf{C} = \mathbf{0}} \cong \mathbb{Q} \, ,
	\end{equation}
	and thus \(\mathcal{H}_{\Q}\) is connected in all three gradings.
\end{defn}

\enter

\begin{rem}
	The three gradings from \defnref{defn:connectedness_gradings_renormalization_hopf_algebra} are further refinements of each other. In particular, the vertex-grading is equivalent to the coupling-grading if each vertex is associated with a unique coupling constant and it is furthermore equivalent to the loop-grading if the theory has only one vertex type. Additionally, the coupling-grading is equivalent to the loop-grading if the corresponding theory has only one coupling constant. We remark that both statements are due to the Euler identity, given in \eqnref{eqn:euler_characteristic}. Moreover, the numbers and multi-indices from \defnref{defn:betti-numbers_and_multi-indices} and the gradings from \defnref{defn:connectedness_gradings_renormalization_hopf_algebra} are compatible with the multiplication of \(\HQ\) via addition, but not with the addition of \(\HQ\), as summands can live in different gradings.
\end{rem}

\enter

\begin{defn}[Projection to divergent graphs] \label{defn:projection_divergent_graphs}
	Let \(\Q\) be a QFT, \(\GQ\) its 1PI Feynman graph set and \(\HQ\) its (associated) renormalization Hopf algebra. We define the projection map to divergent Feynman graphs via
	\begin{subequations}
	\begin{align}
		\Omega \, & : \quad \GQ \to \GQ \, , \quad \Gamma \mapsto \begin{cases} \Gamma & \text{if \(\sdd{\Gamma} \geq 0\)} \\ 0 & \text{else, i.e.\ \(\sdd{\Gamma} < 0\)} \end{cases}
		\intertext{and then extend it additively and multiplicatively to \(\HQ\), i.e.}
		\Omega \, & : \quad \HQ \to \HQ \, , \quad \mathfrak{G} \mapsto \sum_{\set{\alpha_\text{s}, \mathfrak{G}_\text{s}} \in \SQ{\mathfrak{G}}} \alpha_\text{s} \left ( \prod_{\mathfrak{G}_\text{c} \in \CQ{\mathfrak{G}_\text{s}}} \Omega \left ( \mathfrak{G}_\text{c} \right ) \right ) \, ,
	\end{align}
	\end{subequations}
	that is, we keep the summands of \(\mathfrak{G}\) only, if all of its connected components are divergent. Additionally, we also use the following shorthand-notation:
	\begin{subequations}
		\begin{align}
		\overline{\HQ} & := \operatorname{Im} \left ( \Omega \right )
		\intertext{and}
		\overline{\mathfrak{G}} & := \Omega \left ( \mathfrak{G} \right ) \, .
		\end{align}
	\end{subequations}
	We remark that this definition will be useful for combinatorial Green's functions \(\combgreen^r\), combinatorial charges \(\combcharge^v\) and products thereof in the context of Hopf subalgebras for multiplicative renormalization, cf.\ \sectionref{sec:coproduct_and_antipode_identities}.
\end{defn}

\enter

\begin{defn}[Superficially compatible grading] \label{defn:superficially_compatible_grading}
	Given the situation of \defnref{defn:sdd} and \defnref{defn:connectedness_gradings_renormalization_hopf_algebra}, a grading is called superficially compatible, if all Feynman graphs with a given residue and a given grading have the same superficial degree of divergence. Equivalently, the degree of divergence of a Feynman graph depends only on its residue and the given grading. This will be studied in \propref{prop:superficial_grade_compatibility}.
\end{defn}

\enter

\begin{defn}[(Restricted) combinatorial Green's functions] \label{def:combinatorial_greens_functions}
	Let \(\Q\) be a QFT, \(\AQ\) the set of its amplitudes and \(\mathcal{G}_{\Q}\) the set of its Feynman graphs. Given an amplitude \(r \in \AQ\), we set
	\begin{equation}
		\precombgreen^r := \sum_{\substack{\Gamma \in \mathcal{G}_{\Q}\\\res{\Gamma} = r}} \frac{1}{\sym{\Gamma}} \Gamma
	\end{equation}
	and then define the combinatorial Green's function with amplitude \(r\) as the following sum:
	\begin{equation} \label{eqn:combgreen}
		\combgreen^r := \begin{cases} \one + \precombgreen^r & \text{if \(r \in \RQ^{[0]}\)} \\ \one - \precombgreen^r & \text{if \(r \in \RQ^{[1]}\)} \\ \precombgreen^r & \text{else, i.e.\ \(r \in \mathcal{Q}_\Q\)}
	\end{cases}
	\end{equation}
	Furthermore, we denote the restriction of \(\combgreen^r\) to one of the gradings \(\mathbf{g}\) from \defnref{defn:connectedness_gradings_renormalization_hopf_algebra} via
	\begin{equation}
		\rescombgreen^r_\mathbf{g} := \eval{\combgreen^r}_{\mathbf{g}} \, .
	\end{equation}
\end{defn}

\enter

\begin{rem} \label{rem:different_conventions_restricted_greens_functions}
	We remark that restricted combinatorial Green's functions are in the literature often denoted via \(c^r_\mathbf{g}\) and differ by a minus sign from our definition. Our convention is such that they are given as the restriction of the complete combinatorial Green's function to the corresponding grading, which provides additional minus signs for propagator graphs.
\end{rem}

\enter

\begin{defn}[(Restricted) combinatorial charges]
	Let \(v \in \RQ^{[0]}\) be a vertex residue, then we define its combinatorial charge \(\combcharge^v\) via
	\begin{equation}
		\combcharge^v := \frac{\combgreen^v}{\prod_{e \in E \left ( v \right )} \sqrt{\combgreen^e}} \, ,
	\end{equation}
	where \(E \left ( v \right )\) denotes the set of all edges attached to the vertex \(v\). Furthermore, we denote the restriction of \(\combcharge^v\) to one of the gradings \(\mathbf{g}\) from \defnref{defn:connectedness_gradings_renormalization_hopf_algebra} via
	\begin{equation}
		\combcharge^v_\mathbf{g} := \eval{\combcharge^v}_{\mathbf{g}} \, .
	\end{equation}
\end{defn}

\enter

\begin{defn}[(Restricted) products of combinatorial charges] \label{defn:combinatorial_charges}
	Let \(\mathbf{v} \in \ZvQ\) be a multi-index of vertex residues. Then we define the product of combinatorial charges associated to \(\mathbf{v}\) via
	\begin{equation} \label{eqn:combinatorial_charges}
		\combcharge^\mathbf{v} := \prod_{k = 1}^{\mathfrak{v}_\Q} \left ( \combcharge^{v_k} \right )^{(\mathbf{v})_k} \, ,
	\end{equation}
	where \((\mathbf{v})_k\) denotes the \(k\)-th entry of \(\mathbf{v}\). In particular, given a vertex residue \(v \in \RQO\) and a natural number \(n \in \mathbb{N}_+\), we define the exponentiation of the combinatorial charge \(\combcharge^v\) by \(n\) via
	\begin{equation}
		\combcharge^{nv} := \left ( \combcharge^{v} \right )^n \, .
	\end{equation}
	Furthermore, we denote the restriction of \(\combcharge^\mathbf{v}\) to one of the gradings \(\mathbf{g}\) from \defnref{defn:connectedness_gradings_renormalization_hopf_algebra} via
	\begin{equation} \label{eqn:restricted_products_combinatorial_charges}
		\combcharge^\mathbf{v}_\mathbf{g} := \eval{\left ( \prod_{k = 1}^{\mathfrak{v}_\Q} \left ( \combcharge^{v_k} \right )^{(\mathbf{v})_k} \right )}_\mathbf{g} \, .
	\end{equation}
\end{defn}

\enter

\begin{defn}[Sets of combinatorial and physical charges, projection map] \label{defn:sets_of_coupling_constants}
	Let \(\Q\) be a QFT. Then we denote via \(\mathbf{Q}_\Q\) and \(\mathbf{q}_\Q\) the sets of combinatorial and physical charges, respectively. We associate to each vertex residue \(v \in \RQO\) a combinatorial charge and to each interaction monomial in the Lagrange density \(\mathcal{L}_\Q\) a (not necessarily distinct) physical coupling constant. Additionally, we define the set-theoretic projection map\footnote{This map is the set-theoretic restriction of the renormalized Feynman rules, which map combinatorial charges to Feynman integrals corresponding to renormalized physical charges.}
	\begin{equation}
		\operatorname{Cpl} \, : \quad \mathbf{Q}_\Q \surject \mathbf{q}_\Q \, , \quad \combcharge^v \mapsto \theta \left ( v \right ) \, ,
	\end{equation}
	where \(\theta \colon \RQO \to \qQ\) is the map from \eqnref{eqn:coupling-coloring_function}.
\end{defn}

\enter

\begin{lem} \label{lem:v-h-e-sets_and_r-v}
	Given a Feynman graph \(\Gamma \in \GQ\), the sets \(V \left ( \Gamma \right )\) and \(E \left ( \Gamma \right )\) from \defnref{defn:feynman_graphs} and \(H \left ( \Gamma \right )\) and \(C \left ( \Gamma \right )\) from \defnref{defn:half-edges_corollas_external-vertex-residue-set}, viewed as multisets over \(\RQ\), depend only on its residue \(\res{\Gamma}\) and its vertex-grading multi-index \(\vtxgrd{\Gamma}\). In particular, we obtain well-defined sets \(V \left ( r, \mathbf{v} \right )\), \(E \left ( r, \mathbf{v} \right )\), \(H \left ( r, \mathbf{v} \right )\) and \(C \left ( r, \mathbf{v} \right )\), such that we have \(V \left ( r, \mathbf{v} \right ) \cong V \left ( \Gamma \right )\), \(E \left ( r, \mathbf{v} \right ) \cong E \left ( \Gamma \right )\), \(H \left ( r, \mathbf{v} \right ) \cong H \left ( \Gamma \right )\) and \(C \left ( r, \mathbf{v} \right ) \cong C \left ( \Gamma \right )\) as multisets over \(\RQ\), for all \(\Gamma \in \GQ\) with \(\res{\Gamma} = r\) and \(\vtxgrd{\Gamma} = \mathbf{v}\).
\end{lem}

\begin{proof}
	Given \(\Gamma \in \GQ\), then by definition its vertex set \(V \left ( \Gamma \right )\) is a multiset over \(\RQO\), using \(\tau\) from \eqnref{eqn:residue-coloring_function}. Thus it is uniquely characterized via its internal residue multi-index \(\intvtx{\Gamma}\), as it displays the multiplicity of each vertex residue \(r_v \in \RQO\) in the vertex set \(V \left ( \Gamma \right )\). Furthermore, we can reconstruct \(\intvtx{\Gamma}\) from \(\res{\Gamma}\) and \(\vtxgrd{\Gamma}\) using the definition, \eqnref{eqn:resgrd}, i.e.\
	\begin{equation}
		V \left ( r, \mathbf{v} \right ) \cong \vtxgrd{\Gamma} + \extvtx{\Gamma} \, ,
	\end{equation}
	while noting that \(\extvtx{\Gamma}\) is given for connected Feynman graphs \(\Gamma \in \GQ\) with \(\res{\Gamma} \in \RQO\) as the multi-indices having a one for the corresponding vertex residue and zeros else, i.e.\
	\begin{equation}
		\left ( \extvtx{\Gamma} \right )_j = \begin{cases} 1 & \text{if \(\res{\Gamma} = v_j \in \RQO\)} \\ 0 & \text{else} \end{cases} \, ,
	\end{equation}
	and for Feynman graphs \(\Gamma \in \GQ\) with \(\res{\Gamma} \in \big ( \AQ \setminus \RQO \big )\) as the zero-multi-index, i.e.\
	\begin{equation}
		\extvtx{\Gamma} = \mathbf{0} \, .
	\end{equation}
	Thus we have have shown that the set \(V \left ( r, \mathbf{v} \right )\) is well-defined as a multiset over \(\RQO\). Moreover, we can obtain the half-edge set \(H \left ( r, \mathbf{v} \right )\) from \(\res{\Gamma}\), \(V \left ( r, \mathbf{v} \right )\) and the coloring function \(\tau\) via
	\begin{equation}
		H \left ( r, \mathbf{v} \right ) := \set{h_v \in \bigsqcup_{v \in V \left ( r, \mathbf{v} \right )} H \left ( v, \tau \left ( v \right ) \right )} \setminus H \left ( \res{\Gamma}, \tau \left ( \res{\Gamma} \right ) \right ) \, ,
	\end{equation}
	where \(\tau\) indicates the vertex-type of \(v \in V \left ( r, \mathbf{v} \right )\), i.e.\ which edge-types are attachable to it. Then we denote via \(H \left ( v, \tau \left ( v \right ) \right )\) the set of all such pairings \(h_v \cong (v, e)\), take its disjoint union and then remove the set of external half-edges of \(\Gamma\). Additionally, we obtain the edge set \(E \left ( r, \mathbf{v} \right )\) as a multiset over \(\RQI\) from the half-edge set \(H \left ( r, \mathbf{v} \right )\) as follows: We use an equivalence relation \(\sim\) which identifies two half-edges to a single edge, if they are of the same particle type, i.e.\ \(h_1 \sim h_2\) if \(\tau \left ( e_1 \right ) = \tau \left ( e_2 \right )\), and then consider the quotient
	\begin{equation}
		E \left ( r, \mathbf{v} \right ) := H \left ( r, \mathbf{v} \right ) / \sim \, .
	\end{equation}
	We remark that there are in general many possibilities to define \(\sim\), but the resulting multisets are isomorphic, hence it suffices to take an arbitrary choice. In particular, one such choice is the involution \(\iota\) from \defnref{defn:half-edges_corollas_external-vertex-residue-set}. Finally, we obtain the corolla set \(C \left ( r, \mathbf{v} \right )\) from the vertex set \(V \left ( r, \mathbf{v} \right )\) and the half-edge set \(H \left ( r, \mathbf{v} \right )\) by simply associating to each vertex the set of half-edges attached to it, i.e.\
	\begin{equation}
		C \left ( r, \mathbf{v} \right ) := \set{\left . c_v \cong \left ( v, \set{h_v \in H \left ( r, \mathbf{v} \right )} \right ) \right \vert v \in V \left ( r, \mathbf{v} \right )} \, ,
	\end{equation}
	which completes the proof.
\end{proof}

\enter

\begin{defn}[Set of superficially divergent insertable graphs for a Feynman graph] \label{defn:set_of_divergent_insertable_graphs_graphs}
	Let \(\Q\) be a QFT and \(\Gamma \in \mathcal{G}_{\Q}\) a Feynman graph of \(\Q\). Then we denote by \(\IQ{\Gamma}\) the set of superficially divergent graphs that are insertable into \(\Gamma\), i.e.\footnote{We remark that we have \(\one \in \IQ{\Gamma}\) for all \(\Gamma \in \GQ\).}
	\begin{equation}
	\begin{split}
		\IQ{\Gamma} := \left \{ \gamma \in \HQ \, \vphantom{\gamma \in \HQ \, \vert \; \extvtx{\gamma} \leq \intvtx{\Gamma} \text{ and } \sdd{\gamma_\text{c}} \geq 0 \text{ for all } \gamma_\text{c} \in \CQ{\gamma} \phantom{\vert} \text{ and } \res{\gamma_\text{p}} \in E \left ( \Gamma \right ) \text{ for all } \gamma_\text{p} \in \mathcal{P} \left ( \gamma \right )} \right . & \left \vert \; \extvtx{\gamma} \leq \intvtx{\Gamma} \text{ and } \sdd{\gamma_\text{c}} \geq 0 \text{ for all } \gamma_\text{c} \in \CQ{\gamma} \vphantom{\gamma \in \HQ \, \vert \; \extvtx{\gamma} \leq \intvtx{\Gamma} \text{ and } \sdd{\gamma_\text{c}} \geq 0 \text{ for all } \gamma_\text{c} \in \CQ{\gamma} \phantom{\vert} \text{ and } \res{\gamma_\text{p}} \in E \left ( \Gamma \right ) \text{ for all } \gamma_\text{p} \in \mathcal{P} \left ( \gamma \right )} \right . \\ 
		& \phantom{\vert} \left . \; \text{and } \res{\gamma_\text{p}} \in E \left ( \Gamma \right ) \text{ for all } \gamma_\text{p} \in \mathcal{P} \left ( \gamma \right ) \vphantom{\gamma \in \HQ \, \vert \; \extvtx{\gamma} \leq \intvtx{\Gamma} \text{ and } \sdd{\gamma_\text{c}} \geq 0 \text{ for all } \gamma_\text{c} \in \CQ{\gamma} \phantom{\vert} \text{ and } \res{\gamma_\text{p}} \in E \left ( \Gamma \right ) \text{ for all } \gamma_\text{p} \in \mathcal{P} \left ( \gamma \right )} \right \} \, ,
	\end{split}
	\end{equation}
	where \(\mathcal{P} \left ( \gamma \right ) \subseteq \CQ{\gamma}\) denotes the set of connected components of \(\gamma\) which are propagator graphs.
\end{defn}

\enter

\begin{defn}[Insertion factors] \label{defn:ins-aut_ins_insrv}
	Let \(\Q\) be a QFT, \(\GQ\) its Feynman graph set and \(\HQ\) its (associated) renormalization Hopf algebra. Given two Feynman graphs \(\Gamma, \Gamma^\prime \in \mathcal{G}_{\Q}\) and an element in the Hopf algebra \(\gamma \in \HQ\), we want to characterize possible insertions. To this end, we define the following four combinatorial factors:
	\begin{itemize}
		\item \(\ins{\gamma}{\Gamma}\) denotes the number of ways to insert \(\gamma\) into \(\Gamma\)
		\item \(\insaut{\gamma}{\Gamma}{\Gamma^\prime}\) denotes the number of ways to insert \(\gamma\) into \(\Gamma\), such that the insertion is automorphic to \(\Gamma^\prime\)
		\item \(\insrr{\gamma}\) denotes the number of ways to insert \(\gamma\) into a Feynman graph with residue \(r\) and vertex-grading multi-index \(\mathbf{v}\), which is well-defined due to \lemref{lem:v-h-e-sets_and_r-v}
		\item \(\isoemb{\gamma}{\Gamma}\) denotes the number of non-trivial isomorphic embeddings of \(\gamma\) as a subgraph of \(\Gamma\)
	\end{itemize}
	We remark that these numbers are zero, if either \(\gamma\) is not insertable into \(\Gamma\), i.e.\ \(\gamma \notin \IQ{\Gamma}\), if there is no insertion possible which is automorphic to \(\Gamma^\prime\) or if there is no isomorphic embedding possible. Finally, we set for all \(\Gamma \in \GQ\)
	\begin{equation}
		\insaut{\one}{\Gamma}{\Gamma} = \ins{\one}{\Gamma} = \insrr{\one} = \isoemb{\one}{\Gamma} := 1 \, .
	\end{equation}
\end{defn}

\enter

\begin{prop} \label{prop:isomorphism_i}
	Given the situation of \defnref{defn:set_of_divergent_insertable_graphs_graphs}, we have for all Feynman graphs \(\Gamma \in \GQ\)
	\begin{equation}
		\sum_{\gamma \in \IQ{\Gamma}} \frac{\ins{\gamma}{\Gamma}}{\sym{\gamma}} \gamma = \frac{\prod_{v \in V \left ( \Gamma \right )} \overline{\combgreen}^v}{\prod_{e \in E \left ( \Gamma \right )} \overline{\combgreen}^e} \, .
	\end{equation}
\end{prop}

\begin{proof}
	We can insert in each vertex \(v \in V \left ( \Gamma \right )\) at most one superficially divergent vertex correction \(\gamma^v\) with \(\res{\gamma^v} = v\), i.e. a summand of \(\overline{\combgreen}^v\). Furthermore, we can insert in each edge \(e \in E \left ( \Gamma \right )\) arbitrary many superficially divergent edge corrections \(\gamma^e = \prod_i \gamma^e_i\) with \(\operatorname{Res} \big ( \gamma^e_i \big ) = e\) for all \(i\), i.e.\ a summand of \(\textfrac{1}{\overline{\combgreen}^e}\), where the fraction is understood as the formal geometric series \(\textfrac{1}{\left ( 1-x \right )} \equiv \sum_{k = 0}^\infty x^k\).\footnote{We remark that this viewpoint is the reason for the minus sign in the definition of combinatorial Green's function for propagators, i.e.\ \eqnref{eqn:combgreen} of \defnref{def:combinatorial_greens_functions}.} Finally, the prefactor \(\ins{\gamma}{\Gamma}\) corresponds to the multiplicity of similar \(\RQ\)-colored vertices and edges of \(\Gamma\), using \(\tau\) from \eqnref{eqn:residue-coloring_function}.
\end{proof}

\enter

\begin{defn}[Set of superficially divergent insertable graphs for residue and vertex-grading] \label{defn:set_of_divergent_insertable_graphs_grading-residue}
	Let \(\Q\) be a QFT, \(\AQ\) its amplitude set and \(\HQ\) its (associated) renormalization Hopf algebra. Given an amplitude \(r \in \AQ\), a vertex-grading multi-index \(\mathbf{v} \in \ZvQ\) with \(\mathbf{v} \neq \mathbf{0}\) and a Feynman graph \(\Gamma \in \GQ\) with \(\res{\Gamma} = r\) and \(\vtxgrd{\Gamma} = \mathbf{v}\). Then we define the set \(\IQrv\) of superficially divergent graphs insertable into Feynman graphs with residue \(r\) and vertex-grading \(\mathbf{v}\) via\footnote{We remark that we have \(\one \in \IQrv\) for all \(r \in \AQ\) and \(\mathbf{v} \in \ZvQ\) with \(\mathbf{v} \neq \mathbf{0}\).}
	\begin{equation}
		\IQrv := \IQ{\Gamma} \, ,
	\end{equation}
	which is well-defined due to \colref{col:set_of_divergent_insertable_graphs_grading-residue}.
\end{defn}

\enter

\begin{col} \label{col:set_of_divergent_insertable_graphs_grading-residue}
	Given the situation of \defnref{defn:set_of_divergent_insertable_graphs_grading-residue}, the set \(\IQrv\) satisfies
	\begin{equation} \label{eqn:set_of_divergent_insertable_graphs_grading-residue}
		\sum_{\gamma \in \IQrv} \frac{\insrr{\gamma}}{\sym{\gamma}} \gamma = \frac{\prod_{v \in V \left ( r, \mathbf{v} \right )} \overline{\combgreen}^v}{\prod_{e \in E \left ( r, \mathbf{v} \right )} \overline{\combgreen}^e} \, ,
	\end{equation}
	and is thus in particular well-defined.
\end{col}

\begin{proof}
	This follows directly from \lemref{lem:v-h-e-sets_and_r-v} and \propref{prop:isomorphism_i}.
\end{proof}

\enter

\begin{prop} \label{prop:isomorphic_insertable_graph_sets}
	Given the situation of \defnref{defn:set_of_divergent_insertable_graphs_grading-residue}, we have for all amplitudes \(r \in \AQ\) and vertex-grading multi-indices \(\mathbf{v} \in \ZvQ\)
	\begin{equation}
		\sum_{\gamma \in \IQrv} \frac{\insrr{\gamma}}{\sym{\gamma}} \gamma = \begin{cases} \overline{\combgreen}^r \overline{\combcharge}^\mathbf{v} & \text{if \(r \in \RQ\)} \\ \prod_{e \in E \left ( r \right )} \sqrt{\overline{\combgreen}^e} \overline{\combcharge}^\mathbf{v} & \text{else, i.e.\ \(r \in \mathcal{Q}_\Q\)} \end{cases} \, ,
	\end{equation}
	where \(E \left ( r \right )\) denotes the set of edges attached to \(r\) and the square-root is defined via the formal series \(\sqrt{x} \equiv \sum_{k = 0}^\infty \binom{\textfrac{1}{2}}{k} \left ( x - 1 \right )^k\).
\end{prop}

\begin{proof}
	The numerator of the right hand side of \eqnref{eqn:set_of_divergent_insertable_graphs_grading-residue} of \colref{col:set_of_divergent_insertable_graphs_grading-residue} can be expressed as follows:\footnote{The two cases emerge due to the vertex-grading, which treats Feynman graphs with vertex residues differently in order to obtain a valid grading of the renormalization Hopf algebra, cf.\ \eqnref{eqn:resgrd} of \defnref{defn:connectedness_gradings_renormalization_hopf_algebra}.}
	\begin{subequations}
	\begin{equation} \label{eqn:combgreen_vertexset}
		\prod_{v \in V \left ( r, \mathbf{v} \right )} \overline{\combgreen}^v = \begin{cases} \overline{\combgreen}^r \overline{\combgreen}^\mathbf{v} & \text{if \(r \in \RQO\)} \\ \overline{\combgreen}^\mathbf{v} & \text{else, i.e.\ \(r \in \left ( \AQ \setminus \RQO \right )\)} \end{cases} \, ,
	\end{equation}
	where the notation \(\overline{\combgreen}^\mathbf{v} := \prod_{k = 1}^{\mathfrak{v}_\Q} \big ( \overline{\combgreen}^{v_k} \big )^{\mathbf{v}_k}\) is analogous to \eqnref{eqn:combinatorial_charges} of \defnref{defn:combinatorial_charges}. Furthermore, the denominator of the right hand side of \eqnref{eqn:set_of_divergent_insertable_graphs_grading-residue} of \colref{col:set_of_divergent_insertable_graphs_grading-residue} can be expressed as follows:
	\begin{equation} \label{eqn:combgreen_edgeset}
	\begin{split}
		\frac{\one}{\prod_{e \in E \left ( r, \mathbf{v} \right )} \overline{\combgreen}^e} & = \begin{cases} \dfrac{\one}{\prod_{v \in V \left ( r, \mathbf{v} \right )} \left ( \prod_{e \in E \left ( v \right )} \sqrt{\overline{\combgreen}^e} \right )} & \text{if \(r \in \RQO\)} \\[5pt] \dfrac{\prod_{{e_1} \in E \left ( r \right )} \sqrt{\overline{\combgreen}^{e_1}}}{\prod_{v \in V \left ( r, \mathbf{v} \right )} \left ( \prod_{{e_2} \in E \left ( v \right )} \sqrt{\overline{\combgreen}^{e_2}} \right )} & \text{else, i.e.\ \(r \in \left ( \AQ \setminus \RQO \right )\)} \end{cases} \\
		& = \begin{cases} \dfrac{\overline{\combcharge}^\mathbf{v}}{\overline{\combgreen}^\mathbf{v}} & \text{if \(r \in \RQO\)} \\[10pt] \dfrac{\overline{\combgreen}^r \overline{\combcharge}^\mathbf{v}}{\overline{\combgreen}^\mathbf{v}} & \text{if \(r \in \RQI\)} \\[5pt] \dfrac{\prod_{e \in E \left ( r \right )} \sqrt{\overline{\combgreen}^e} \overline{\combcharge}^\mathbf{v}}{\overline{\combgreen}^\mathbf{v}} & \text{else, i.e.\ \(r \in \mathcal{Q}_\Q\)} \end{cases}
	\end{split}
	\end{equation}
	\end{subequations}
	Multiplying \eqnref{eqn:combgreen_vertexset} with \eqnref{eqn:combgreen_edgeset}, we obtain
	\begin{equation}
		\frac{\prod_{v \in V \left ( r, \mathbf{v} \right )} \overline{\combgreen}^v}{\prod_{e \in E \left ( r, \mathbf{v} \right )} \overline{\combgreen}^e} = \begin{cases} \overline{\combgreen}^r \overline{\combcharge}^\mathbf{v} & \text{if \(r \in \RQ\)} \\ \prod_{e \in E \left ( r \right )} \sqrt{\overline{\combgreen}^e} \overline{\combcharge}^\mathbf{v} & \text{else, i.e.\ \(r \in \mathcal{Q}_\Q\)} \end{cases} \, .
	\end{equation}
	Finally, the prefactor \(\insrr{\gamma}\) corresponds to the multiplicity of similar \(\RQ\)-colored vertices and edges of Feynman graphs with residue \(r\) and vertex-grading multi-index \(\mathbf{v}\), using \(\tau\) from \eqnref{eqn:residue-coloring_function}.
\end{proof}

\enter

\begin{lem}[{\cite[Lemma 12]{vSuijlekom_QCD}}] \label{lem:sym-factors_and_ins-factors}
	Given the situation of \defnref{defn:automorphisms_of_feynman_graphs} and \defnref{defn:ins-aut_ins_insrv}, we have for all Feynman graphs \(\Gamma \in \GQ\) and their corresponding subgraphs \(\one \subseteq \gamma \subseteq \Gamma\)
	\begin{equation}
		\frac{\sym{\gamma} \sym{\Gamma / \gamma}}{\sym{\Gamma}} = \frac{\insaut{\gamma}{\Gamma / \gamma}{\Gamma}}{\isoemb{\gamma}{\Gamma}} \, .
	\end{equation}
\end{lem}

\begin{proof}
	Let \(\Gamma \in \GQ\) be a Feynman graph. Then, by definition, we have
	\begin{equation}
		\sym{\Gamma} = \# \operatorname{Aut} \left ( \Gamma \right ) \, ,
	\end{equation}
	where the automorphisms are fixing the external leg structure by definition, cf.\ \defnref{defn:automorphisms_of_feynman_graphs}. Thus, for a given subgraph \(\gamma \subseteq \Gamma\), we have
	\begin{equation}
		\sym{\gamma} \sym{\Gamma / \gamma} = \# \operatorname{Aut} \left ( \gamma \right ) \# \operatorname{Aut} \left ( \Gamma / \gamma \right ) \, ,
	\end{equation}
	which counts all automorphisms of \(\Gamma / \gamma\) times those of \(\gamma\), fixing both their external legs. Thus, comparing to \(\sym{\Gamma}\), the following two things can appear: The automorphism group \(\operatorname{Aut} \left ( \Gamma \right )\) might contain automorphisms which exchange non-trivial isomorphic embeddings \(\gamma, \gamma^\prime \subseteq \Gamma\) and can thus contain automorphisms exceeding the set \(\operatorname{Aut} \left ( \gamma \right ) \cup \operatorname{Aut} \left ( \Gamma / \gamma \right )\). Contrary, the quotient graph \(\Gamma / \gamma\) might possess symmetries which get spoiled after the insertion of \(\gamma\) into \(\Gamma / \gamma\). These two possibilities are reflected by the quotient \(\insaut{\gamma}{\Gamma / \gamma}{\Gamma} / \isoemb{\gamma}{\Gamma}\), as it counts the number of equivalent insertions of \(\gamma\) into \(\Gamma / \gamma\) automorphic to \(\Gamma\) modulo the additional symmetries that might appear, cf.\ \defnref{defn:ins-aut_ins_insrv}. Thus we obtain
	\begin{equation}
		\frac{\sym{\gamma} \sym{\Gamma / \gamma}}{\sym{\Gamma}} = \frac{\insaut{\gamma}{\Gamma / \gamma}{\Gamma}}{\isoemb{\gamma}{\Gamma}} \, ,
	\end{equation}
	as claimed.
\end{proof}

\enter

\begin{defn}[Algebra of formal (Feynman) integral expressions] \label{defn:formal_feynman_integral_expressions}
	Let \(\field\) be a field and \(\mathcal{E}\) the \(\field\)-vector space generated by the set of formal integral expressions, that is pairs \((D,I)\), where \(D\) is a domain and \(I\) a differential form on it. Addition is then declared via
	\begin{equation}
		(D_1,I_1) + (D_2,I_2) := (D_1 \oplus D_2, I_1 \oplus 0_2 + 0_1 \oplus I_2) \, ,
	\end{equation}
	where \(0_i\) is the zero differential form on the domain \(D_i\), and scalar multiplication is declared via
	\begin{equation}
		k (D,I) := (D,kI)
	\end{equation}
	for \(k \in \field\). Furthermore, we turn \(\mathcal{E}\) into an algebra by declaring the multiplication via
	\begin{equation}
		(D_1,I_1) \times (D_2,I_2) := (D_1 \otimes_\field D_2, I_1 \otimes_\field I_2) \, , \label{eqn:multiplication_map_ffie}
	\end{equation}
	which we call \(\mu\). Moreover, we address the name `formal integral expression' by defining the evaluation character (i.e.\ algebra morphism)
	\begin{equation}
		\operatorname{Int} \, : \quad \mathcal{E}_\text{Fin} \to \mathbb{C} \, , \quad (D,I) \mapsto \int_D I \, ,
	\end{equation}
	where \(\mathcal{E}_\text{Fin} \subset \mathcal{E}\) is the subalgebra where the evaluation map is finite and thus well-defined. In particular, we fix the normal subgroups \(\mathbf{1}_{\mathcal{E}_\text{Fin}} := \operatorname{Int}^{-1} \left ( 1 \right ) \subset \mathcal{E}_\text{Fin}\) and \(\mathbf{1}_{\mathcal{E}} := \iota \left ( \mathbf{1}_{\mathcal{E}_\text{Fin}} \right )\), where \(\iota \colon \mathcal{E}_\text{Fin} \hookrightarrow \mathcal{E}\) is the natural inclusion map. Both of these groups consist of formal integral expressions \((D,I)\) with \(\operatorname{Int} \left ( D,I \right ) = 1\). Therefore, we will treat \(\mathbf{1}_{\mathcal{E}}\) as the equivalence class of `units' on \(\mathcal{E}\). Finally, given a QFT \(\Q\), we define the algebra of its formal Feynman integral expressions as follows: We set \(\field := \mathbb{Q}\) and restrict the allowed domains \(D\) and differential forms \(I\) according to the chosen Feynman integral representation (position space, momentum space, parametric space, etc.).
\end{defn}

\enter

\begin{rem}
	 The setup of \defnref{defn:formal_feynman_integral_expressions} allows us in particular to address ill-defined integral expressions by externalizing the integration process.
\end{rem}

\enter

\begin{defn}[Feynman rules, regularization and renormalization schemes] \label{defn:fr_reg_ren_counterterm}
	Let \(\Q\) be a QFT, \(\HQ\) its (associated) renormalization Hopf algebra and \(\EQ\) its algebra of formal Feynman integral expressions. Then we define its Feynman rules as the following character (i.e.\ algebra morphism)
	\begin{equation}
		\Phi \, : \quad \HQ \to \EQ \, , \quad \Gamma \mapsto (D_\Gamma, I_\Gamma) \, ,
	\end{equation}
	where \((D_\Gamma, I_\Gamma)\) is the formal Feynman integral expression for the Feynman graph \(\Gamma\). Furthermore, we introduce a regularization scheme \(\mathscr{E}\) as a map\footnote{There exist renormalization schemes, such as kinematic renormalization schemes, that are well-defined without a previous regularization step. These can be seen as embedded into our framework by simply setting \(\mathscr{E} := \operatorname{Id}_{\EQ}\) and considering it as the natural inclusion of \(\EQ\) into \(\EQ^\varepsilon\).}
	\begin{equation} \label{eqn:regularization_scheme}
		\mathscr{E} \, : \quad \EQ \hookrightarrow \EQ^\varepsilon \, , \quad (D,I) \mapsto \left ( D,I_\mathscr{E} \left ( \varepsilon \right ) \right ) := \left ( D,\sum_{i = 0}^\infty I_i \, \varepsilon^i \right ) \, ,
	\end{equation}
	where \(\EQ^\varepsilon := \EQ [ [ \varepsilon ] ] \supset \EQ\) and the coefficients of the Taylor series are differential forms \(I_i\) on \(D\). Additionally, the regulated formal integral expressions \(\left ( D,I_\mathscr{E} \left ( \varepsilon \right ) \right )\) are subject to the boundary condition \(I \left ( 0 \right ) \equiv I\), which is equivalent to \(I_0 := I\), and the integrability condition \(\operatorname{Int} \left ( D,I_\mathscr{E} \left ( \varepsilon \right ) \right ) < \infty\), for all \(\varepsilon \in J\) with \(J \subseteq [0,\infty)\) a non-empty interval. We then set the regularized Feynman rules as the map
	\begin{equation}
		\regFR \, : \quad \HQ \to {\EQ}^\varepsilon \, , \quad \Gamma \mapsto \left ( \mathscr{E} \circ \Phi \right ) \left ( \Gamma, \varepsilon \right ) \, .
	\end{equation}
	Moreover, we introduce a renormalization scheme as a linear map\footnote{Sometimes, if convenient, we view \(\mathscr{R}\) also as endomorphism on \(\EQ^\varepsilon\) with image \(\EQ^\varepsilon_-\) and cokernel \(\EQ^\varepsilon_+\).}
	\begin{equation} \label{eqn:renormalization_scheme}
		\mathscr{R} \, : \quad \EQ^\varepsilon \surject \EQ^\varepsilon_- \, , \quad \left ( D,I_\mathscr{E} \left ( \varepsilon \right ) \right ) \mapsto \begin{cases} (D,0_D) & \text{if \(\left ( D,I_\mathscr{E} \left ( \varepsilon \right ) \right ) \in \operatorname{Ker} \left ( \mathscr{R} \right )\)} \\ \left ( D,I_{\mathscr{E}, \mathscr{R}} \left ( \varepsilon \right ) \right ) & \text{else} \end{cases} \, ,
	\end{equation}
	where \({\EQ}^\varepsilon_- := \operatorname{Im} \left ( \mathscr{R} \right ) \subset \EQ^\varepsilon\) and \(0_D\) is the zero differential form on \(D\), for all \(\varepsilon \in \mathbb{R}\). Additionally, to ensure locality of the counterterm, \(\mathscr{R}\) needs to be a Rota-Baxter operator of weight \(\lambda = -1\), i.e.\ fulfill
	\begin{equation}
		\mu \circ \left ( \mathscr{R} \otimes \mathscr{R} \right ) + \mathscr{R} \circ \mu = \mathscr{R} \circ \mu \circ \left ( \mathscr{R} \otimes \id + \id \otimes \mathscr{R} \right ) \, ,
	\end{equation}
	where \(\mu\) denotes the multiplication on \(\EQ^\varepsilon\) (and by abuse of notation also on \(\EQ^\varepsilon_-\) via restriction) from \eqnref{eqn:multiplication_map_ffie}. In particular, \((\EQ^\varepsilon, \mathscr{R})\) is a Rota-Baxter algebra of weight \(\lambda = -1\) and \(\mathscr{R}\) induces the splitting
	\begin{equation}
		\EQ^\varepsilon \cong \EQ^\varepsilon_+ \oplus \EQ^\varepsilon_-
	\end{equation}
	with \(\EQ^\varepsilon_+ := \operatorname{CoKer} \left ( \mathscr{R} \right )\) and \(\EQ^\varepsilon_- := \operatorname{Im} \left ( \mathscr{R} \right )\). Then we can introduce the counterterm map \(\countertermsymbol\), sometimes also called `twisted antipode', recursively via the normalization \(\counterterm{\one} \in \mathbf{1}_{\EQ^\varepsilon}\) and
	\begin{equation}
		\countertermsymbol \, : \quad \operatorname{Aug} \left ( \HQ \right ) \to \EQ^\varepsilon_- \, , \quad \Gamma \mapsto - \renscheme{\countertermsymbol \star \left ( \regFR \circ \mathscr{A} \right )} \left ( \Gamma \right )
	\end{equation}
	else, where \(\mathscr{A} \colon \HQ \surject \operatorname{Aug} \left ( \HQ \right )\) is the projector onto the augmentation ideal from \defnref{defn:augmentation_ideal}. Next we define renormalized Feynman rules via
	\begin{equation}
		\Phi_\mathscr{R} \, : \quad \HQ \to \EQ^\varepsilon_+ \, , \quad \Gamma \mapsto \underset{\varepsilon \mapsto 0}{\operatorname{Lim}} \left ( \countertermsymbol \star \Phi \right ) \left ( \Gamma \right ) \, ,
	\end{equation}
	where the corresponding formal Feynman integral expression is well-defined in the limit \(\varepsilon \mapsto 0\), if the cokernel \(\operatorname{CoKer} \left ( \mathscr{R} \right )\) consists only of convergent formal Feynman integral expressions, cf.\ \lemref{lem:finite_renormalization_schemes}. We remark that the renormalized Feynman rules \(\Phi_\mathscr{R}\) and the counterterm map \(\countertermsymbol\) correspond to the algebraic Birkhoff decomposition of the Feynman rules \(\Phi\) with respect to the renormalization scheme \(\mathscr{R}\), as was first observed in \cite{Connes_Kreimer_0} and e.g.\ reviewed in \cite{Guo,Panzer}. Finally, we remark that the above discussion can be also lifted to the algebra of meromorphic functions \(\mathcal{M}^\varepsilon := \mathbb{C} \big [ \varepsilon^{-1}, \varepsilon \big ] \big ] \), if a suitable regularization scheme \(\mathscr{E}\) is chosen,\footnote{In the sense that the integrated expressions do not contain essential singularities in the regulator.} by setting
	\begin{equation}
		\widetilde{\mathscr{E}} \, : \quad \EQ \to \mathcal{M}^\varepsilon \, , \quad \left ( D,I_\mathscr{E} \left ( \varepsilon \right ) \right ) \mapsto f_\mathscr{E} \left ( \varepsilon \right ) := \int_D \eval{\left ( I_\mathscr{E} \left ( \varepsilon \right ) \right )}_{\varepsilon \in J} \, ,
	\end{equation}
	for fixed external momentum configurations away from Landau singularities. Then we can proceed as before by setting a renormalization scheme as a linear map
	\begin{equation}
		\widetilde{\mathscr{R}} \, : \quad \mathcal{M}^\varepsilon \surject \mathcal{M}^\varepsilon_- \, , \quad f_\mathscr{E} \left ( \varepsilon \right ) \mapsto \begin{cases} 0 & \text{if \(f_\mathscr{E} \left ( \varepsilon \right ) \in \operatorname{Ker} \big ( \widetilde{\mathscr{R}} \big )\)} \\ f_{\mathscr{E}, \mathscr{R}} \left ( \varepsilon \right ) & \text{else} \end{cases} \, ,
	\end{equation}
	where \(\mathcal{M}^\varepsilon_- := \operatorname{Im} \big ( \widetilde{\mathscr{R}} \big ) \subset \mathcal{M}^\varepsilon\), and the rest analogously.
\end{defn}

\enter

\begin{defn}[Hopf subalgebras for multiplicative renormalization] \label{defn:hopf_subalgebras_renormalization_hopf_algebra}
	Let \(\Q\) be a QFT, \(\RQ\) its weighted residue set, \(\HQ\) its (associated) renormalization Hopf algebra and \(\rescombgreen^r_\mathbf{G} \in \HQ\) a restricted combinatorial Green's function, where \(\mathbf{G}\) and \(\mathbf{g}\) denote one of the gradings from \defnref{defn:connectedness_gradings_renormalization_hopf_algebra}. We are interested in Hopf subalgebras which correspond to multiplicative renormalization, i.e.\ Hopf subalgebras of \(\HQ\) such that the coproduct factors over restricted combinatorial Green's functions as follows:
	\begin{equation}
		\Delta \left ( \rescombgreen^r_{\mathbf{G}} \right ) = \sum_{\mathbf{g}} \mathfrak{P}_{\mathbf{g}} \left ( \rescombgreen^r_{\mathbf{G}} \right ) \otimes \rescombgreen^r_{\mathbf{G} - \mathbf{g}} \, , \label{eqn:hopf_subalgebras_multi-index}
	\end{equation}
	where \(\mathfrak{P}_{\mathbf{g}} \left ( \rescombgreen^r_{\mathbf{G}} \right ) \in \HQ\) is a polynomial in graphs such that each summand has multi-index \(\mathbf{g}\).\footnote{There exist closed expressions for the polynomials \(\mathfrak{P}_{\mathbf{g}} \left ( \rescombgreen^r_{\mathbf{G}} \right )\) as we will see in \sectionref{sec:coproduct_and_antipode_identities}, in particular \propref{prop:coproduct_greensfunctions},, which were first introduced in \cite{Yeats_PhD}.}
\end{defn}

\enter

\begin{rem} \label{rem:hopf_subalgebras_renormalization_hopf_algebra}
	Given the situation of \defnref{defn:fr_reg_ren_counterterm} and assume that the (associated) renormalization Hopf algebra \(\HQ\) possesses Hopf subalgebras in the sense of \defnref{defn:hopf_subalgebras_renormalization_hopf_algebra}. Then we can calculate the \(Z\)-factor for a given residue \(r \in \RQ\) via
	\begin{equation}
		Z^r_{\mathscr{E}, \mathscr{R}} \left ( \varepsilon \right ) := \counterterm{\combgreen^r} \, .
	\end{equation}
	More details in this direction can be found in \cite{Panzer,vSuijlekom_Multiplicative} (with a different notation). Additionally, we remark that the existence of the Hopf subalgebras from \defnref{defn:hopf_subalgebras_renormalization_hopf_algebra} depends crucially on the grading \(\mathbf{g}\). In particular, for the loop-grading these Hopf subalgebras exist if and only if the QFT has only one fundamental interaction, i.e.\ \(\RQO\) is a singleton. Furthermore, they exist for the coupling-grading if and only if the QFT has for each fundamental interaction a different coupling constant, i.e.\ \(\# \QQ = \# \qQ\). Finally, they exist always for the vertex-grading, as we will see in \propref{prop:coproduct_greensfunctions}, cf.\ \sectionsaref{sec:coproduct_and_antipode_identities}{sec:quantum_gauge_symmetries_and_subdivergences}.
\end{rem}

\enter

\begin{lem} \label{lem:finite_renormalization_schemes}
	The image of the renormalized Feynman rules \(\operatorname{Im} \left ( \Phi_\mathscr{R} \right )\) consists of convergent integral expressions, if the cokernel \(\operatorname{CoKer} \left ( \mathscr{R} \right )\) of the corresponding renormalization scheme \(\mathscr{R} \in \operatorname{End} \left ( \EQ^\varepsilon \right )\) does.
\end{lem}

\begin{proof}
	The theorem about the algebraic Birkhoff decomposition, first observed in \cite{Connes_Kreimer_0}, states in this context that
	\begin{align}
		\Phi_\mathscr{R} & \, : \quad \HQ \to \EQ^\varepsilon_+
		\intertext{and}
		\countertermsymbol & \, : \quad \HQ \to \EQ^\varepsilon_- \, ,
	\end{align}
	where \(\EQ^\varepsilon_+ := \operatorname{CoKer} \left ( \mathscr{R} \right )\) and \(\EQ^\varepsilon_- := \operatorname{Im} \left ( \mathscr{R} \right )\), and thus \(\operatorname{Im} \left ( \Phi_\mathscr{R} \right )\) consists of finite integral expressions, if \(\operatorname{CoKer} \left ( \mathscr{R} \right )\) does.
\end{proof}

\enter

\begin{defn}[Proper renormalization schemes] \label{defn:proper_renormalization_scheme}
	A renormalization scheme \(\mathscr{R} \in \operatorname{End} \left ( \EQ^\varepsilon \right )\) is called proper, if both its kernel \(\operatorname{Ker} \left ( \mathscr{R} \right )\) and its cokernel \(\operatorname{CoKer} \left ( \mathscr{R} \right )\) consist only of convergent integral expressions.\footnote{We allow, as an exception, superficially divergent graphs to be in the kernel of \(\mathscr{R}\), if they would lead to an ill-defined coalgebra structure on the renormalization Hopf algebra. See \cite[Subsection 3.3]{Prinz_2} for a detailed discussion on this matter.} In particular, we demand that
	\begin{equation}
		\operatorname{Im} \left ( \Phi \circ \Omega \right ) \subseteq \operatorname{CoIm} \left ( \mathscr{R} \right ) \, ,
	\end{equation}
	i.e.\ the image of superficially divergent graphs under the Feynman rules is a subset of the coimage of a proper renormalization scheme.
\end{defn}

\enter

\begin{rem}
	\defnref{defn:proper_renormalization_scheme} is motivated by the fact that in physics we want renormalization schemes to produce finite, i.e.\ integrable, renormalized Feynman rules and furthermore preserve the locality of the theory, i.e.\ remove divergences of Feynman integrals via contributions from themselves.
\end{rem}

\section{A superficial argument} \label{sec:a_superficial_argument}

In this section we study combinatorial properties of the `superficial degree of divergence (SDD)' from \defnref{defn:sdd}. This integer,  combinatorially associated to each Feynman graph \(\Gamma \in \GQ\), provides a measure of the ultraviolet divergence of the corresponding Feynman integral. In fact, a result of Weinberg states that the ultraviolet divergence of the (formal) Feynman integral \(\Phi \left ( \Gamma \right )\) is bounded by a polynomial of degree \(\operatorname{Max}_{\one \subsetneq \gamma \subseteq \Gamma} \sdd{\gamma}\), where the maximum is considered over all non-empty 1PI subgraphs and \(\omega \colon \GQ \to \mathbb{Z}\) denotes the SDD of \(\gamma\) \cite{Weinberg}. More precisely, the corresponding Feynman integral converges if \(\sdd{\Gamma} < 0\) and \(\sdd{\gamma} < 0\) for all non-empty 1PI subgraphs \(\one \subsetneq \gamma \subset \Gamma\). We start this section by providing an alternative definition of the SDD in terms of weights of corollas in \defnref{defn:asdd} and \lemref{lem:asdd}. With this criterion on hand, we show in \thmref{thm:asdd} that the SDD of Feynman graphs with a fixed residue depends affine-linearly on their vertex-grading. Then we give in \colref{col:weights_of_corollas_and_renormalizability} an alternative characterization to the classification of QFTs into (super-/non-)renormalizability via weights of corollas. Building upon this result, we introduce the notion of a `cograph-divergent QFT' in \defnref{defn:cograph-divergent}, which is shown in \propref{prop:proj_div_graphs_coprod} to be the obstacle for the compatibility of coproduct identities with the projection to divergent Feynman graphs from \defnref{defn:projection_divergent_graphs}. Next we study the superficial compatibility of the coupling-grading and loop-grading in \propref{prop:superficial_grade_compatibility}. We complete this section by showing that (effective) Quantum General Relativity coupled to the Standard Model (QGR-SM) satisfies both, i.e.\ is cograph-divergent and has superficially compatible coupling-grading in \colref{col:qgr-sm_is_cograph-divergent} and \colref{col:qgr-sm_is_sqgsc}. The results in this section are in particular useful for the following sections, where we want to state our results not only for the renormalizable case, but include also the more involved super- and non-renormalizable cases, or even mixes thereof. Finally, this allows us to apply our results to QGR-SM.

\enter

\begin{defn}[Weights of corollas] \label{defn:asdd}
	Let \(\Q\) be a QFT and \((\RQ, \omega)\) its weighted residue set. Given a vertex residue \(v \in \RQO\), we define the weight of its corolla via
	\begin{equation} \label{eqn:pre-asdd}
		\csdd{v} \equiv \sdd{c_v} := \sdd{v} + \frac{1}{2} \sum_{e \in E \left ( v \right )} \sdd{e} \, .
	\end{equation}
\end{defn}

\enter

\begin{lem} \label{lem:asdd}
	Given the situation of \defnref{defn:asdd}, the superficial degree of divergence of a Feynman graph \(\Gamma \in \GQ\) can be equivalently calculated via\footnote{We remark that the equivalence holds only for non-trivial graphs, i.e.\ graphs with at least one vertex.}
	\begin{equation}
		\omega \, : \quad \GQ \to \mathbb{Z} \, , \quad \Gamma \mapsto d \bettii{\Gamma} + \sum_{v \in V \left ( \Gamma \right )} \csdd{v} - \frac{1}{2} \sum_{e \in E \left ( \res{\Gamma} \right )} \omega \left ( e \right ) \, . \label{eqn:prealternative_superficial_degree_of_divergence}
	\end{equation}
\end{lem}

\begin{proof}
	This follows directly from the combination of \eqnref{eqn:sdd} from \defnref{defn:sdd} and \eqnref{eqn:pre-asdd} from \defnref{defn:asdd}.
\end{proof}

\enter

\begin{thm}[Superficial degree of divergence via residue and vertex-grading] \label{thm:asdd}
	Given the situation of \lemref{lem:asdd}, the superficial degree of divergence of a Feynman graph \(\Gamma \in \GQ\) can be decomposed as follows:\footnote{The function \(\sigma\) can be expressed equivalently as a linear functional in the coarser gradings, if they are superficially compatible, cf.\ \defnref{defn:superficially_compatible_grading}.}
	\begin{subequations}
	\begin{align}
		\sdd{\Gamma} & \equiv \rsdd{\Gamma} + \ssdd{\Gamma} \, , \label{eqn:asdd}
		\intertext{where \(\rsdd{\Gamma}\) depends only on \(\res{\Gamma}\) and \(\ssdd{\Gamma}\) depends only on \(\vtxgrd{\Gamma}\). In particular, we have:}
		\rsdd{\Gamma} & := \begin{cases} \omega \left ( \res{\Gamma} \right ) & \text{if \(\res{\Gamma} \in \RQO\)} \\ d - \dfrac{1}{2} \sum_{e \in E \left ( \res{\Gamma} \right )} \omega \left ( e \right ) & \text{else, i.e.\ \(\res{\Gamma} \in \left ( \AQ \setminus \RQO \right )\)} \end{cases} \label{eqn:rsdd}
		\intertext{and}
		\ssdd{\Gamma} & := \sum_{i = 1}^{\mathfrak{v}_\Q} \left ( d \left ( \frac{1}{2} \val{v_i} - 1 \right ) + \csdd{v_i} \right ) \left ( \vtxgrd{\Gamma} \right )_i \label{eqn:ssdd}
	\end{align}
	\end{subequations}
	Notably, \(\ssdd{\Gamma}\) is linear in \(\vtxgrd{\Gamma}\) and thus \(\sdd{\Gamma}\) is affine linear in \(\vtxgrd{\Gamma}\).
\end{thm}

\begin{proof}
	We start with \eqnref{eqn:prealternative_superficial_degree_of_divergence} from \defnref{defn:asdd}: First we rewrite the loop number using the Euler characteristic\footnote{For the Euler characteristic, \eqnref{eqn:euler_characteristic}, we need to either ignore the external (half-)edges or assume that they are attached to external vertices, as in \defnref{defn:feynman_graphs}, and adjust the loop-number accordingly.}
	\begin{equation}
		\bettii{\Gamma} = \bettio{\Gamma} - \# V \left ( \Gamma \right ) + \# E \left ( \Gamma \right ) \, , \label{eqn:euler_characteristic}
	\end{equation}
	where \(\bettio{\Gamma} = 1\), as \(\Gamma \in \GQ\) is connected. Then we express \(\# V \left ( \Gamma \right )\) in terms of \(\vtxgrd{\Gamma}\) and \(\res{\Gamma}\) via
	\begin{equation} \label{eqn:vertex-set_and_grading}
	\begin{split}
		\# V \left ( \Gamma \right ) & = \sum_{i = 1}^{\mathfrak{v}_\Q} \left ( \intvtx{\Gamma} \right )_i\\
		& = \sum_{i = 1}^{\mathfrak{v}_\Q} \left ( \left ( \vtxgrd{\Gamma} \right )_i + \left ( \extvtx{\Gamma} \right )_i \right )\\
		& = \sum_{i = 1}^{\mathfrak{v}_\Q} \left ( \left ( \vtxgrd{\Gamma} \right )_i \right ) + \delta_{\res{\Gamma} \in \RQO} \, ,
	\end{split}
	\end{equation}
	where in the last equality we have used that \(\Gamma \in \GQ\) is connected by setting
	\begin{equation}
		\delta_{\res{\Gamma} \in \RQO} = \begin{cases} 1 & \text{if \(\res{\Gamma} \in \RQO\)} \\ 0 & \text{else, i.e.\ \(\res{\Gamma} \in \left ( \AQ \setminus \RQO \right )\)} \end{cases} \, .
	\end{equation}
	Furthermore, we express \(\# E \left ( \Gamma \right )\) in terms of \(\vtxgrd{\Gamma}\) and the valences of the corresponding vertices via
	\begin{equation} \label{eqn:edge-set_and_grading}
	\begin{split}
		\# E \left ( \Gamma \right ) & = \frac{1}{2} \sum_{i = 1}^{\mathfrak{v}_\Q} \val{v_i} \left ( \intvtx{\Gamma} \right )_i - \frac{1}{2} \sum_{i = 1}^{\mathfrak{v}_\Q} \val{v_i} \left ( \extvtx{\Gamma} \right )_i\\
		& = \frac{1}{2} \sum_{i = 1}^{\mathfrak{v}_\Q} \val{v_i} \left ( \vtxgrd{\Gamma} \right )_i \, .
	\end{split}
	\end{equation}
	Thus, combining \eqnssaref{eqn:euler_characteristic}{eqn:vertex-set_and_grading}{eqn:edge-set_and_grading}, we obtain
	\begin{equation}
		\bettii{\Gamma} = 1 - \delta_{\res{\Gamma} \in \RQO} + \frac{1}{2} \sum_{i = 1}^{\mathfrak{v}_\Q} \left ( \val{v_i} - 2 \right ) \left ( \vtxgrd{\Gamma} \right )_i \, . \label{eqn:bettii_residue_vertex-grading}
	\end{equation}
	We proceed by rewriting
	\begin{equation}
		\sum_{v \in V \left ( \Gamma \right )} \csdd{v} = \sum_{i = 1}^{\mathfrak{v}_\Q} \csdd{v_i} \left ( \vtxgrd{\Gamma} \right )_i \, .
	\end{equation}
	Finally, plugging the above results into \eqnref{eqn:prealternative_superficial_degree_of_divergence} from \defnref{defn:asdd} yields the claimed result.
\end{proof}

\enter

\begin{col}[Weights of corollas and renormalizability] \label{col:weights_of_corollas_and_renormalizability}
	Given the situation of \thmref{thm:asdd} and a vertex residue \(v \in \RQO\). Then its corolla \(c_v\) is renormalizable if its weight is
	\begin{equation}
		\csdd{v} \equiv d \left ( 1 - \frac{1}{2} \val{v} \right ) \, , \label{eqn:weight_corolla}
	\end{equation}
	non-renormalizable if its weight is bigger and super-renormalizable if its weight is smaller. In particular, the QFT \(\Q\) is renormalizable, non-renormalizable or super-renormalizable if all of its corollas are.
\end{col}

\begin{proof}
	Before presenting the actual argument, we recall that a QFT \(\Q\) is called renormalizable, if, for a fixed residue \(r\), the superficial degree of divergence is independent of the grading, non-renormalizable if the superficial degree of divergence increases with increasing grading and super-renormalizable if the superficial degree of divergence decreases with increasing grading. Using \eqnref{eqn:ssdd} from \thmref{thm:asdd}, we directly obtain the claimed bound
	\begin{equation}
	\csdd{v} = d \left ( 1 - \frac{1}{2} \val{v} \right ) \, ,
	\end{equation}
	as for this value the superficial degree of divergence of a Feynman graph is independent of the corolla \(c_v\), it is positively affected if its weight is bigger and it is negatively affected if its weight is smaller.
\end{proof}

\enter

\begin{col} \label{col:ssdd_decomposition}
	Given the situation of \thmref{thm:asdd}, the dependence of \(\ssdd{\Gamma}\) on \(\vtxgrd{\Gamma}\) can be furthermore refined by decomposing
	\begin{equation}
		\ssdd{\Gamma} \equiv \ssddn{\Gamma} + \ssddr{\Gamma} + \ssdds{\Gamma} \label{eqn:ssdd_decomposition}
	\end{equation}
	where
	\begin{subequations}
	\begin{align}
		\ssddn{\Gamma} & \equiv \left \vert \ssddn{\Gamma} \right \vert \, , \\
		\ssddr{\Gamma} & \equiv 0
		\intertext{and}
		\ssdds{\Gamma} & \equiv - \left \vert \ssdds{\Gamma} \right \vert \, .
	\end{align}
	\end{subequations}
	In particular, \(\ssddn{\Gamma}\) depends only on the non-renormalizable corollas, \(\ssddr{\Gamma}\) depends only on the renormalizable corollas and \(\ssdds{\Gamma}\) depends only on the super-renormalizable corollas.
\end{col}

\begin{proof}
	This follows directly from \thmref{thm:asdd} and \colref{col:weights_of_corollas_and_renormalizability}.
\end{proof}

\enter

\begin{defn}[Cograph-divergent QFTs] \label{defn:cograph-divergent}
	Let \(\Q\) be a QFT with residue set \(\RQ\) and weighted Feynman graph set \((\GQ, \omega)\). We call \(\Q\) cograph-divergent, if for each superficially divergent Feynman graph \(\Gamma \in \overline{\GQ}\) and its superficially divergent subgraphs \(\gamma \in \DQ{\Gamma}\), the corresponding cographs \(\Gamma / \gamma\) are also all superficially divergent, i.e.\ fulfill \(\sdd{\Gamma / \gamma} \geq 0\).
\end{defn}

\enter

\begin{rem}
	\defnref{defn:cograph-divergent} is trivially satisfied for super-renormalizable and renormalizable QFTs. However, it is the obstacle for the compatibility of coproduct identities with the projection to divergent graphs for non-renormalizable QFTs, as will be shown in \propref{prop:proj_div_graphs_coprod}. Furthermore, we show in \colref{col:qgr-sm_is_cograph-divergent} that (effective) Quantum General Relativity coupled to the Standard Model is cograph-divergent.
\end{rem}

\enter

\begin{prop} \label{prop:proj_div_graphs_coprod}
	Let \(\Q\) be a QFT with (associated) renormalization Hopf algebra \(\HQ\). Then coproduct identities are compatible with the projection to divergent graphs from \defnref{defn:projection_divergent_graphs} if and only if \(\Q\) is cograph-divergent: More precisely, given the identity (we assume \(\overline{\mathfrak{G}}_\mathbf{V} \neq 0\))\footnote{We remark that \(\mathfrak{h}_{\mathbf{V} - \mathbf{v}} \equiv \overline{\mathfrak{h}_{\mathbf{V} - \mathbf{v}}}\) by the definition of the coproduct in (associated) renormalization Hopf algebras, cf.\ \defnref{defn:renormalization_hopf_algebra}.}
	\begin{align}
		\D{\mathfrak{G}_\mathbf{V}} & = \sum_\mathbf{v} \mathfrak{h}_{\mathbf{V} - \mathbf{v}} \otimes \mathfrak{H}_\mathbf{v}
		\intertext{then this implies}
		\Delta \big ( \overline{\mathfrak{G}}_\mathbf{V} \big ) & = \sum_\mathbf{v} \mathfrak{h}_{\mathbf{V} - \mathbf{v}} \otimes \overline{\mathfrak{H}}_\mathbf{v}\label{eqn:div-coprod}
	\end{align}
	if and only if \(\Q\) is cograph-divergent.
\end{prop}

\begin{proof}
	Since the coproduct is additive, we can without loss of generality assume that \(\mathfrak{G}_\mathbf{V}\) consists of only one summand. Thus,
	\begin{equation}
		\mathfrak{G}_\mathbf{V} \equiv \alpha \prod_i \Gamma_i
	\end{equation}
	with \(\alpha \in \mathbb{Q}\) and \(\Gamma_i \in \GQ\) are 1PI Feynman graphs whose vertex-grading adds up to \(\mathbf{V}\). Furthermore, we have either \(\overline{\mathfrak{G}}_\mathbf{V} = \mathfrak{G}_\mathbf{V}\) or \(\overline{\mathfrak{G}}_\mathbf{V} = 0\) by definition, cf.\ \defnref{defn:projection_divergent_graphs}, where the latter case is excluded by assumption. Using the linearity and multiplicativity of the coproduct,
	\begin{equation}
		\D{\alpha \prod_i \Gamma_i} = \alpha \prod_i \D{\Gamma_i} \, ,
	\end{equation}
	we can reduce our calculation to the 1PI Feynman graphs \(\Gamma_i\) via
	\begin{equation}
		\D{\Gamma_i} = \sum_{\gamma_i \in \DQ{\Gamma_i}} \gamma_i \otimes \Gamma_i / \gamma_i \, .
	\end{equation}
	If \(\Q\) is cograph-divergent, we obtain \(\overline{\Gamma_i / \gamma_i} \equiv \Gamma_i / \gamma_i\), which concludes the proof.
\end{proof}

\enter

\begin{lem} \label{lem:cograph-divergence_criterion}
	Let \(\Q\) be a QFT consisting only of non-renormalizable and renormalizable corollas. Then \(\Q\) is cograph-divergent if \(\sdd{v} \geq 0\) for all vertex-residues \(v \in \RQO\).
\end{lem}

\begin{proof}
	This is a direct consequence of \eqnref{eqn:asdd} from \thmref{thm:asdd} together with \eqnref{eqn:ssdd_decomposition} from \colref{col:ssdd_decomposition}: We obtain \(\sdd{\Gamma} \geq 0\) for all Feynman graphs \(\Gamma \in \GQ\) with \(\res{\Gamma} \in \RQ\), if \(\rsdd{r} \geq 0\) for all residues \(r \in \RQ\). This follows from the fact that the contribution from \(\sigma\) is non-negative by the assumption that \(\Q\) consists only of non-renormalizable and renormalizable corollas. Furthermore, as the weight of an edge is negative in our conventions, \(\rsdd{e} \geq 0\) is trivially satisfied for all edge-residues \(e \in \RQI\). Thus, the property \(\rsdd{r} \geq 0\) reduces to \(\rsdd{v} \equiv \sdd{v} \geq 0\) for all vertex residues \(v \in \RQO\). As the residue of a Feynman graph \(\Gamma\) agrees with any of its cographs, i.e.\ \(\res{\Gamma} = \res{\Gamma / \gamma}\) for any subgraph \(\gamma \subset \Gamma\), we obtain the claimed statement.
\end{proof}

\enter

\begin{col} \label{col:qgr-sm_is_cograph-divergent}
	(Effective) Quantum General Relativity coupled to the Standard Model in 4 dimensions of spacetime is cograph-divergent.
\end{col}

\begin{proof}
	(Effective) Quantum General Relativity (QGR) coupled to the Standard Model (SM) consists of a non-renormalizable (QGR) and a renormalizable (SM) sub-QFT. Furthermore, its vertices \(v \in \mathcal{R}^{[0]}_\text{QGR-SM}\) are either independent, linear dependent or quadratic dependent on momenta, and thus satisfy \(\sdd{v} \geq 0\). Hence \lemref{lem:cograph-divergence_criterion} applies, which concludes the proof.
\end{proof}

\enter

\begin{rem} \label{rem:proj_div_graphs_non-cogr-div-qft}
	A direct consequence of \thmref{thm:asdd} is that, for a fixed residue \(r \in \RQ\), the zero locus of the superficial degree of divergence \(\omega\) is a hyperplane in \(\ZvQ\). More precisely, we define the function
	\begin{equation}
		\omega^r \, : \quad \ZvQ \to \mathbb{Z} \, , \quad \mathbf{v} \mapsto \rho \left ( r \right ) + \sigma \left ( \mathbf{v} \right )
	\end{equation}
	and set \(\mathscr{H}^r_0 := \left ( \omega^r \right )^{-1} \left ( 0 \right ) \subset \ZvQ\). Accordingly, we decompose
	\begin{equation}
		\ZvQ \cong \mathscr{H}^r_+ \sqcup \mathscr{H}^r_0 \sqcup \mathscr{H}^r_-
	\end{equation}
	as sets, where \(\mathscr{H}^r_\pm\) are defined such that \(\eval[1]{\omega^r}_{\mathscr{H}^r_+} > 0\) and \(\eval[1]{\omega^r}_{\mathscr{H}^r_-} < 0\). We apply this knowledge to \propref{prop:proj_div_graphs_coprod}: Suppose \(\Q\) is a non-cograph-divergent QFT and consider for simplicity that \(\mathfrak{G}_\mathbf{V}\) is given as the sum over 1PI Feynman graphs with a fixed residue \(r\). Then, the coproduct identity
	\begin{align}
		\D{\mathfrak{G}_\mathbf{V}} & = \sum_{\mathbf{v} \in \ZvQ} \mathfrak{h}_{\mathbf{V} - \mathbf{v}} \otimes \mathfrak{H}_\mathbf{v}
		\intertext{implies only the following splitted identity}
		\Delta \big ( \overline{\mathfrak{G}}_\mathbf{V} \big ) & = \sum_{\mathbf{v} \in \left ( \mathscr{H}^r_+ \sqcup \mathscr{H}^r_0 \right )} \mathfrak{h}_{\mathbf{V} - \mathbf{v}} \otimes \overline{\mathfrak{H}}_\mathbf{v} + \sum_{\mathbf{v}^\prime \in \mathscr{H}^r_-} \mathfrak{h}_{\mathbf{V} - \mathbf{v}^\prime} \otimes \mathfrak{H}_{\mathbf{v}^\prime} \, ,
	\end{align}
	where the elements \(\mathfrak{H}_{\mathbf{v}^\prime}\) consist precisely of the convergent cographs.
\end{rem}

\enter

\begin{rem}
	Given the situation of \propref{prop:proj_div_graphs_coprod} and \remref{rem:proj_div_graphs_non-cogr-div-qft}, analogous statements also hold in the case of the antipode due to \lemref{lem:coproduct_and_antipode_identities}.
\end{rem}

\enter

\begin{prop}[Superficial grade compatibility] \label{prop:superficial_grade_compatibility}
	Given the situation of \defnref{defn:superficially_compatible_grading}, the coupling-grading is superficially compatible if and only if
	\begin{equation} \label{eqn:coupling-grading_superficial-compatibility}
		m \csdd{v} = n \csdd{w}
	\end{equation}
	holds for all \(v, w \in \RQO\) and \(m, n \in \mathbb{N}_+\) with \(\theta \left ( v \right )^m = \theta \left ( w \right )^n\). Furthermore, the loop-grading is superficially compatible if and only if
	\begin{equation} \label{eqn:loop-grading_superficial-compatibility}
		\frac{1}{\left ( \val{v} - 2 \right )} \csdd{v} = \frac{1}{\left ( \val{w} - 2 \right )} \csdd{w}
	\end{equation}
	holds for all \(v, w \in \RQO\), where \(\val{v}\) denotes the valence of \(v\).
\end{prop}

\begin{proof}
	We start with the superficial compatibility of the coupling-grading: Two vertex residues \(v, w \in \RQO\) contribute to the same coupling-grading, if and only if there exist natural numbers \(m, n \in \mathbb{N}_+\) such that \(\theta \left ( v \right )^m = \theta \left ( w \right )^n\). Thus, \eqnref{eqn:coupling-grading_superficial-compatibility} is the condition for the superficial compatibility of the coupling-grading and we proceed to the superficial-compatibility of the loop-grading: Loop-grading is only sensible if just one coupling constant is present. In this case, different vertex residues might be scaled via different powers of the same coupling constant. Applying the previous result, loop-grading is superficially compatible if these powers depend only on the valences of the vertices, i.e.\ two vertex residues \(v, w \in \RQO\) are considered equivalent, if and only if \(\theta \left ( v \right )^{\val{w} - 2} = \theta \left ( w \right )^{\val{v} - 2}\). Thus, \eqnref{eqn:loop-grading_superficial-compatibility} is the condition for the superficial compatibility of the loop-grading, which concludes the proof.
\end{proof}

\enter

\begin{rem}
	Having a superficially compatible coupling-grading is a necessary condition for the validity of quantum gauge symmetries on the level of Feynman rules --- however it is not sufficient. This will be discussed in \sectionref{sec:quantum_gauge_symmetries_and_renormalized_feynman_rules}.
\end{rem}

\enter

\begin{col} \label{col:qgr-sm_is_sqgsc}
	(Effective) Quantum General Relativity coupled to the Standard Model has a superficially compatible coupling-grading.
\end{col}

\begin{proof}
	We start by considering the pure gravitational part, i.e.\ gravitons and graviton-ghosts: The Feynman rules of (effective) Quantum General Relativity (QGR) are such that each vertex \(v \in \mathcal{R}_\text{QGR}^{[0]}\) has weight \(\sdd{v} = 2\) and each edge \(e \in \mathcal{R}_\text{QGR}^{[1]}\) has weight \(\sdd{e} = -2\), as the corresponding Feynman rules are quadratic and inverse quadratic in momenta, respectively. Thus, the corolla-weight of a vertex \(v \in \mathcal{R}_\text{QGR}^{[0]}\) is
	\begin{equation}
		\csdd{v} = 2 - \val{v} \, .
	\end{equation}
	Applying \eqnref{eqn:loop-grading_superficial-compatibility} from \propref{prop:superficial_grade_compatibility} yields
	\begin{equation}
	\begin{split}
		\frac{1}{\left ( \val{v} - 2 \right )} \csdd{v} & = \frac{2 - \val{v}}{\val{v} - 2}\\
		& \equiv -1 \, ,
	\end{split}
	\end{equation}
	which shows that QGR has even a superficially compatible loop-grading. Furthermore, the pure Standard Model (SM) part is renormalizable. Thus, applying \colref{col:weights_of_corollas_and_renormalizability}, the corolla-weight of a vertex \(v \in \mathcal{R}_\text{SM}^{[0]}\) is
	\begin{equation}
	\begin{split}
		\frac{1}{\left ( \val{v} - 2 \right )} \csdd{v} & = \frac{d \left ( 1 - \textfrac{1}{2} \val{v} \right )}{\left ( \val{v} - 2 \right )}\\
		& \equiv - \frac{d}{2} \, ,
	\end{split}
	\end{equation}
	which shows that also the SM has even a superficially compatible loop-grading. Finally, we consider the mixed part, i.e.\ SM residues with a positive number of gravitons attached to it.\footnote{We remark that this includes also edges from the Standard Model, as they become vertices when attaching graviton half-edges to them.} It follows from the corresponding Feynman rules that the weights of these corollas depend only on the SM residue and are thus independent of the number of gravitons attached to them. Conversely, increasing the number of gravitons of such a vertex by gluing a three-valent graviton tree (consisting of a vertex and a propagator) also leaves the weight of the corolla unchanged (as the net difference is \(2 - 2 = 0\)), which finishes the proof.
\end{proof}

\section{Coproduct and antipode identities} \label{sec:coproduct_and_antipode_identities}

In this section we state and generalize coproduct identities, known in the literature for the case of renormalizable QFTs with a single coupling constant \cite{vSuijlekom_QCD,vSuijlekom_BV,Yeats_PhD,Borinsky_Feyngen}. We reprove these identities and generalize them to cover the more involved cases of super- and non-renormalizable QFTs and QFTs with several vertex residues. Since coproduct identities imply recursive antipode identities by \lemref{lem:coproduct_and_antipode_identities}, we will in the following only discuss the coproduct cases.

\enter

\begin{lem}[Coproduct and antipode identities] \label{lem:coproduct_and_antipode_identities}
	Let \(\Q\) be a QFT with (associated) renormalization Hopf algebra \(\HQ\) and grading (multi-)indices \(\mathbf{G}\) and \(\mathbf{g}\). Then coproduct identities are equivalent to recursive antipode identities as follows: Given the identity
	\begin{align}
		\D{\mathfrak{G}} & = \sum_\mathbf{g} \mathfrak{h}_\mathbf{g} \otimes \mathfrak{H}_\mathbf{g} \, ,
		\intertext{then this is equivalent to the recursive identity}
		\antipode{\mathfrak{G}} & = - \sum_{\mathbf{g}} \antipode{\mathfrak{h}_\mathbf{g}} \mathscr{A} \left ( \mathfrak{H}_\mathbf{g} \right ) \, ,
	\end{align}
	where \(\mathscr{A}\) is the projector onto the augmentation ideal, cf.\ \defnref{defn:augmentation_ideal}.
\end{lem}

\begin{proof}
	This follows immediately from the recursive definition of the antipode on generators \(\Gamma \in \HQ\) via
	\begin{equation}
	\begin{split}
		\antipode{\Gamma} & := - \sum_{\gamma \in \DQ{\Gamma}} \antipode{\gamma} \mathscr{A} \left ( \Gamma / \gamma \right ) \\
		& \phantom{:} \equiv - \left ( S \star \mathscr{A} \right ) \left ( \Gamma \right ) \, ,
	\end{split}
	\end{equation}
	which is then linearly and multiplicatively extended to all of \(\HQ\), such that
	\begin{equation}
		S \star \operatorname{Id}_{\HQ} \equiv \one \circ \coone \equiv \operatorname{Id}_{\HQ} \star S
	\end{equation}
	holds, where \(\operatorname{Id}_{\HQ} \in \operatorname{End} \left ( \HQ \right )\) is the identity-endomorphism of \(\HQ\).\footnote{In particular, \(S\) is the \(\star\)-inverse to the identity \(\operatorname{Id}_{\HQ}\) and \(\one \circ \coone\) the \(\star\)-identity.}
\end{proof}

\enter

\begin{prop}[Coproduct identities for (divergent/restricted) combinatorial Green's functions\footnote{This is a direct generalization of \cite[Lemma 4.6]{Yeats_PhD}, \cite[Proposition 16]{vSuijlekom_QCD} and \cite[Theorem 1]{Borinsky_Feyngen} to super- and non-renormalizable theories, theories with several vertex residues and such with longitudinal and transversal degrees of freedom.}] \label{prop:coproduct_greensfunctions}
	Let \(\Q\) be a QFT, \(\HQ\) its (associated) renormalization Hopf algebra, \(r \in \RQ\) a residue and \(\mathbf{V} \in \ZvQ\) a vertex-grading multi-index. Then the following identities hold:
	{
	\allowdisplaybreaks
	\begin{subequations}
	\begin{align}
		\Delta \left ( \combgreen^r \right ) & = \sum_{\mathbf{v} \in \ZvQ} \overline{\combgreen}^r \overline{\combcharge}^\mathbf{v} \otimes \rescombgreen^r_{\mathbf{v}} \, , \label{eqn:coproduct_greensfunctions} \\
		\Delta \left ( \rescombgreen^r_{\mathbf{V}} \right ) & = \sum_{\mathbf{v} \in \ZvQ} \eval{\left ( \overline{\combgreen}^r \overline{\combcharge}^\mathbf{v} \right )}_{\mathbf{V} - \mathbf{v}} \otimes \rescombgreen^r_{\mathbf{v}} \, , \label{eqn:coproduct_greensfunctions_restricted}
		\intertext{and, provided that \(\Q\) is cograph-divergent and \(\overline{\rescombgreen}^r_{\mathbf{V}} \neq 0\),}
		\Delta \big ( \overline{\rescombgreen}^r_{\mathbf{V}} \big ) & = \sum_{\mathbf{v} \in \ZvQ} \eval{\left ( \overline{\combgreen}^r \overline{\combcharge}^\mathbf{v} \right )}_{\mathbf{V} - \mathbf{v}} \otimes \overline{\rescombgreen}^r_{\mathbf{v}} \, . \label{eqn:coproduct_greensfunctions_restricted_divergent}
	\end{align}
	\end{subequations}
	}%
\end{prop}

\begin{proof}
	We start with \eqnref{eqn:coproduct_greensfunctions} by using the linearity of the coproduct, \propref{prop:isomorphic_insertable_graph_sets} and \lemref{lem:sym-factors_and_ins-factors}:
	\begin{subequations}
	\begin{equation}
	\begin{split}
		\D{\rescombgreen^r} & = \sum_{\substack{\Gamma \in \GQ\\\res{\Gamma} = r}} \frac{1}{\sym{\Gamma}} \D{\Gamma}\\
		& = \sum_{\substack{\Gamma \in \GQ\\\res{\Gamma} = r}} \sum_{\gamma \in \DQ{\Gamma}} \frac{1}{\sym{\Gamma}} \gamma \otimes \Gamma / \gamma\\
		& = \sum_{\substack{\Gamma \in \GQ\\\res{\Gamma} = r}} \sum_{\gamma \in \DQ{\Gamma}} \frac{\insaut{\gamma}{\Gamma / \gamma}{\Gamma}}{\isoemb{\gamma}{\Gamma}} \left ( \frac{1}{\sym{\gamma}} \gamma \right ) \otimes \left ( \frac{1}{\sym{\Gamma / \gamma}} \Gamma / \gamma \right )\\
		& = \sum_{\substack{\Gamma^\prime \in \GQ\\\resnolim{\Gamma^\prime} = r}} \left ( \sum_{\substack{\gamma^\prime \in \IQnolim{\Gamma^\prime}}} \frac{\ins{\gamma^\prime}{\Gamma^\prime}}{\sym{\gamma^\prime}} \gamma^\prime \right ) \otimes \left ( \frac{1}{\sym{\Gamma^\prime}} \Gamma^\prime \right ) + \overline{\combgreen}^r \otimes \one\\
		& = \sum_{\mathbf{v} \in \ZvQ} \left ( \sum_{\substack{\gamma^\prime \in \IQrv}} \frac{\insrr{\gamma^\prime}}{\sym{\gamma^\prime}} \gamma^\prime \right ) \otimes \left ( \sum_{\substack{\Gamma^\prime \in \GQ\\\resnolim{\Gamma^\prime} = r\\\vtxgrdnolim{\Gamma^\prime} = \mathbf{v}}} \frac{1}{\sym{\Gamma^\prime}} \Gamma^\prime \right ) + \overline{\combgreen}^r \otimes \one\\
		& = \sum_{\mathbf{v} \in \ZvQ} \overline{\combgreen}^r \overline{\combcharge}^{\mathbf{v}} \otimes \rescombgreen^r_{\mathbf{v}}
	\end{split}
	\end{equation}
	From this we proceed to \eqnref{eqn:coproduct_greensfunctions_restricted} via restriction:
	\begin{equation}
	\begin{split}
		\D{\combgreen^r_\mathbf{V}} & = \eval{\left ( \D{\combgreen^r} \right )}_\mathbf{V}\\
		& = \eval{\left ( \sum_{\mathbf{v} \in \ZvQ} \overline{\combgreen}^r \overline{\combcharge}^{\mathbf{v}} \otimes \rescombgreen^r_{\mathbf{v}} \right )}_\mathbf{V}\\
		& = \sum_{\mathbf{v} \in \ZvQ} \eval{\left ( \overline{\combgreen}^r \overline{\combcharge}^\mathbf{v} \right )}_{\mathbf{V} - \mathbf{v}} \otimes \rescombgreen^r_{\mathbf{v}} \, .
	\end{split}
	\end{equation}
	\end{subequations}
	Finally, \eqnref{eqn:coproduct_greensfunctions_restricted_divergent} follows from \eqnref{eqn:coproduct_greensfunctions_restricted} together with the assumption of \(\Q\) being cograph-divergent and the application of \propref{prop:proj_div_graphs_coprod}.
\end{proof}

\enter

\begin{prop}[Coproduct identities for (divergent/restricted) combinatorial charges] \label{prop:coproduct_charges}
	Let \(\Q\) be a QFT, \(\HQ\) its (associated) renormalization Hopf algebra, \(v \in \RQO\) a vertex residue and \(\mathbf{V} \in \ZvQ\) a vertex-grading multi-index. Then the following identities hold:
	{
	\allowdisplaybreaks
	\begin{subequations}
	\begin{align}
		\Delta \left ( \combcharge^v \right ) & = \sum_{\mathbf{v} \in \ZvQ} \overline{\combcharge}^v \overline{\combcharge}^\mathbf{v} \otimes \combcharge^v_{\mathbf{v}} \, , \label{eqn:coproduct_charges} \\
		\Delta \left ( \combcharge^v_{\mathbf{V}} \right ) & = \sum_{\mathbf{v} \in \ZvQ} \eval{\left ( \overline{\combcharge}^v \overline{\combcharge}^\mathbf{v} \right )}_{\mathbf{V} - \mathbf{v}} \otimes \combcharge^v_{\mathbf{v}} \, , \label{eqn:coproduct_charges_restricted}
		\intertext{and, provided that \(\Q\) is cograph-divergent and \(\overline{\combcharge}^v_{\mathbf{V}} \neq 0\),}
		\Delta \big ( \overline{\combcharge}^v_{\mathbf{V}} \big ) & = \sum_{\mathbf{v} \in \ZvQ} \eval{\left ( \overline{\combcharge}^v \overline{\combcharge}^\mathbf{v} \right )}_{\mathbf{V} - \mathbf{v}} \otimes \overline{\combcharge}^v_{\mathbf{v}} \, . \label{eqn:coproduct_charges_restricted_divergent}
	\end{align}
	\end{subequations}
	}%
\end{prop}

\begin{proof}
	We start with \eqnref{eqn:coproduct_charges} by using the linearity and multiplicativity of the coproduct and \propref{prop:coproduct_greensfunctions}:
	\begin{subequations}
	\begin{equation}
	\begin{split}
		\D{\combcharge^v} & = \frac{\D{\combgreen^v}}{\prod_{e \in E \left ( v \right )} \sqrt{\D{\combgreen^e}}}\\
		& = \frac{\sum_{{\mathbf{v}_v} \in \ZvQ} \overline{\combgreen}^v \overline{\combcharge}^{\mathbf{v}_v} \otimes \rescombgreen^v_{{\mathbf{v}_v}}}{\prod_{e \in E \left ( v \right )} \sqrt{\sum_{{\mathbf{v}_e} \in \ZvQ} \overline{\combgreen}^e \overline{\combcharge}^{\mathbf{v}_e} \otimes \rescombgreen^e_{{\mathbf{v}_e}}}}\\
		& = \frac{\left ( \overline{\combgreen}^v \otimes \one \right ) \left ( \sum_{{\mathbf{v}_v} \in \ZvQ} \overline{\combcharge}^{\mathbf{v}_v} \otimes \rescombgreen^v_{{\mathbf{v}_v}} \right )}{\prod_{e \in E \left ( v \right )} \sqrt{\left ( \overline{\combgreen}^e \otimes \one \right ) \left ( \sum_{{\mathbf{v}_e} \in \ZvQ} \overline{\combcharge}^{\mathbf{v}_e} \otimes \rescombgreen^e_{{\mathbf{v}_e}} \right )}}\\
		& = \left ( \overline{\combcharge}^v \otimes \one \right ) \left ( \frac{\left ( \sum_{{\mathbf{v}_v} \in \ZvQ} \overline{\combcharge}^{\mathbf{v}_v} \otimes \rescombgreen^v_{{\mathbf{v}_v}} \right )}{\prod_{e \in E \left ( v \right )} \sqrt{\left ( \sum_{{\mathbf{v}_e} \in \ZvQ} \overline{\combcharge}^{\mathbf{v}_e} \otimes \rescombgreen^e_{{\mathbf{v}_e}} \right )}} \right )\\
		& = \left ( \overline{\combcharge}^v \otimes \one \right ) \left ( \sum_{\mathbf{v} \in \ZvQ} \overline{\combcharge}^\mathbf{v} \otimes \combcharge^v_{\mathbf{v}} \right )\\
		& = \sum_{\mathbf{v} \in \ZvQ} \overline{\combcharge}^v \overline{\combcharge}^\mathbf{v} \otimes \combcharge^v_{\mathbf{v}}
	\end{split}
	\end{equation}
	From this we proceed to \eqnref{eqn:coproduct_charges_restricted} via restriction:
	\begin{equation}
	\begin{split}
		\D{\combcharge^v_\mathbf{V}} & = \eval{\left ( \D{\combcharge^v} \right )}_\mathbf{V}\\
		& = \eval{\left ( \sum_{\mathbf{v} \in \ZvQ} \overline{\combcharge}^v \overline{\combcharge}^{\mathbf{v}} \otimes \combcharge^v_{\mathbf{v}} \right )}_\mathbf{V}\\
		& = \sum_{\mathbf{v} \in \ZvQ} \eval{\left ( \overline{\combcharge}^v \overline{\combcharge}^\mathbf{v} \right )}_{\mathbf{V} - \mathbf{v}} \otimes \combcharge^v_{\mathbf{v}} \, ,
	\end{split}
	\end{equation}
	\end{subequations}
	Finally, \eqnref{eqn:coproduct_charges_restricted_divergent} follows from \eqnref{eqn:coproduct_charges_restricted} together with the assumption of \(\Q\) being cograph-divergent and the application of \propref{prop:proj_div_graphs_coprod}.
\end{proof}

\enter

\begin{prop}[Coproduct identities for exponentiated (divergent/restricted) combinatorial charges] \label{prop:coproduct_exponentiated_combinatorial_charges}
	Let \(\Q\) be a QFT, \(\HQ\) its (associated) renormalization Hopf algebra, \(v \in \RQO\) a vertex residue, \(\mathbf{V} \in \ZvQ\) a vertex-grading multi-index and \(m \in \mathbb{Q}\) a rational number.\footnote{The power of an element in the renormalization Hopf algebra \(\mathfrak{G} \in \HQ\) via a non-natural number \(m \in \left ( \mathbb{Q} \setminus \mathbb{N}_0 \right )\) is understood via the formal binomial series, i.e.\
	\begin{equation}
		\mathfrak{G}^m := \sum_{n = 0}^\infty \binom{m}{n} \left ( \mathfrak{G} - \one \right )^n \, .
	\end{equation}
	More generally, if the renormalization Hopf algebra is considered over the field \(\field\), then the following statements are true for \(m \in \field\).} Then the following identities hold:
	{
	\allowdisplaybreaks
	\begin{subequations}
	\begin{align}
		\Delta \left ( \combcharge^{mv} \right ) & = \sum_{\mathbf{v} \in \ZvQ} \overline{\combcharge}^{mv} \overline{\combcharge}^\mathbf{v} \otimes \combcharge^{mv}_{\mathbf{v}} \, , \label{eqn:coproduct_powers_charges} \\
		\Delta \left ( \combcharge^{mv}_{\mathbf{V}} \right ) & = \sum_{\mathbf{v} \in \ZvQ} \eval{\left ( \overline{\combcharge}^{mv} \overline{\combcharge}^\mathbf{v} \right )}_{\mathbf{V} - \mathbf{v}} \otimes \combcharge^{mv}_{\mathbf{v}} \, , \label{eqn:coproduct_powers_charges_restricted}
		\intertext{and, provided that \(\Q\) is cograph-divergent and \(\overline{\combcharge}^{mv}_{\mathbf{V}} \neq 0\),}
		\Delta \big ( \overline{\combcharge}^{mv}_{\mathbf{V}} \big ) & = \sum_{\mathbf{v} \in \ZvQ} \eval{\left ( \overline{\combcharge}^{mv} \overline{\combcharge}^\mathbf{v} \right )}_{\mathbf{V} - \mathbf{v}} \otimes \overline{\combcharge}^{mv}_{\mathbf{v}} \, . \label{eqn:coproduct_powers_charges_restricted_divergent}
	\end{align}
	\end{subequations}
	}%
\end{prop}

\begin{proof}
	We start with \eqnref{eqn:coproduct_powers_charges} by using the linearity and multiplicativity of the coproduct and \propref{prop:coproduct_charges}:
	\begin{subequations}
	\begin{equation}
	\begin{split}
		\Delta \left ( \combcharge^{mv} \right ) & = \left ( \D{\combcharge^v} \right )^m\\
		& = \left ( \sum_{\mathbf{v} \in \ZvQ} \overline{\combcharge}^v \overline{\combcharge}^\mathbf{v} \otimes \combcharge^v_{\mathbf{v}} \right )^m\\
		& = \left ( \left ( \overline{\combcharge}^v \otimes \one \right ) \left ( \sum_{\mathbf{v} \in \ZvQ} \overline{\combcharge}^\mathbf{v} \otimes \combcharge^v_{\mathbf{v}} \right ) \right )^m\\
		& = \left ( \overline{\combcharge}^{mv} \otimes \one \right ) \left ( \sum_{\mathbf{v} \in \ZvQ} \overline{\combcharge}^\mathbf{v} \otimes \combcharge^v_{\mathbf{v}} \right )^m\\
		& = \sum_{\mathbf{v} \in \ZvQ} \overline{\combcharge}^{mv} \overline{\combcharge}^\mathbf{v} \otimes \combcharge^{mv}_{\mathbf{v}}
	\end{split}
	\end{equation}
	From this we proceed to \eqnref{eqn:coproduct_powers_charges_restricted} via restriction:
	\begin{equation}
	\begin{split}
		\Delta \left ( \combcharge^{mv}_{\mathbf{V}} \right ) & = \eval{\left ( \D{\combcharge^{mv}} \right )}_\mathbf{V}\\
		& = \eval{\left ( \sum_{\mathbf{v} \in \ZvQ} \overline{\combcharge}^{mv} \overline{\combcharge}^\mathbf{v} \otimes \combcharge^{mv}_{\mathbf{v}} \right )}_\mathbf{V}\\
		& \phantom{ = } \sum_{\mathbf{v} \in \ZvQ} \eval{\left ( \overline{\combcharge}^{mv} \overline{\combcharge}^\mathbf{v} \right )}_{\mathbf{V} - \mathbf{v}} \otimes \combcharge^{mv}_{\mathbf{v}}
	\end{split}
	\end{equation}
	\end{subequations}
	Finally, \eqnref{eqn:coproduct_powers_charges_restricted_divergent} follows from \eqnref{eqn:coproduct_powers_charges_restricted} together with the assumption of \(\Q\) being cograph-divergent and the application of \propref{prop:proj_div_graphs_coprod}.
\end{proof}

\section{Quantum gauge symmetries and subdivergences} \label{sec:quantum_gauge_symmetries_and_subdivergences}

In this section we give a precise definition of `quantum gauge symmetries (QGS)' in \defnref{defn:quantum_gauge_symmetries} and prove in \thmref{thm:quantum_gauge_symmetries_induce_hopf_ideals} that they induce Hopf ideals in the (associated) renormalization Hopf algebra, even for super- and non-renormalizable QFTs, QFTs with several coupling constants and QFTs with a transversal structure. This means that \(Z\)-factor identities coming from (generalized) gauge symmetries, such as \eqnsaref{eqns:z-factor_identities_qym}{eqns:z-factor_identities_qgr}, are compatible with the renormalization of subdivergences. Furthermore, we illustrate our framework with Quantum Yang--Mills theory in \exref{exmp:qym} and (effective) Quantum General Relativity in \exref{exmp:qgr}. Additionally, we mention the works \cite{Prinz_2,Borinsky_PhD}, which contain general results on Hopf ideals in the context of renormalization Hopf algebras.

\enter

\begin{defn}[Quantum gauge symmetries] \label{defn:quantum_gauge_symmetries}
	Let \(\Q\) be a QFT whose coupling-grading is superficially compatible, \(\mathbf{Q}_\Q\) the set of its combinatorial charges with cardinality \(\mathfrak{v}_\Q\) and \(\mathbf{q}_\Q\) the set of its physical charges with cardinality \(\mathfrak{q}_\Q\), cf.\ \defnref{defn:sets_of_coupling_constants}. Suppose that \(\mathfrak{v}_\Q > \mathfrak{q}_\Q\), then we define the following set of equivalence relations, to which we refer as `quantum gauge symmetries (QGS)', via
	\begin{equation} \label{eqn:equivalence_relation_combinatorial_charges}
		\left ( \overline{\combcharge}^v_\mathbf{C} \right )^m \sim \left ( \overline{\combcharge}^w_\mathbf{C} \right )^n \quad : \! \! \iff \quad \cpl{\combcharge^v}^m \equiv \cpl{\combcharge^w}^n
	\end{equation}
	for all \(v, w \in \RQ^{[0]}\), \(m, n \in \mathbb{N}_+\) and \(\mathbf{C} \in \mathbb{Z}^{\mathfrak{q}_\Q}\). Explicitly, the product of combinatorial charges is defined as the sum over all possibilities to connect their external edges to trees with the respective combinatorial charges as vertices. If \(\Q\) possesses additionally a transversal structure, cf.\ \defnref{defn:transversal_structure}, then the external edges are additionally required to respect the `physical' and `unphysical' labels. We remark that this is automatic for internal edges due to the restriction to particular coupling-gradings, as `unphysical' edges are indexed by the respective gauge fixing parameters. In particular, this requires connecting gauge field edges to be:
	\begin{itemize}
		\item Transversal, if they are related to higher valent gauge field vertices
		\item Longitudinal, if they are related to ghost edges
	\end{itemize}
	Finally, the set of all quantum gauge symmetries of \(\Q\) is denoted via \(\operatorname{QGS}_\Q\) and elements therein are given as quadruples of the form \(\set{v, m; w, n} \in \operatorname{QGS}_\Q\).
\end{defn}

\enter

\begin{exmp}[Quantum Yang--Mills theory] \label{exmp:qym}
	We continue the first example started in \sectionref{sec:introduction}: Consider Quantum Yang--Mills theory with a Lorenz gauge fixing. Then, the Lagrange density is given via
	\begin{equation}
	\begin{split}
		\mathcal{L}_\text{QYM} & := \mathcal{L}_\text{YM} + \mathcal{L}_\text{GF} + \mathcal{L}_\text{Ghost} \\ & \phantom{:} = \eta^{\mu \nu} \eta^{\rho \sigma} \delta_{a b} \left ( - \frac{1}{4 \mathrm{g}^2} F^a_{\mu \rho} F^b_{\nu \sigma} - \frac{1}{2 \xi} \big ( \partial_\mu A^a_\nu \big ) \big ( \partial_\rho A^b_\sigma \big ) \right ) \dif V_\eta \\
		& \phantom{:=} + \eta^{\mu \nu} \left ( \frac{1}{\xi} \overline{c}_a \left ( \partial_\mu \partial_\nu c^a \right ) + \mathrm{g} \tensor{f}{^a _b _c} \overline{c}_a \left ( \partial_\mu \big ( c^b A^c_\nu \big ) \right ) \right ) \dif V_\eta \, ,
	\end{split}
	\end{equation}
	where \(F^a_{\mu \nu} := \mathrm{g} \big ( \partial_\mu A^a_\nu - \partial_\nu A^a_\mu \big ) - \mathrm{g}^2 \tensor{f}{^a _b _c} A^b_\mu A^c_\nu\) is the local curvature form of the gauge boson \(A^a_\mu\). Furthermore, \(\dif V_\eta := \dif t \wedge \dif x \wedge \dif y \wedge \dif z\) denotes the Minkowskian volume form. Additionally, \(\eta^{\mu \nu} \partial_\mu A^a_\nu \equiv 0\) is the Lorenz gauge fixing functional and \(\xi\) denotes the gauge fixing parameter. Finally, \(c^a\) and \(\overline{c}_a\) are the gauge ghost and gauge antighost, respectively. Then, we have the following identities, where \(l\) of the unspecified external gauge boson legs are considered to be longitudinally projected:\footnote{The residues were drawn with JaxoDraw \cite{Binosi_Theussl,Binosi_Collins_Kaufhold_Theussl}.}
	\begin{subequations}
	\begin{equation}
		\cpl{\left ( \combcharge^{\scriptstyle{T} \tcgreen{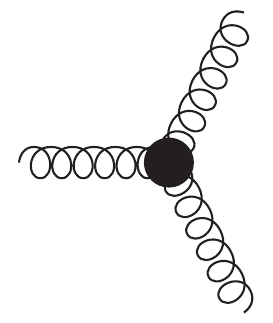}} \right )^{2_{\scriptscriptstyle{T}}}} = \xi^{\textfrac{l}{2}} \mathrm{g}^2 = \cpl{\combcharge^{\cgreen{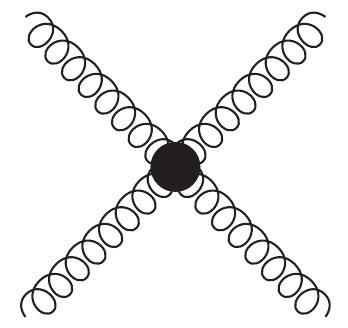}}} \, ,
	\end{equation}
	where \(T\) denotes the projection operator onto the transversal degree of freedom and \(2_{\scriptstyle{T}}\) the squaring by joining the two transversal half-edges, and
	\begin{equation}
		\cpl{\combcharge^{\tcgreen{c-gluontriple}^{\scriptstyle{L}}_{\scriptstyle{L}}}} = \xi^{1 + \textfrac{l}{2}} \mathrm{g} = \cpl{\combcharge^{\tcgreen{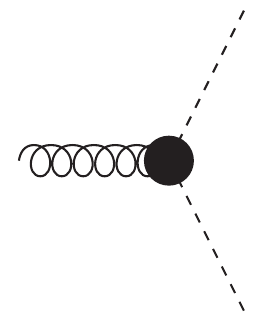}}} \, , \label{eqn:cpl_gluon_gluon-ghost}
	\end{equation}
	\end{subequations}
	where \(L\) denotes the projection operator onto the longitudinal degree of freedom, cf.\ \defnref{defn:transversal_structure}. More precisely, in the case of the Lorenz gauge fixing, we have
	\begin{subequations} \label{eqn:projection_tensors_qym}
	\begin{align}
		L^\nu_\mu & := \frac{1}{p^2} p^\nu p_\mu \, , \\
		I^\nu_\mu & := \delta^\nu_\mu
		\intertext{and}
		T^\nu_\mu & := I^\nu_\mu - L^\nu_\mu \, ,
	\end{align}
	\end{subequations}
	where a short calculation verifies \(L^2 = L\), \(I^2 = I\) and \(T^2 = T\). This implies:\footnote{We only display the generating set.}
	\begin{equation}
	\begin{split}
		\operatorname{QGS}_\text{QYM} & = \set{\set{{\scriptstyle{T} \tcgreen{c-gluontriple}}, 2_{\scriptstyle{T}}; {\cgreen{c-gluonquartic}}, 1} , \set{{\tcgreen{c-gluontriple}^{\scriptstyle{L}}_{\scriptstyle{L}}}, 1; {\tcgreen{c-gluonghosttriple}}, 1}}
	\end{split}
	\end{equation}
\end{exmp}

\enter

\begin{exmp}[(Effective) Quantum General Relativity] \label{exmp:qgr}
	We continue the second example started in \sectionref{sec:introduction}: Consider (effective) Quantum General Relativity with the metric decomposition \(g_{\mu \nu} \equiv \eta_{\mu \nu} + \varkappa h_{\mu \nu}\), where \(h_{\mu \nu}\) is the graviton field and \(\varkappa := \sqrt{\kappa}\) the graviton coupling constant (with \(\kappa := 8 \pi G\) the Einstein gravitational constant), and a linearized de Donder gauge fixing. Then, the Lagrange density is given via
	\begin{equation}
	\begin{split}
		\mathcal{L}_\text{QGR} & := \mathcal{L}_\text{GR} + \mathcal{L}_\text{GF} + \mathcal{L}_\text{Ghost} \\ & \phantom{:} = - \frac{1}{2 \varkappa^2} \left ( \sqrt{- \Det{g}} R + \frac{1}{2 \zeta} \eta^{\mu \nu} \deDonder^{(1)}_\mu \deDonder^{(1)}_\nu \right ) \dif V_\eta \\ & \phantom{:=} - \frac{1}{2} \eta^{\rho \sigma} \left ( \frac{1}{\zeta} \overline{C}^\mu \left ( \partial_\rho \partial_\sigma C_\mu \right ) + \overline{C}^\mu \left ( \partial_\mu \big ( \tensor{\Gamma}{^\nu _\rho _\sigma} C_\nu \big ) - 2 \partial_\rho \big ( \tensor{\Gamma}{^\nu _\mu _\sigma} C_\nu \big ) \right ) \right ) \dif V_\eta \, ,
	\end{split}
	\end{equation}
	where \(R := g^{\nu \sigma} \tensor{R}{^\mu _\nu _\mu _\sigma}\) is the Ricci scalar (with \(\tensor{R}{^\rho _\sigma _\mu _\nu} := \partial_\mu \tensor{\Gamma}{^\rho _\nu _\sigma} - \partial_\nu \tensor{\Gamma}{^\rho _\mu _\sigma} + \tensor{\Gamma}{^\rho _\mu _\lambda} \tensor{\Gamma}{^\lambda _\nu _\sigma} - \tensor{\Gamma}{^\rho _\nu _\lambda} \tensor{\Gamma}{^\lambda _\mu _\sigma}\) the Riemann tensor). Again, \(\dif V_\eta := \dif t \wedge \dif x \wedge \dif y \wedge \dif z\) denotes the Minkowskian volume form, which is related to the Riemannian volume form \(\dif V_g\) via \(\dif V_g \equiv \sqrt{- \Det{g}} \dif V_\eta\). Additionally, \(\deDonder^{(1)}_\mu := \eta^{\rho \sigma} \Gamma_{\mu \rho \sigma} \equiv 0\) is the linearized de Donder gauge fixing functional and \(\zeta\) the gauge fixing parameter. Finally, \(C_\mu\) and \(\overline{C}^\mu\) are the graviton-ghost and graviton-antighost, respectively. Again, we refer to \cite{Prinz_2,Prinz_4} for more detailed introductions and further comments on the chosen conventions. Then, we have the following identities, where \(k \in \mathbb{N}_0\) denotes additional graviton legs and \(l\) of the unspecified external graviton legs are considered to be longitudinally projected:\footnote{Again, the residues were drawn with JaxoDraw \cite{Binosi_Theussl,Binosi_Collins_Kaufhold_Theussl}.}
	\begin{subequations}
	\begin{equation}
		\cpl{\left ( \combcharge^{\bbsT \tcgreen{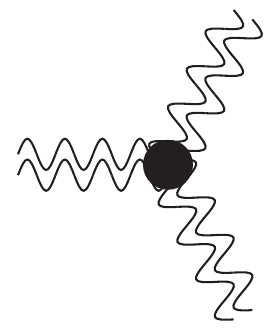}} \right ) \bullet_{\bbsT} \left ( \combcharge^{\bbsT \tcgreen{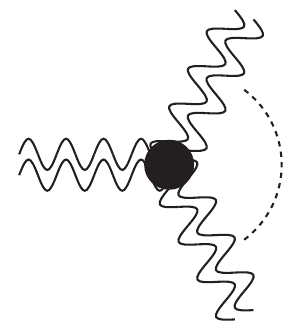} \legnumberexponent{k}} \right )} = \zeta^{\textfrac{l}{2}} \varkappa^{k+2} = \cpl{\combcharge^{\tcgreen{c-gravitonmultiple} \legnumberexponentlong{k+1}}}
	\end{equation}
	where \(\bbT\) denotes the projection operator onto the transversal degree of freedom and \(\bullet_{\bbsT}\) the product by joining the two transversal half-edges together, and
	\begin{equation}
		\cpl{\combcharge^{\tcgreen{c-gravitonmultiple}^{\bbsL}_{\bbsL} \legnumberexponentlongitudinal{k}}} = \zeta^{1 + \textfrac{l}{2}} \varkappa^{k+1} = \cpl{\combcharge^{\tcgreen{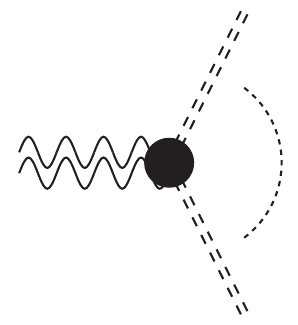} \legnumberexponentghost{k}}} \, , \label{eqn:cpl_graviton_graviton-ghost}
	\end{equation}
	\end{subequations}
	where \(\bbL\) denotes the projection operator onto the longitudinal degree of freedom, cf.\ \defnref{defn:transversal_structure}. More precisely, in the case of the linearized de Donder gauge fixing, we have\footnote{We remark that, unlike in the case of Quantum Yang--Mills theory with a Lorenz gauge fixing in Equations \eqref{eqn:projection_tensors_qym}, the projection tensors \(\bbL\) and \(\bbT\) are not symmetric with respect to the index-pairs \(\mu \nu\) and \(\rho \sigma\). The indices \(\mu \nu\) belong to vertex Feynman rules and the indices \(\rho \sigma\) belong to the propagator Feynman rule. This reflects their respective weights as tensor densities, which will be discussed and studied further in \cite{Prinz_8,Prinz_5,Prinz_6}.}
	\begin{subequations} \label{eqn:projection_tensors_qgr}
	\begin{align}
		\bbL^{\rho \sigma}_{\mu \nu} & := \frac{1}{2 p^2} \left ( \delta^\rho_\mu p^\sigma p_\nu + \delta^\sigma_\mu p^\rho p_\nu + \delta^\rho_\nu p^\sigma p_\mu + \delta^\sigma_\nu p^\rho p_\mu - 2 \eta^{\rho \sigma} p_\mu p_\nu \right ) \, , \\
		\bbI^{\rho \sigma}_{\mu \nu} & := \frac{1}{2} \left ( \delta^\rho_\mu \delta^\sigma_\nu + \delta^\sigma_\mu \delta^\rho_\nu \right )
		\intertext{and}
		\bbT^{\rho \sigma}_{\mu \nu} & := \bbI^{\rho \sigma}_{\mu \nu} - \bbL^{\rho \sigma}_{\mu \nu} \, ,
	\end{align}
	\end{subequations}
	where a short calculation verifies \(\bbL^2 = \bbL\), \(\bbI^2 = \bbI\) and \(\bbT^2 = \bbT\). This implies:\footnote{Again, we only display the generating set.}
	\begin{subequations}
	\begin{equation}
	\begin{split}
		\operatorname{QGS}_\text{QGR} & = \set{ \left . \set{{{\raisebox{0.1ex}{$\bbsT$}} \tcgreen{c-gravitonmultiple} \legnumber{i}},  \tilde{\jmath} \left ( i, j \right )_{\bbsT}; {{\raisebox{0.1ex}{$\bbsT$}} \tcgreen{c-gravitonmultiple} \legnumber{j}}, \tilde{\imath} \left ( i, j \right )_{\bbsT}} \, \right \vert \, i, j \in \mathbb{N}_0} \\ & \phantom{= \{} \bigcup \set{\left . \set{{\tcgreen{c-gravitonmultiple}^{\raisebox{0.1ex}{$\bbsL$}}_{\raisebox{0.1ex}{$\bbsL$}} \legnumberlongitudinal{k}}, 1; {\tcgreen{c-gravitonghostmultiple} \legnumberghost{k}}, 1} \, \right \vert \, k \in \mathbb{N}_0} \, ,
	\end{split}
	\end{equation}
	where \(i, j, k \in \mathbb{N}_0\) denote additional graviton legs, and with
	\begin{align}
		\tilde{\imath} \left ( i, j \right ) & := \frac{i+3}{\GCD{i+3}{j+3}}
		\intertext{and}
		\tilde{\jmath} \left ( i, j \right ) & := \frac{j+3}{\GCD{i+3}{j+3}} \, ,
	\end{align}
	\end{subequations}
	where \(\GCD{m}{n}\) denotes the greatest common divisor of the two natural numbers \(m, n \in \mathbb{N}_+\),\footnote{More precisely, let \(m = \prod_{k = 0}^\infty \left ( p_k \right )^{i_k}\) and \(n = \prod_{k = 0}^\infty  \left ( p_k \right )^{j_k}\) be the respective prime factorizations. Then, we have \begin{equation} \GCD{m}{n} := \prod_{k = 0}^\infty \left ( p_k \right )^{\operatorname{Min} \left ( i_k, j_k \right )} \, . \end{equation}} and finally \(\tilde{\imath} \left ( i, j \right )_{\bbsT}\) and \(\tilde{\jmath} \left ( i, j \right )_{\bbsT}\) denote, as exponents of residues, all possibilities to glue them to trees with transversal intermediate graviton edges. This will be studied further in \cite{Prinz_5,Prinz_6}, using the general theory and Feynman rules developed in \cite{Prinz_2,Prinz_4}.
\end{exmp}

\enter

\begin{defn}[Quantum gauge symmetry ideal] \label{defn:qgs_ideal}
	Given the situation of \defnref{defn:quantum_gauge_symmetries}, a cograph-divergent QFT \(\Q\) and its (associated) renormalization Hopf algebra. Then we define for each quantum gauge symmetry \(\set{v,m;w,n} \in \operatorname{QGS}_\Q\) the ideal
	\begin{subequations} \label{eqns:qgs_ideal}
	\begin{align}
		\iQ^{\set{v,m;w,n}} & := \sum_{\mathbf{C}^\prime \in \mathbb{Z}^{\mathfrak{q}_\Q}} \left \langle \overline{\combcharge}^{mv}_{\mathbf{C}^\prime} - \overline{\combcharge}^{nw}_{\mathbf{C}^\prime} \right \rangle_{\HQ}
		\intertext{and then sum over all such quantum gauge symmetries \(\set{v,m;w,n} \in \operatorname{QGS}_\Q\) to obtain the ideal}
		\iQ & := \sum_{\set{v,m;w,n} \in \operatorname{QGS}_\Q} \iQ^{\set{v,m;w,n}} \, .
	\end{align}
	\end{subequations}
\end{defn}

\enter

\begin{lem} \label{lem:elements_in_iq}
	Given the situation of \defnref{defn:qgs_ideal} and let
	\begin{equation} \label{eqn:coupling-coloring_function_multi-indices}
		\boldsymbol{\theta} \, : \quad \ZvQ \to \ZqQ \, , \quad \mathbf{V} \mapsto \mathbf{C}
	\end{equation}
	be the function mapping a vertex-grading multi-index to its corresponding coupling-grading multi-index with respect to the function \(\theta\) from Equation~(\ref{eqn:coupling-coloring_function}). Then we have \(\big ( \overline{\combcharge}^\mathbf{v}_\mathbf{C} - \overline{\combcharge}^\mathbf{w}_\mathbf{C} \big ) \in \iQ\) if and only if \(\boldsymbol{\theta} \left ( \mathbf{v} \right ) = \boldsymbol{\theta} \left ( \mathbf{w} \right )\) for two vertex-grading multi-indices \(\mathbf{v}, \mathbf{w} \in \ZvQ\).
\end{lem}

\begin{proof}
	Before we start with the actual proof, we emphasize that \(\theta \colon \AQ \to \qQ\) is the function mapping a vertex residue to its associated (product of) coupling constant(s) and thus \(\boldsymbol{\theta} \colon \ZvQ \to \ZqQ\) is the function mapping a multi-index of vertex residues to its associated multi-index of coupling constants. From the very definition of the quantum gauge symmetry ideal \(\iQ\) we know that \(\big ( \overline{\combcharge}^{mv}_\mathbf{C} - \overline{\combcharge}^{nw}_\mathbf{C} \big ) \in \iQ\) if and only if \(\set{v,m;w,n} \in \operatorname{QGS}_\Q\) is a quantum gauge symmetry. Thus we will now show that we can rewrite the element \(\big ( \overline{\combcharge}^\mathbf{v}_\mathbf{C} - \overline{\combcharge}^\mathbf{w}_\mathbf{C} \big )\) as a sum of the form
	\begin{equation} \label{eqn:to_show_qgs-equivalence}
		\left ( \overline{\combcharge}^\mathbf{v}_\mathbf{C} - \overline{\combcharge}^\mathbf{w}_\mathbf{C} \right ) = \sum_{\mathbf{C}^\prime \in \mathbb{Z}^{\mathfrak{q}_\Q}} \left ( \prod_{i = 1}^K \overline{\combcharge}^{m_i v_i}_{\mathbf{C}^\prime} - \prod_{j = 1}^K \overline{\combcharge}^{n_j w_j}_{\mathbf{C}^\prime} \right ) \mathfrak{H}_{\mathbf{C} - \mathbf{C}^\prime} \, ,
	\end{equation}
	where \(\set{v_k ,m_k ; w_k ,n_k} \in \operatorname{QGS}_\Q\) are quantum gauge symmetries for all \(k \in \set{1, \dots, K}\) and \(\mathfrak{H}_{\mathbf{C} - \mathbf{C}^\prime} \in \HQ\) are elements in the associated renormalization Hopf algebra. For the first direction, we assume the equality
	\begin{equation}
		\mathbf{c} := \boldsymbol{\theta} \left ( \mathbf{v} \right ) = \boldsymbol{\theta} \left ( \mathbf{w} \right ) \, .
	\end{equation}
	Then, we observe that due to the very definition of quantum gauge symmetries, cf.\ \defnref{defn:quantum_gauge_symmetries}, there exist quantum gauge symmetries \(\set{\set{v_k ,m_k ; w_k ,n_k}}_{k = 1}^K\) such that the two vertex-grading multi-indices are related via
	\begin{equation}
		\mathbf{r} := \mathbf{v} - \sum_{k = 1}^K m_k \mathbf{e}_{v_k} = \mathbf{w} - \sum_{k = 1}^K n_k \mathbf{e}_{w_k} \, ,
	\end{equation}
	where \(\mathbf{e}_v\) is the unit multi-index with respect to the vertex \(v\). Furthermore, since
	\begin{equation}
		\left \vert \mathbf{c} \right \vert := \sum_{i = 1}^{\ZqQ} \mathbf{c}_i < \infty \, ,
	\end{equation}
	we observe that this is always possible for a finite number of quantum gauge symmetries, i.e.\ \(K \in \mathbb{N}\). Additionally, given that \(\mathbf{v} = \mathbf{w}\) leads to the trivial case \(\big ( \overline{\combcharge}^\mathbf{v}_\mathbf{C} - \overline{\combcharge}^\mathbf{w}_\mathbf{C} \big ) = 0 \in \iQ\), we assume from now on \(\mathbf{v} \neq \mathbf{w}\), which implies \(K \in \mathbb{N}_+\). We now proceed with \eqnref{eqn:to_show_qgs-equivalence} by writing out the term in the brackets by setting \(\mathfrak{H}_{\mathbf{C} - \mathbf{C}^\prime}  := \overline{\combcharge}^\mathbf{r}_{\mathbf{C} - \mathbf{C}^\prime}\) and using the definition given in \eqnref{eqn:restricted_products_combinatorial_charges}:
	\begin{equation}
	\begin{split}
		\left ( \prod_{i = 1}^K \overline{\combcharge}^{m_i v_i}_{\mathbf{C}^\prime} \right . & - \left . \prod_{j = 1}^K \overline{\combcharge}^{n_j w_j}_{\mathbf{C}^\prime} \right ) = \eval{\left ( \prod_{i = 1}^K \left ( \combcharge^{v_i} \right )^{m_i} - \prod_{j = 1}^K \left ( \combcharge^{w_j} \right )^{n_j} \right )}_{\mathbf{C}^\prime} \\
		 = & \eval{\left ( \! \left ( \prod_{i = 1}^{K-1} \left ( \combcharge^{v_i} \right )^{m_i} \right ) \! \Big ( \! \left ( \combcharge^{v_K} \right )^{m_K} - \left ( \combcharge^{w_K} \right )^{n_K} \! \Big ) \! \right )}_{\mathbf{C}^\prime} \\
		 & + \eval{\left ( \! \left ( \prod_{i = 1}^{K-2} \left ( \combcharge^{v_i} \right )^{m_i} \right ) \! \Big ( \! \left ( \combcharge^{v_{K-1}} \right )^{m_{K-1}} - \left ( \combcharge^{w_{K-1}} \right )^{n_{K-1}} \! \Big ) \! \left ( \combcharge^{w_K} \right )^{n_K} \right )}_{\mathbf{C}^\prime} \\
		 & + \eval{\left ( \! \left ( \prod_{i = 1}^{K-3} \left ( \combcharge^{v_i} \right )^{m_i} \right ) \! \Big ( \! \left ( \combcharge^{v_{K-2}} \right )^{m_{K-2}} - \left ( \combcharge^{w_{K-2}} \right )^{n_{K-2}} \! \Big ) \! \left ( \prod_{j = K-1}^{K} \left ( \combcharge^{w_j} \right )^{n_j} \right ) \! \right )}_{\mathbf{C}^\prime} \\
		 & + \dots \\
		 & + \eval{\left ( \! \left ( \prod_{i = 1}^2 \left ( \combcharge^{v_i} \right )^{m_i} \right ) \! \Big ( \! \left ( \combcharge^{v_3} \right )^{m_3} - \left ( \combcharge^{w_3} \right )^{n_3} \! \Big ) \! \left ( \prod_{j = 4}^{K} \left ( \combcharge^{w_j} \right )^{n_j} \right ) \! \right )}_{\mathbf{C}^\prime} \\
		 & + \eval{\left ( \! \left ( \combcharge^{v_1} \right )^{m_1} \! \Big ( \! \left ( \combcharge^{v_2} \right )^{m_2} - \left ( \combcharge^{w_2} \right )^{n_2} \! \Big ) \! \left ( \prod_{j = 3}^{K} \left ( \combcharge^{w_j} \right )^{n_j} \right ) \! \right )}_{\mathbf{C}^\prime} \\
		 & + \eval{\left ( \! \Big ( \! \left ( \combcharge^{v_1} \right )^{m_1} - \left ( \combcharge^{w_1} \right )^{n_1} \! \Big ) \! \left ( \prod_{j = 2}^{K} \left ( \combcharge^{w_j} \right )^{n_j} \right ) \! \right )}_{\mathbf{C}^\prime}
	\end{split}
	\end{equation}
	From this it is straightforward to see that each summand lies in the quantum gauge symmetry ideal \(\iQ\) and furthermore that the intermediate terms cancel pairwise. This directly implies the implication \(\big ( \overline{\combcharge}^\mathbf{v}_\mathbf{C} - \overline{\combcharge}^\mathbf{w}_\mathbf{C} \big ) \in \iQ\) if \(\boldsymbol{\theta} \left ( \mathbf{v} \right ) = \boldsymbol{\theta} \left ( \mathbf{w} \right )\). Conversely, given that \(\big ( \overline{\combcharge}^\mathbf{v}_\mathbf{C} - \overline{\combcharge}^\mathbf{w}_\mathbf{C} \big ) \in \iQ\), we obtain a decomposition as in \eqnref{eqn:to_show_qgs-equivalence} by the very definition of the quantum gauge symmetry ideal, cf.\ \defnref{defn:qgs_ideal}, which directly implies the equality \(\boldsymbol{\theta} \left ( \mathbf{v} \right ) = \boldsymbol{\theta} \left ( \mathbf{w} \right )\). This shows the claimed equivalence and thus concludes the proof.
\end{proof}

\enter

\begin{thm}[Quantum gauge symmetries induce Hopf ideals\footnote{This is a direct generalization of \cite[Theorem 15]{vSuijlekom_BV} to super- and non-renormalizable theories, theories with several coupling constants and such with longitudinal and transversal degrees of freedom.}] \label{thm:quantum_gauge_symmetries_induce_hopf_ideals}
	Given the situation of \defnref{defn:qgs_ideal}, the ideal \(\iQ\) is a Hopf ideal, i.e.\ satisfies:
	\begin{enumerate}
		\item \(\Delta \big ( \iQ \big ) \subseteq \HQ \otimes \iQ + \iQ \otimes \HQ\)
		\item \(\coone \big ( \iQ \big ) = 0\)
		\item \(S \big ( \iQ \big ) \subseteq \iQ\)
	\end{enumerate}
\end{thm}

\begin{proof}
	We start with the following calculation, using \propref{prop:coproduct_exponentiated_combinatorial_charges}:
	\begin{equation}
	\begin{split}
		\D{\overline{\combcharge}^{mv}_{\mathbf{V} - m \mathbf{e}_v} - \overline{\combcharge}^{nw}_{\mathbf{V} - n \mathbf{e}_w}} & = \sum_{\mathbf{v}^\prime \in \ZvQ} \left ( \eval{\left ( \overline{\combcharge}^{mv} \overline{\combcharge}^{\mathbf{v}^\prime} \right )}_{\mathbf{V} - \mathbf{v}^\prime - m \mathbf{e}_v} \otimes \overline{\combcharge}^{mv}_{\mathbf{v}^\prime} \right . \\ & \phantom{= \sum_{\mathbf{v}^\prime \in \ZvQ} (} \left . - \eval{\left ( \overline{\combcharge}^{nw} \overline{\combcharge}^{\mathbf{v}^\prime} \right )}_{\mathbf{V} - \mathbf{v}^\prime - n \mathbf{e}_w} \otimes \overline{\combcharge}^{nw}_{\mathbf{v}^\prime} \right ) \\
		& = \sum_{\mathbf{v}^\prime \in \ZvQ} \Big ( \overline{\combcharge}^{\mathbf{v}^\prime + m \mathbf{e}_v}_{\mathbf{V} - \mathbf{v}^\prime - m \mathbf{e}_v} \otimes \overline{\combcharge}^{mv}_{\mathbf{v}^\prime} - \overline{\combcharge}^{\mathbf{v}^\prime + n \mathbf{e}_w}_{\mathbf{V} - \mathbf{v}^\prime - n \mathbf{e}_w} \otimes \overline{\combcharge}^{nw}_{\mathbf{v}^\prime} \Big ) \\
		& = \sum_{\mathbf{v} \in \ZvQ} \left ( \overline{\combcharge}^\mathbf{v}_{\mathbf{V} - \mathbf{v}} \otimes \overline{\combcharge}^{mv}_{\mathbf{v} - m \mathbf{e}_v} - \overline{\combcharge}^\mathbf{v}_{\mathbf{V} - \mathbf{v}} \otimes \overline{\combcharge}^{nw}_{\mathbf{v} - n \mathbf{e}_w} \right ) \\
		& = \sum_{\mathbf{v} \in \ZvQ} \overline{\combcharge}^\mathbf{v}_{\mathbf{V} - \mathbf{v}} \otimes \left ( \overline{\combcharge}^{mv}_{\mathbf{v} - m \mathbf{e}_v} - \overline{\combcharge}^{nw}_{\mathbf{v} - n \mathbf{e}_w} \right )
	\end{split}
	\end{equation}
	Thus, when summing over all vertex-gradings \(\mathbf{v} \in \ZvQ\) that contribute to a particular coupling-grading \(\mathbf{c} \in \ZqQ\), briefly denoted via \(\mathbf{v} \in \boldsymbol{\theta}^{-1} \left ( \mathbf{c} \right )\) with \(\boldsymbol{\theta} \colon \ZvQ \to \ZqQ\) the function defined in \eqnref{eqn:coupling-coloring_function_multi-indices}, we obtain:
	\begin{equation} \label{eqn:coproduct_qgs}
	\begin{split}
		\D{\overline{\combcharge}^{mv}_{\mathbf{C} - m \mathbf{e}_{\theta \left ( v \right )}} - \overline{\combcharge}^{nw}_{\mathbf{C} - n \mathbf{e}_{\theta \left ( w \right )}}} & = \sum_{\mathbf{v} \in \ZvQ} \overline{\combcharge}^\mathbf{v}_{\mathbf{C} - \boldsymbol{\theta} \left ( \mathbf{v} \right )} \otimes \left ( \overline{\combcharge}^{mv}_{\mathbf{v} - m \mathbf{e}_v} - \overline{\combcharge}^{nw}_{\mathbf{v} - n \mathbf{e}_w} \right ) \\
		& = \sum_{\mathbf{c} \in \ZqQ} \sum_{\mathbf{v} \in \boldsymbol{\theta}^{-1} \left ( \mathbf{c} \right )} \overline{\combcharge}^\mathbf{v}_{\mathbf{C} - \mathbf{c}} \otimes \left ( \overline{\combcharge}^{mv}_{\mathbf{v} - m \mathbf{e}_v} - \overline{\combcharge}^{nw}_{\mathbf{v} - n \mathbf{e}_w} \right )
	\end{split}
	\end{equation}
	We proceed by introducing the following equivalence relation: Given two elements \(\mathfrak{G}, \mathfrak{H} \in \HQ\), we set
	\begin{subequations} \label{eqns:qgs_equivalence_relations}
	\begin{align}
		\mathfrak{G} \sim \mathfrak{H} \quad & : \iff \quad \exists \, \mathfrak{I} \subset \iQ \; \text{ such that } \; \mathfrak{G} = \mathfrak{H} + \mathfrak{I} \, , \label{eqn:first_equivalence}
	\intertext{which, on the level of restricted products of combinatorial charges, is due to \lemref{lem:elements_in_iq} equivalent to the following equivalence relation}
		\overline{\combcharge}^\mathbf{v}_{\mathbf{C}^\prime} \sim \overline{\combcharge}^\mathbf{w}_{\mathbf{C}^\prime} \quad & : \iff \quad \boldsymbol{\theta} \left ( \mathbf{v} \right ) = \boldsymbol{\theta} \left ( \mathbf{w} \right ) \, . \label{eqn:second_equivalence}
	\end{align}
	\end{subequations}
	Coming back to \eqnref{eqn:coproduct_qgs}, we now want to implement the equivalence relation of \eqnsref{eqns:qgs_equivalence_relations} on the left-hand side of the tensor product by adding terms in \(\iQ \otimes \HQ\): This implies that we can define equivalence classes of restricted combinatorial charges where the exponent is now a coupling-grading multi-index \(\mathbf{c} \in \ZqQ\), which we denote via \(\big [ \overline{\combcharge}^\mathbf{c}_{\mathbf{C} - \mathbf{c}} \big ]\). More precisely, let \(\mathbf{v}_1, \dots,  \mathbf{v}_L \in \ZvQ\) be all vertex-grading multi-indices with \(\boldsymbol{\theta} \left ( \mathbf{v}_l \right ) = \mathbf{c}\) for \(l \in \set{1, \dots, L}\) and \(L \in \mathbb{N}\). Then, the second sum in the second line of \eqnref{eqn:coproduct_qgs} reads
	\begin{equation}
		\sum_{\mathbf{v} \in \boldsymbol{\theta}^{-1} \left ( \mathbf{c} \right )} \overline{\combcharge}^\mathbf{v}_{\mathbf{C} - \mathbf{c}} \otimes \left ( \overline{\combcharge}^{mv}_{\mathbf{v} - m \mathbf{e}_v} - \overline{\combcharge}^{nw}_{\mathbf{v} - n \mathbf{e}_w} \right ) = \sum_{l = 1}^L \overline{\combcharge}^{\mathbf{v}_l}_{\mathbf{C} - \mathbf{c}} \otimes \left ( \overline{\combcharge}^{mv}_{\mathbf{v}_l - m \mathbf{e}_v} - \overline{\combcharge}^{nw}_{\mathbf{v}_l - n \mathbf{e}_w} \right )
	\end{equation}
	and the combinatorial charges \(\overline{\combcharge}^{\mathbf{v}_l}\) can be identified, modulo the addition of terms in \(\iQ \otimes \HQ\), to their equivalence class \(\big [ \overline{\combcharge}^\mathbf{c}_{\mathbf{C} - \mathbf{c}} \big ]\). Combining these results, we finally obtain:
	\begin{equation}
	\begin{split}
		\D{\overline{\combcharge}^{mv}_{\mathbf{C} - m \mathbf{e}_{\theta \left ( v \right )}} - \overline{\combcharge}^{nw}_{\mathbf{C} - n \mathbf{e}_{\theta \left ( w \right )}}} & \simeq_{\iQ \otimes \HQ} \sum_{\mathbf{c} \in \ZqQ} \left [ \overline{\combcharge}^\mathbf{c}_{\mathbf{C} - \mathbf{c}} \right ] \otimes \left ( \overline{\combcharge}^{mv}_{\mathbf{c} - m \mathbf{e}_{\theta \left ( v \right )}} - \overline{\combcharge}^{nw}_{\mathbf{c} - n \mathbf{e}_{\theta \left ( w \right )}} \right ) \\
		& \subseteq \HQ \otimes \iQ + \iQ \otimes \HQ \, ,
	\end{split}
	\end{equation}
	where \(\simeq_{\iQ \otimes \HQ}\) denotes equality modulo the addition of elements in \(\iQ \otimes \HQ\). Additionally, we remark the equality \(\theta \left ( v \right )^m \equiv \theta \left ( w \right )^n\) for quantum gauge symmetries \(\set{v,m;w,n} \in \operatorname{QGS}_\Q\), which relates the calculations in this proof to the definition of the quantum gauge symmetry ideal, cf.\ \defnref{defn:qgs_ideal}, by setting
	\begin{equation}
		\mathbf{C}^\prime := \mathbf{C} - m \mathbf{e}_{\theta \left ( v \right )} \equiv \mathbf{C} - n \mathbf{e}_{\theta \left ( w \right )} \, .
	\end{equation}
	This shows condition 1. Condition 2 follows immediately, as \(\one \notin \iQ\), i.e.\ \(\iQ \subset \operatorname{Aug} \left ( \HQ \right )\). Finally, condition 3 follows from \lemref{lem:coproduct_and_antipode_identities} together with condition 1, which finishes the proof.
\end{proof}

\enter

\begin{rem}
	\thmref{thm:quantum_gauge_symmetries_induce_hopf_ideals} describes the most general situation, as it includes also super- and non-renormalizable QFTs, QFTs with several coupling constants and QFTs with a transversal structure. Therefore it can be applied to (effective) Quantum General Relativity in the sense of \cite{Kreimer_QG1}, possibly coupled to matter from the Standard Model \cite{Romao_Silva}, cf.\ e.g.\ \cite{Prinz_2,Prinz_4}. Slightly less general results in this direction can be found in \cite{Kreimer_Anatomy,vSuijlekom_QED,vSuijlekom_QCD,vSuijlekom_BV,Kreimer_vSuijlekom,Kreimer_QG1,Kreimer_Core}, some of them using the language of Hochschild cohomology.\footnote{We also mention the relevant conference proceedings \cite{Kreimer_QG2,vSuijlekom_GF,vSuijlekom_BRST,vSuijlekom_Combinatorics,vSuijlekom_Ha-pQGT}.}
\end{rem}

\enter

\begin{col} \label{col:equivalence_vtx-grd_cpl-grd}
	Given the situation of \thmref{thm:quantum_gauge_symmetries_induce_hopf_ideals}, the vertex-grading and coupling-grading are equivalent if either \(\mathfrak{v}_\Q = \mathfrak{q}_\Q\) with \(\boldsymbol{\theta}\) bijective, or in the quotient Hopf algebra \(\HQ / \iQ\). In the latter case, the ideal \(\iQ\) is the smallest Hopf ideal with this property.
\end{col}

\begin{proof}
	This follows directly from the definition of the ideal \(\iQ\) in \defnref{defn:qgs_ideal}, which is designed such that in the quotient \(\HQ / \iQ\) all Feynman graphs are identified that contribute to restricted (powers of) combinatorial charges which are associated with (powers) of the same physical coupling constant.
\end{proof}

\enter

\begin{rem} \label{rem:quotient_vertex-grading_coupling-grading}
	The Hopf ideal \(\iQ\) from \thmref{thm:quantum_gauge_symmetries_induce_hopf_ideals} is defined such that in the quotient Hopf algebra \(\HQ / \iQ\) the coproduct and antipode identities from \sectionref{sec:coproduct_and_antipode_identities}, which are valid for vertex-grading, also hold for coupling-grading, cf.\ \colref{col:hopf_subalgebras_coupling-grading}. Thus it is possible to combine the \(Z\)-factors for the set \(\QQ\) to \(Z\)-factors for the set \(\qQ\), if the criteria from \thmref{thm:criterion_ren-hopf-mod} or \colref{col:qgs_and_rfr} are satisfied.
\end{rem}

\section{Quantum gauge symmetries and renormalized Feynman rules} \label{sec:quantum_gauge_symmetries_and_renormalized_feynman_rules}

Having established that `quantum gauge symmetries (QGS)' are compatible with the treatment of subdivergences in \thmref{thm:quantum_gauge_symmetries_induce_hopf_ideals}, we now turn our attention to their relation with renormalized Feynman rules. We start this section with the definition of the gauge theory renormalization Hopf module: Here we implement the quantum gauge symmetries only on the left-hand side of the tensor product of the coproduct, i.e.\ only on the superficially divergent subgraphs. As such, it is the weakest requirement for renormalized Feynman rules to possess quantum gauge symmetries. More precisely, in this setting the relations are only implemented on the \(Z\)-factors, i.e.\ \(\mathscr{R}\)-divergent contributions of the Feynman rules. In \thmref{thm:criterion_ren-hopf-mod} we provide criteria for this compatibility of quantum gauge symmetries with the unrenormalized Feynman rules and the chosen renormalization scheme. Then we show in \colref{col:qgs_and_rfr} that under mild assumptions on the unrenormalized Feynman rules this statement is independent of the chosen renormalization scheme. In particular, this result states that in this case we can implement the quantum gauge symmetries directly on the renormalization Hopf algebra by taking the quotient with respect to the quantum gauge symmetry Hopf ideal. Finally, we remark that for theories with a transversal structure these mild assumptions correspond precisely to the respective cancellation identities. Thus, combining these results, this shows the well-definedness of the Corolla polynomial without reference to a particular renormalization scheme, cf.\ \remref{rem:corolla_polynomial}.

\enter

\begin{defn}[Gauge theory renormalization Hopf module, \cite{Kissler_PhD}] \label{defn:renormalization_hopf_module}
	Let \(\Q\) be a cograph-divergent QFT with quantum gauge symmetries, i.e.\ \(\operatorname{QGS}_\Q \neq \emptyset\), \(\HQ\) its (associated) renormalization Hopf algebra and \(\iQ\) the  corresponding quantum gauge symmetry Hopf ideal. Let
	\begin{equation} \label{eqn:qgs_projection_map}
		\pi_\Q : \, \HQ \surject \HQ / \iQ
	\end{equation}
	denote the projection map. We consider \(\HQ\) as a left Hopf module over \(\HQ / \iQ\) with the usual Hopf structures as in \defnref{defn:renormalization_hopf_algebra}. The interesting map is the comodule map, defined via
	\begin{equation}
		\delta \, : \quad \HQ \to \left ( \HQ / \iQ \right ) \otimes \HQ \, , \quad \Gamma \mapsto \big ( \pi_\Q \otimes \id_{\HQ} \! \big ) \circ \D{\Gamma} \, .
	\end{equation}
	Then we define the renormalized Feynman rules \(\Phi_\mathscr{R}\) using the comodule map \(\delta\) instead of the coproduct \(\Delta\), i.e.\ defining the counterterm map \(\countertermsymbol\) on the quotient \(\HQ / \iQ\).
\end{defn}

\enter

\begin{col} \label{col:hopf_subalgebras_coupling-grading}
	Given the situation of \defnref{defn:renormalization_hopf_module}, we have
	\begin{subequations}
	\begin{align}
		\delta \left ( \combgreen^r \right ) & = \sum_{\mathbf{c} \in \mathbb{Z}^{\mathfrak{q}_\Q}} \left [ \overline{\combgreen}^r \overline{\combcharge}^\mathbf{c} \right ] \otimes \rescombgreen^r_{\mathbf{c}} \, , \\
		\delta \left ( \rescombgreen^r_{\mathbf{C}} \right ) & = \sum_{\mathbf{c} \in \mathbb{Z}^{\mathfrak{q}_\Q}} \eval{\left [ \overline{\combgreen}^r \overline{\combcharge}^\mathbf{c} \right ]}_{\mathbf{C} - \mathbf{c}} \otimes \rescombgreen^r_{\mathbf{c}} \, ,
		\intertext{and, provided that \(\overline{\rescombgreen}^r_{\mathbf{C}} \neq 0\),}
		\delta \left ( \overline{\rescombgreen}^r_{\mathbf{C}} \right ) & = \sum_{\mathbf{c} \in \mathbb{Z}^{\mathfrak{q}_\Q}} \eval{\left [ \overline{\combgreen}^r \overline{\combcharge}^\mathbf{c} \right ]}_{\mathbf{C} - \mathbf{c}} \otimes \overline{\rescombgreen}^r_{\mathbf{c}} \, ,
	\end{align}
	\end{subequations}
	where the equivalence classes on the left-hand side of the tensor product are with respect to the equivalence relation of Equations~(\ref{eqns:qgs_equivalence_relations}), i.e.\ modulo the addition of elements in \(\iQ\). Analogous results also hold in the cases of \propsaref{prop:coproduct_charges}{prop:coproduct_exponentiated_combinatorial_charges}.
\end{col}

\begin{proof}
	This follows directly from \propref{prop:coproduct_greensfunctions} together with \colref{col:equivalence_vtx-grd_cpl-grd}.
\end{proof}

\enter

\begin{rem}
	\colref{col:hopf_subalgebras_coupling-grading} states that the gauge theory renormalization Hopf module from \defnref{defn:renormalization_hopf_module} possesses Hopf subalgebras in the sense of \defnref{defn:hopf_subalgebras_renormalization_hopf_algebra} for coupling-grading, cf.\ \remref{rem:hopf_subalgebras_renormalization_hopf_algebra}. This implies that the subdivergence structure of QFTs is compatible with quantum gauge symmetries in the sense of \defnref{defn:quantum_gauge_symmetries}. Furthermore, it is obvious by construction that this is the weakest requirement for the compatibility of quantum gauge symmetries with multiplicative renormalization, cf.\ \colref{col:equivalence_vtx-grd_cpl-grd}. Their validity on the level of renormalized Feynman rules, i.e.\ the existence and well-definedness of the maps
	\begin{subequations}
	\begin{align}
	\widetilde{\countertermsymbol} & := \countertermsymbol \circ \left ( \pi_\Q \right )^{-1} \, : \quad \HQ / \iQ \to \EQ
	\intertext{and}
	\widetilde{\renFR} & := \renFR \circ \left ( \pi_\Q \right )^{-1} \, : \quad \HQ / \iQ \to \EQ \, ,
	\end{align}
	\end{subequations}
	where \(\left ( \pi_\Q \right )^{-1}\) is any right inverse to the projection map \(\pi_\Q\) from \eqnref{eqn:qgs_projection_map}, with respect to the following commuting diagrams
	\begin{equation}
	\begin{tikzcd}[row sep=huge]
		\HQ \arrow[swap]{d}{\pi_\Q} \arrow{r}{\countertermsymbol} & \EQ \\
		\HQ / \iQ \arrow[dashed, swap]{ur}{\widetilde{\countertermsymbol}} &
	\end{tikzcd}
	\qquad \text{and} \qquad
	\begin{tikzcd}[row sep=huge]
		\HQ \arrow[swap]{d}{\pi_\Q} \arrow{r}{\renFR} & \EQ \\
		\HQ / \iQ \arrow[dashed, swap]{ur}{\widetilde{\renFR}} &
	\end{tikzcd} \, ,
	\end{equation}
	are then studied in the following \lemref{lem:criterion_ren-hopf-mod}, \thmref{thm:criterion_ren-hopf-mod} and \colref{col:qgs_and_rfr}. Moreover, given a quantum gauge theory with a transversal structure, cf.\ \defnref{defn:transversal_structure}, we stress the following additional compatibility issue: Recall the setup of \defnref{defn:residue_amplitude_and_coupling_constant_set} where we represented each particle type by at least two edges to disentangle their physical and unphysical degrees of freedom, cf.\ \remref{rem:longitudinal_and_transversal_gauge_fields}. Then the Feynman rules are required to be compatible with this decomposition as follows: The divergent Feynman graphs and their residues need to behave similar with respect to these physical and unphysical projections. This ensures that the contraction of subdivergences is a well-defined operation and thus is a necessary condition to construct the renormalization Hopf algebra, cf.\ \cite[Subsection 3.3]{Prinz_2}. We will study this further in \cite{Prinz_8}, cf.\ \cite{Prinz_5,Prinz_6}, using cancellation identities and Feynman graph cohomology.
\end{rem}

\enter

\begin{lem} \label{lem:criterion_ren-hopf-mod}
	The gauge theory renormalization Hopf module from \defnref{defn:renormalization_hopf_module} is compatible with renormalized Feynman rules if \(\, \iQ \in \operatorname{Ker} \big ( \countertermsymbol \big )\). More precisely, if for all \(\set{v, m; w, n} \in \operatorname{QGS}_\Q\) and all \(\mathbf{C}^\prime \in \ZqQ\) we have
	\begin{equation}
		\counterterm{\overline{\combcharge}^{mv}_{\mathbf{C}^\prime}} = \counterterm{\overline{\combcharge}^{nw}_{\mathbf{C}^\prime}} \, . \label{eqn:well-definedness_counterterm-map}
	\end{equation}
\end{lem}

\begin{proof}
	This statement is equivalent to the well-definedness of the counterterm-map on the equivalence classes of the QGS-equivalence relation: Indeed, we have
	\begin{equation}
		\overline{\combcharge}^{mv}_{\mathbf{C}^\prime} \simeq_{\operatorname{QGS}_\Q} \overline{\combcharge}^{nw}_{\mathbf{C}^\prime} \, ,
	\end{equation}
	and thus \eqnref{eqn:well-definedness_counterterm-map} ensures that the counterterm-map can be unambiguously defined on the corresponding equivalence classes.
\end{proof}

\enter

\begin{thm}[Quantum gauge symmetries and renormalized Feynman rules] \label{thm:criterion_ren-hopf-mod}
	Given the situation of \lemref{lem:criterion_ren-hopf-mod} and a proper renormalization scheme \(\mathscr{R}\), then \eqnref{eqn:well-definedness_counterterm-map} is equivalent to one of the following identities for each \(\mathbf{c} \in \ZqQ\):
	\begin{enumerate}
		\item \(\big [ \overline{\combcharge}^\mathbf{c}_{\mathbf{C} - \mathbf{c}} \big ] = \left [ 0 \right ]\)
		\item \(\overline{\combcharge}^{mv}_{\mathbf{c} - m \mathbf{e}_{\theta \left ( v \right )}} = \overline{\combcharge}^{nw}_{\mathbf{c} - n \mathbf{e}_{\theta \left ( w \right )}} = 0\)
		\item \(\renscheme{\left ( \Phi \circ \mathscr{A} \right ) \big ( \overline{\combcharge}^{mv}_{\mathbf{c} - m \mathbf{e}_{\theta \left ( v \right )}} - \overline{\combcharge}^{nw}_{\mathbf{c} - n \mathbf{e}_{\theta \left ( w \right )}} \big )} = 0\)
	\end{enumerate}
\end{thm}

\begin{proof}
	Using \eqnref{eqn:coproduct_qgs} from the proof of \thmref{thm:quantum_gauge_symmetries_induce_hopf_ideals}, we obtain
	\begin{equation}
	\begin{split}
		\countertermsymbol \, \Big ( \overline{\combcharge}^{mv}_{\mathbf{C} - m \mathbf{e}_{\theta \left ( v \right )}} & - \overline{\combcharge}^{nw}_{\mathbf{C} - n \mathbf{e}_{\theta \left ( w \right )}} \Big ) = \\ & \phantom{=} - \sum_{\mathbf{c} \in \ZqQ} \renscheme{\counterterm{\left [ \overline{\combcharge}^\mathbf{c}_{\mathbf{C} - \mathbf{c}} \right ]} \FRP{\overline{\combcharge}^{mv}_{\mathbf{c} - m \mathbf{e}_{\theta \left ( v \right )}} - \overline{\combcharge}^{nw}_{\mathbf{c} - n \mathbf{e}_{\theta \left ( w \right )}}}} \, ,
	\end{split}
	\end{equation}
	which vanishes, if for all \(\mathbf{c} \in \ZqQ\)
	\begin{equation}
		\renscheme{\counterterm{\left [ \overline{\combcharge}^\mathbf{c}_{\mathbf{C} - \mathbf{c}} \right ]} \FRP{\overline{\combcharge}^{mv}_{\mathbf{c} - m \mathbf{e}_{\theta \left ( v \right )}} - \overline{\combcharge}^{nw}_{\mathbf{c} - n \mathbf{e}_{\theta \left ( w \right )}}}} = 0 \, . \label{eqn:qgs_criterion_fr_rs}
	\end{equation}
	Since the coupling-grading is required to be superficially compatible, cf.\ \defnref{defn:superficially_compatible_grading} and \defnref{defn:quantum_gauge_symmetries}, we have either
	\begin{equation}
		\left [ \overline{\combcharge}^\mathbf{c}_{\mathbf{C} - \mathbf{c}} \right ] = \left [ 0 \right ] \qquad \text{or} \qquad \left [ \overline{\combcharge}^\mathbf{c}_{\mathbf{C} - \mathbf{c}} \right ] = \Big [ \combcharge^\mathbf{c}_{\mathbf{C} - \mathbf{c}} \Big ]
	\end{equation}
	and either
	\begin{equation}
		\overline{\combcharge}^{mv}_{\mathbf{c}^\prime} = \overline{\combcharge}^{nw}_{\mathbf{c}^\prime} = 0 \qquad \text{or} \qquad \overline{\combcharge}^{mv}_{\mathbf{c}^\prime} = \combcharge^{mv}_{\mathbf{c}^\prime} \! \quad \text{and} \quad \overline{\combcharge}^{nw}_{\mathbf{c}^\prime} = \combcharge^{nw}_{\mathbf{c}^\prime} \, ,
	\end{equation}
	where \(\mathbf{c}^\prime := \mathbf{c} - m \mathbf{e}_{\theta \left ( v \right )} \equiv \mathbf{c} - n \mathbf{e}_{\theta \left ( w \right )}\). We proceed by noting that \(\mathscr{R}\) is a linear map, whose kernel and cokernel consists only of convergent formal Feynman integral expressions, as it is required to be proper, cf.\ \defnref{defn:proper_renormalization_scheme}. This directly implies, by the recursive structure of the counterterm, that \(\operatorname{Ker} \big ( \countertermsymbol \big )\) consists only of convergent formal Feynman integral expressions. Thus, \eqnref{eqn:qgs_criterion_fr_rs} is equivalent to one of the following identities for each \(\mathbf{c} \in \ZqQ\):
	\begin{enumerate}
		\item \(\big [ \overline{\combcharge}^\mathbf{c}_{\mathbf{C} - \mathbf{c}} \big ] = \left [ 0 \right ]\)
		\item \(\overline{\combcharge}^{mv}_{\mathbf{c} - m \mathbf{e}_{\theta \left ( v \right )}} = \overline{\combcharge}^{nw}_{\mathbf{c} - n \mathbf{e}_{\theta \left ( w \right )}} = 0\)
		\item \(\renscheme{\left ( \Phi \circ \mathscr{A} \right ) \big ( \overline{\combcharge}^{mv}_{\mathbf{c} - m \mathbf{e}_{\theta \left ( v \right )}} - \overline{\combcharge}^{nw}_{\mathbf{c} - n \mathbf{e}_{\theta \left ( w \right )}} \big )} = 0\)
	\end{enumerate}
	This is the claimed statement and thus finishes the proof.
\end{proof}

\enter

\begin{rem}
	Given the situation of \lemref{lem:criterion_ren-hopf-mod} and \thmref{thm:criterion_ren-hopf-mod}. We note that \eqnref{eqn:well-definedness_counterterm-map} and condition 3 are criteria for both, the unrenormalized Feynman rules \(\Phi\) and the renormalization scheme \(\mathscr{R}\). More precisely, \eqnref{eqn:well-definedness_counterterm-map} states that a common counterterm can be chosen for residues that are related via quantum gauge symmetries. Furthermore, condition 3 states that the corresponding \(\mathscr{R}\)-divergent contributions from restricted combinatorial charges coincide. In contrast, conditions 1 and 2 are solely criteria for the unrenormalized Feynman rules \(\Phi\).
\end{rem}

\enter

\begin{col}[Quantum gauge symmetries and renormalized Feynman rules] \label{col:qgs_and_rfr}
	The quotient Hopf algebra \(\HQ / \iQ\) is compatible with renormalized Feynman rules if \(\, \iQ \in \operatorname{Ker} \left ( \Phi_\mathscr{R} \right )\). More precisely, if for all \(\set{v, m; w, n} \in \operatorname{QGS}_\Q\) and all \(\mathbf{C}^\prime \in \ZqQ\) we have
	\begin{equation}
		\RFR{\overline{\combcharge}^{mv}_{\mathbf{C}^\prime}} = \RFR{\overline{\combcharge}^{nw}_{\mathbf{C}^\prime}} \, . \label{eqn:well-definedness_renormalized-fr}
	\end{equation}
	If the renormalization scheme \(\mathscr{R}\) is proper then this is equivalent to one of the following identities for each \(\mathbf{c} \in \ZqQ\):
	\begin{enumerate}
		\item \(\big [ \overline{\combcharge}^\mathbf{c}_{\mathbf{C} - \mathbf{c}} \big ] = \left [ 0 \right ]\)
		\item \(\overline{\combcharge}^{mv}_{\mathbf{c} - m \mathbf{e}_{\theta \left ( v \right )}} = \overline{\combcharge}^{nw}_{\mathbf{c} - n \mathbf{e}_{\theta \left ( w \right )}} = 0\)
		\item \(\Phi \big ( \overline{\combcharge}^{mv}_{\mathbf{c} - m \mathbf{e}_{\theta \left ( v \right )}} - \overline{\combcharge}^{nw}_{\mathbf{c} - n \mathbf{e}_{\theta \left ( w \right )}} \big ) = 0\)
	\end{enumerate}
\end{col}

\begin{proof}
	Again, using \eqnref{eqn:coproduct_qgs} from the proof of \thmref{thm:quantum_gauge_symmetries_induce_hopf_ideals} and the same reasoning as in the proof of \thmref{thm:criterion_ren-hopf-mod}, we obtain
	\begin{equation}
		\RFR{\overline{\combcharge}^{mv}_{\mathbf{C} - m \mathbf{e}_{\theta \left ( v \right )}} - \overline{\combcharge}^{nw}_{\mathbf{C} - n \mathbf{e}_{\theta \left ( w \right )}}} = \sum_{\mathbf{c} \in \ZqQ} \counterterm{\left [ \overline{\combcharge}^\mathbf{c}_{\mathbf{C} - \mathbf{c}} \right ]} \FR{\overline{\combcharge}^{mv}_{\mathbf{c} - m \mathbf{e}_{\theta \left ( v \right )}} - \overline{\combcharge}^{nw}_{\mathbf{c} - n \mathbf{e}_{\theta \left ( w \right )}}} \, ,
	\end{equation}
	which vanishes, if for all \(\mathbf{c} \in \ZqQ\)
	\begin{equation}
		\counterterm{\left [ \overline{\combcharge}^\mathbf{c}_{\mathbf{C} - \mathbf{c}} \right ]} \FR{\overline{\combcharge}^{mv}_{\mathbf{c} - m \mathbf{e}_{\theta \left ( v \right )}} - \overline{\combcharge}^{nw}_{\mathbf{c} - n \mathbf{e}_{\theta \left ( w \right )}}} = 0 \, . \label{eqn:qgs_criterion_rfr_fr}
	\end{equation}
	Once more, using the same reasoning as in the proof of \thmref{thm:criterion_ren-hopf-mod}, we obtain that \eqnref{eqn:qgs_criterion_rfr_fr} is equivalent to one of the following identities for each \(\mathbf{c} \in \ZqQ\):
	\begin{enumerate}
		\item \(\big [ \overline{\combcharge}^\mathbf{c}_{\mathbf{C} - \mathbf{c}} \big ] = \left [ 0 \right ]\)
		\item \(\overline{\combcharge}^{mv}_{\mathbf{c} - m \mathbf{e}_{\theta \left ( v \right )}} = \overline{\combcharge}^{nw}_{\mathbf{c} - n \mathbf{e}_{\theta \left ( w \right )}} = 0\)
		\item \(\Phi \big ( \overline{\combcharge}^{mv}_{\mathbf{c} - m \mathbf{e}_{\theta \left ( v \right )}} - \overline{\combcharge}^{nw}_{\mathbf{c} - n \mathbf{e}_{\theta \left ( w \right )}} \big ) = 0\)
	\end{enumerate}
	This is the claimed statement and thus finishes the proof.
\end{proof}

\enter

\begin{rem} \label{rem:corolla_polynomial}
	Given the situation of \colref{col:qgs_and_rfr} and a proper renormalization scheme \(\mathscr{R}\), cf.\ \defnref{defn:proper_renormalization_scheme}. We note that while \eqnref{eqn:well-definedness_renormalized-fr} is a criterion for both, the unrenormalized Feynman rules \(\Phi\) and the renormalization scheme \(\mathscr{R}\), conditions 1 to 3 are solely criteria for the unrenormalized Feynman rules \(\Phi\). More precisely, \eqnref{eqn:well-definedness_renormalized-fr} states that the values of renormalized Feynman rules are equivalent for residues that are related via quantum gauge symmetries. In contrast, conditions 1 to 3 state that the corresponding unrenormalized Feynman rules \(\Phi\) coincide on the restricted combinatorial charges. We remark, however, that if we consider a quantum gauge theory with a transversal structure, cf.\ \defnref{defn:transversal_structure}, then the vertex-residues of the combinatorial charges as well as their coupling-gradings includes the `physical' and `unphysical' labels. Thus, the mentioned criteria need only to hold for all such restrictions individually. This is directly related to cancellation identities \cite{tHooft_Veltman,Citanovic,Sars_PhD,Kissler_Kreimer,Gracey_Kissler_Kreimer,Kissler}, which are graphical versions of (generalized) Ward--Takahashi and Slavnov--Taylor identities \cite{Ward,Takahashi,Taylor,Slavnov}. More precisely, they indicate the behavior of unrenormalized (tree) Feynman diagrams with respect to longitudinal and transversal projections. Thus, given that they hold for Quantum Yang--Mills theory, \colref{col:qgs_and_rfr} implies the well-definedness of the Corolla polynomial without reference to a particular renormalization scheme \(\mathscr{R}\). The Corolla polynomial is a graph polynomial that relates amplitudes in Quantum Yang--Mills theory to amplitudes in \(\phi^3_4\)-Theory \cite{Kreimer_Yeats,Kreimer_Sars_vSuijlekom,Kreimer_Corolla}.  More precisely, this graph polynomial is based on half-edges and is used for the construction of the so-called Corolla differential that relates the corresponding parametric Feynman integral expressions \cite{Kreimer_Sars_vSuijlekom,Sars_PhD,Golz_PhD}. Thereby, the corresponding cancellation-identities are implicitly encoded into a double complex of Feynman graphs, called Feynman graph cohomology \cite{Kreimer_Sars_vSuijlekom,Berghoff_Knispel}. This double-complex can be interpreted as a perturbative version of BRST cohomology \cite{Becchi_Rouet_Stora_1,Becchi_Rouet_Stora_2,Becchi_Rouet_Stora_3,Tyutin}, where the precise relation will be studied in future work \cite{Prinz_8,Prinz_5,Prinz_6}. We remark that this construction has been successfully generalized to Quantum Yang--Mills theories with spontaneous symmetry breaking \cite{Prinz_1} and Quantum Electrodynamics with spinors \cite{Golz_1,Golz_2,Golz_3}. The possibility to reformulate (effective) Quantum General Relativity in this framework will be also studied in future work.
\end{rem}

\enter

\begin{rem}
	It is possible to endow the character group of the renormalization Hopf algebra \(\HQ\) with a manifold structure such that it becomes a regular Lie group in the sense of Milnor, cf.\ \cite{Milnor}. Then, in this setting, the character group on the quotient Hopf algebra \(\HQ / \iQ\) is a closed Lie subgroup thereof \cite{Bogfjellmo_Dahmen_Schmeding_1,Dahmen_Schmeding,Bogfjellmo_Schmeding,Bogfjellmo_Dahmen_Schmeding_2}.
\end{rem}

\section{Conclusion} \label{sec:conclusion}

We have studied the renormalization of gauge theories and gravity in the Hopf algebra setup of Connes and Kreimer. The main results are \thmref{thm:quantum_gauge_symmetries_induce_hopf_ideals} and \thmref{thm:criterion_ren-hopf-mod}, where the first states that quantum gauge symmetries correspond to Hopf ideals in the renormalization Hopf algebra and the second provides criteria for their validity on the level of renormalized Feynman rules. To this end, we studied combinatorial properties of the superficial degree of divergence in \sectionref{sec:a_superficial_argument} and generalized known coproduct and antipode identities to the super- and non-renormalizable cases in \sectionref{sec:coproduct_and_antipode_identities}. Additionally, we extended this framework to theories with multiple vertex residues and coupling constants in \defnref{defn:connectedness_gradings_renormalization_hopf_algebra} and discussed the incorporation of transversal structures in \remref{rem:longitudinal_and_transversal_gauge_fields}. Then we illustrated the developed theory in the cases of Quantum Yang--Mills theory in \exref{exmp:qym} and (effective) Quantum General Relativity in \exref{exmp:qgr}. Finally, we discussed, as a direct consequence of our findings, the well-definedness of the Corolla polynomial without reference to a particular renormalization scheme in \remref{rem:corolla_polynomial}, cf.\ \cite{Kreimer_Yeats,Kreimer_Sars_vSuijlekom,Kreimer_Corolla}. In future work, we want to study the incorporation of cancellation identities via Feynman graph cohomology into the present setup \cite{Prinz_8}. Additionally, we investigate further on the case of (effective) Quantum General Relativity, possibly coupled to matter from the Standard Model, as was started in \cite{Prinz_2,Prinz_4}. In particular, this includes the BRST double complex of diffeomorphisms and gauge transformations \cite{Prinz_5} and the corresponding longitudinal and transversal projection operators together with cancellation identities \cite{Prinz_6}.

\section*{Acknowledgments}
\addcontentsline{toc}{section}{Acknowledgments}

The author thanks Dirk Kreimer, Henry Ki{\ss}ler and Ren\'{e} Klausen for illuminating and helpful discussions. Additionally, the author thanks the anonymous referee for his careful proofreading and valuable feedback, which has lead to various improvements. This research is supported by the International Max Planck Research School for Mathematical and Physical Aspects of Gravitation, Cosmology and Quantum Field Theory.

\bibliography{References}{}
\bibliographystyle{babunsrt}

\end{document}